\documentclass[11pt]{article}

\usepackage{style}
\usepackage{shortcuts}

\definecolor{amethyst}{rgb}{0.6, 0.4, 0.8}

\title{A Truly Subcubic Combinatorial Algorithm \\ for Induced $4$-Cycle Detection
}

\author{
Amir Abboud\thanks{Weizmann Institute of Science, \url{amir.abboud@gmail.com}. }
\and 
Shyan Akmal\orcidlink{0000-0002-7266-2041}\thanks{Max Planck Institute for Informatics, \url{shyan.akmal@gmail.com}. }
\and
Nick Fischer\orcidlink{0009-0001-0909-3296}\thanks{Max Planck Institute for Informatics, \url{nfischer@mpi-inf.mpg.de}.  }
}
\date{\vspace{-3ex}}

\begin{document}

\maketitle

\vspace{-4ex}
\begin{abstract}
\noindent
In this paper, we present the first truly subcubic, combinatorial algorithm for detecting an induced $4$-cycle in a graph. The running time is $O(n^{2.84})$ on $n$-node graphs, thus separating the task of detecting induced $4$-cycles from detecting triangles, which requires $n^{3-o(1)}$ time combinatorially under the popular Boolean Matrix Multiplication hypothesis.

Significant work has gone into characterizing the 
exact time complexity of induced subgraph detection, relative to the complexity of detecting cliques of various sizes. Prior work identified the question of whether induced $4$-cycle detection is triangle-hard as the only remaining case towards completing the lowest level of the classification, dubbing it a \emph{curious case} [\textcolor{ForestGreen}{Dalirrooyfard, Vassilevska W., FOCS 2022}]. Our result can be seen as a negative resolution of this question.

Our algorithm deviates from previous techniques in the large body of subgraph detection algorithms and employs 
the trendy topic of graph decomposition that has hitherto been restricted to more global problems (as in the use of expander decompositions for flow problems) or to shaving subpolynomial factors (as in the application of graph regularity lemmas). While our algorithm is slower than the (non-combinatorial) state-of-the-art $\tilde{O}(n^{\omega})$ time algorithm based on polynomial identity testing [\textcolor{ForestGreen}{Vassilevska W., Wang, Williams, Yu, SODA 2014}], combinatorial advancements often come with other benefits. In particular, we give the first nontrivial \emph{deterministic} algorithm for detecting induced $4$-cycles.
\end{abstract}

\vspace{1cm}

\paragraph*{Acknowledgements}
We thank Virginia Vassilevska Williams for answering our questions about existing lower bounds for subgraph detection, and Jakob Nogler for pointing out writing mistakes in an earlier version of this paper.
This work is partially funded by the Ministry of Education and Science of Bulgaria
(support for INSAIT, part of the Bulgarian National Roadmap for Research Infrastructure). This work is part of the project CONJEXITY that has received funding from the European Research Council (ERC) under the European Union's Horizon Europe research and innovation programme (grant agreement No.~101078482).
This work received funding from the Klaus Tschira Boost Fund, a joint
initiative of GSO Guidance, Skills \& Opportunities e.V. and the Klaus Tschira Stiftung.

\pagenumbering{gobble}
\newpage
\tableofcontents
\newpage
\pagenumbering{arabic}

\section{Introduction}
\label{sec:intro}

Detecting small patterns in large graphs arises naturally 
across the sciences.
This task 
 plays a particularly important role in multiple fields of theoretical computer science, where it has been explored from every imaginable angle, making it one of the most extensively studied problems in algorithmic graph theory.
In this paper, we focus on one of its most basic forms.
For a fixed pattern graph $H$, in the \emph{induced $H$ detection} problem we are given a host graph $G$
on~$n$ vertices, 
and are tasked with determining if $G$ contains  $H$ as an induced\footnote{An induced subgraph is obtained by taking a subset of nodes and all edges among them.
Thus, e.g., the $4$-cycle is a non-induced subgraph of the $4$-clique, but it is not an induced subgraph.} subgraph.

Significant work \cite{corneil1985linear,nevsetvril1985complexity,Olariu1988,KloksKratschMuller2000,eisenbrand2004complexity,KowalukLingasLundell2013,VassilevskaWangWilliamsYu2014,floderus2015induced,FloderusKLL15,BlaserKomarathSreenivasaiah2018,DalirrooyfardVassilevskaVuong2019,DalirrooyfardVassilevska2022} has gone into characterizing the 
exact time complexity of induced $H$ detection based on 
the structural properties of the patterns $H$,
with researchers attempting to order the complexity of these problems into a hierarchy relative to the complexity of detecting cliques of various sizes.
A pattern $H$ is
placed in the $k^{\text{th}}$ level of the hierarchy 
if solving induced $H$ detection has the same
time complexity 
as detecting a $k$-clique in an $n$-node graph. 
This is done by (1) presenting an algorithm running in the same time as $k$-clique detection, and (2) showing a reduction from $k$-clique. 

Ignoring subpolynomial factors, the longstanding upper bounds for detecting $k$-cliques in $n$-node graphs
 are $O(n^{k})$ using ``combinatorial'' algorithms (by brute force) 
 and 
 $O(n^{\omega k/3})$ in general~\cite{nevsetvril1985complexity},\footnote{To be precise, the upper bound $O(n^{\omega k/3
 })$ only applies when $k$ is divisible by $3$. 
 The best-known upper bound in general is $O(n^{\omega(\lfloor k/3\rfloor, \lceil k/3\rceil, \lceil(k-1)/3\rceil)})$~\cite{eisenbrand2004complexity}, where $\omega(\cdot, \cdot, \cdot)$ denotes the exponent of rectangular matrix multiplication.}
 where  $\omega \le 2.3713$
  is the exponent of matrix multiplication.
  In particular, the bounds for triangle detection, which correspond to the first nontrivial level of the hierarchy ($k=3$), are $O(n^3)$ combinatorially and $O(n^{\omega})$ in general.
    The term ``combinatorial'' intuitively refers to algorithms that avoid the use of fast matrix multiplication; we refer the reader to \cite[Section 1.1]{AbboudFischerKelleyLovettMeka2024} for an extensive discussion on the motivations behind seeking such algorithms and for operational definitions of this concept.
Under the popular $k$-Clique conjecture from fine-grained complexity \cite{AbboudBackursVassilevska2015}, these longstanding bounds cannot be improved by a polynomial factor. Moreover, a truly subcubic combinatorial algorithm for triangle detection refutes the central Boolean Matrix Multiplication (BMM) conjecture \cite{VassilevskaWilliams2018}.
Thus under these conjectures, placing all patterns $H$ into this hierarchy would characterize the time complexity of induced $H$ detection both for combinatorial algorithms and in general.

 In this paper, we aim to complete the lowest level of the hierarchy, which 
 corresponds to 
 classifying which patterns can be solved in $O(n^2)$ time and which are as hard as triangle detection. 
 That is, we would like to know which patterns can be solved in linear time (in the input size) and which cannot.
As has been shown in prior work (and explained below), all that remains is \emph{the curious case of $4$-cycle} asking the simple but infamous open question (see e.g., \cite{Eschen2011,DalirrooyfardVassilevska2022}):

    \begin{center}
        \textit{Is detecting induced 4-cycles as hard as detecting triangles?}
    \end{center}

Let us now explain why $4$-cycle is the only unclassified case.
    First, since our interest is only in the complexity in terms of $n$ (disregarding the number of edges in $G$), we may assume that detecting an induced copy of $H$ is equivalent to detecting an induced copy of its complement $\bar{H}$, just by complementing the input graph.
    Second, it is known that if $H$ contains a triangle
    or its complement as an induced subgraph, the induced $H$ detection problem is at least as hard as triangle detection \cite[Theorem 1.1]{DalirrooyfardVassilevskaVuong2019}.
    Therefore, any pattern that contains a $K_3$ or $\bar{K_3}$ is already classified as being in the third level of the hierarchy (or higher).
    By classic bounds on Ramsey numbers (see e.g., \cite[Theorem 1.4]{LiLin2022}), this immediately classifies all patterns $H$ with at least six vertices as triangle hard.
    By inspection, all pattern graphs $H$ on three to five vertices also contain $K_3$ or $\bar{K_3}$ as a subgraph, except for the $5$-cycle 
    $C_5$, the $4$-cycle $C_4$, the $3$-edge-path~$P_3$, and the $2$-edge-path $P_2$ (as well as their complements).
    The latter two patterns fall into the lowest level of the hierarchy: a graph free of induced $P_2$ must be a disjoint union of cliques and can be recognized with a connected components algorithm, while a graph free of $P_3$ is called a cograph and can be recognized by a simple linear-time algorithm as well \cite{corneil1985linear}.
    This leaves us with~$C_5$ and~$C_4$.
    A folklore reduction shows that induced 5-cycle detection is triangle hard \cite[Section 1.1]{DahlgaardKnudsenStockel2017}.
    Designing such a reduction for 4-cycles would answer the above question positively.

    On the algorithms front, one can always reduce induced $H$ detection to $k$-clique detection where $k$ is the number of nodes in $H$, meaning that the $4$-cycle is in the $4^{\text{th}}$ level of the hierarchy or below.
    In a well-known paper~\cite{VassilevskaWangWilliamsYu2014}, Vassilevska Williams, Wang, Williams, and Yu proved that the 4-cycle is in the $3^{\text{rd}}$ level of the hierarchy: it can be detected in $\tilde{O}(n^3)$ time combinatorially and in $\tilde{O}(n^\omega)$ time via fast matrix multiplication (in fact, they prove this holds for \emph{all} four-node patterns except for the clique and independent set on four vertices). Beating this triangle-detection runtime (either the combinatorial or the general one) would answer the above question negatively.

    Up to this work, the prevailing intuition has been that induced $4$-cycle detection is a hard problem, leading to an obsession with the search for a reduction from triangle detection (see e.g.,\ early attempts at partial progress in this area \cite[Corollary 5]{floderus2015induced}, as well as more recent results trying to shed light on what such a reduction should look like \cite[Theorem 2.4]{DalirrooyfardVassilevska2022}).
    Let us mention some of the reasons for this.
    First, unlike triangle, the simple algorithm for induced $4$-cycle  has time complexity corresponding to the time complexity of $4$-clique; even a sub-quartic combinatorial algorithm requires the heavy machinery of pattern polynomials and polynomial identity testing (which may not even be called combinatorial under a more conservative definition). In fact, as we discuss below, the deterministic time complexity of induced $4$-cycle from prior work is \emph{super-cubic}.
    Another natural setting in which $4$-cycle appears to be much harder than triangle is that of counting: by a known reduction, counting the number of induced $4$-cycles is not only triangle-hard, but is even as hard as counting \emph{$4$-cliques} \cite{KloksKratschMuller2000}.
    Still, the only success at proving a lower bound under fine-grained complexity hypotheses is a recent proof by Dalirrooyfard and Vassilevska Williams that on graphs with $\Theta(n^{3/2})$ edges, detecting induced 4-cycles requires at least $n^{2-o(1)}$ time \cite[Theorem 2.4]{DalirrooyfardVassilevska2022}.
    This demonstrates that induced 4-cycle detection is unlikely to be solvable in linear time (in the number of edges), but does not tell us whether the existing triangle detection runtime bound for this problem is tight or not.
    Notably, proving such a conditional lower bound for triangle detection (in sparse graphs) itself is a big open question, giving yet another setting in which $4$-cycle appears harder. 
    In fact, as a function of the number of edges $m$ the best bound for induced $4$-cycle detection is \smash{$\tilde{O}(m^{\frac{4\omega-1}{2\omega+1}})$} time~\cite[Corollary 4.1]{VassilevskaWangWilliamsYu2014}, which is higher than the best~\smash{$O(m^{\frac{2\omega}{\omega+1}})$}~time bound \cite{AlonYusterZwick1997} known for triangle detection.

    Meanwhile, perhaps the only reason to think that $4$-cycle is easier than triangle comes from the analogy with the \emph{non-induced} case, where detecting a (non-induced) $4$-cycle has a classic $O(n^2)$ time combinatorial algorithm.
    At some level, the easiness of non-induced $4$-cycle comes from the fact that a dense graph cannot be $C_4$-free.
    While this no longer applies for induced $4$-cycles, it is still true that dense induced $C_4$-free graphs possess a lot of structure.
    What this structure is and how we might capitalize on it algorithmically is far from obvious; such a technique may have far-reaching consequences. 
    One should, of course, be careful with such analogies, since the induced case is often much harder than the non-induced case (e.g.,\ non-induced $k$-path \mbox{is in $2^{k}\poly(n)$ time} \cite{williams2009finding}, while induced $k$-path requires $n^{\Omega(k)}$ time under fine-grained hypotheses \cite{DalirrooyfardVassilevska2022}).

\newpage

Let us take a step back and motivate the above main question from a different perspective. 
Fine-grained complexity aims to reveal how complexity arises by identifying the atomic computational tasks that cannot be solved in linear time.
Triangle detection, the task of finding three objects that are in a pairwise relationship with each other, has been clearly established as such an atomic hard problem. 
Induced 4-cycle detection has a different flavor, asking for four objects satisfying two kinds of constraints (edges and non-edges), and has hitherto seemed to embody a hard task that is not explainable by a reduction from triangle. In particular, from the perspective of combinatorial algorithms, it is one of the simplest problems that cannot be solved combinatorially in subcubic time, and yet is not known to be triangle-hard.
Does this mean that a new conjecture is due, highlighting another atomic hard problem?

\paragraph{Our Results.}

The main result of this paper is a new algorithm for induced 4-cycle detection that breaks the cubic barrier with a \emph{combinatorial} algorithm.
Under the BMM hypothesis, this gives a counterintuitive separation between induced 4-cycle detection and triangle detection.

    \begin{restatable}{theorem}{deterministic}
        \label{thm:deterministic}
        There is a deterministic, combinatorial algorithm solving induced 4-cycle detection on graphs with $n$ vertices in $\tilde{O}(n^{3-1/6}) \le O(n^{2.84})$ time.
    \end{restatable}

From the perspective of combinatorial algorithms, our result answers the above question negatively, showing that induced 4-cycle detection is strictly easier than triangle detection and hence does not belong in the $3^{\text{rd}}$ level of the hierarchy.
What this means depends on whether its complexity will end up improving all the way down to $\tilde{O}(n^2)$ or not.
If it does, and we find this likely, then it simply means that induced $4$-cycle is an easy pattern that belongs to the lowest level of the hierarchy.
If, on the other hand, one discovers a super-quadratic lower bound, it would leave induced $4$-cycle as an intermediate problem in the hierarchy and show that a hierarchy based on the relationship to $k$-clique cannot be complete.

Our algorithm deviates significantly 
from all previous techniques in the literature of induced subgraph detection and is based on a clique decomposition for induced $4$-cycle-free graphs.
It follows the direction hinted at above, in which we identify an interesting property of dense induced $4$-cycle-free graphs that can be exploited algorithmically; for us, one such structural property is the existence of large cliques and a corresponding clique decomposition. Figuring out how to exploit this decomposition algorithmically is the hard part, and we do this by extracting much more structure from induced $4$-cycle freeness. We refer to Section~\ref{sec:overview} for a detailed technical overview.
To our knowledge, this is the first work in the large body of subgraph detection algorithms that employs 
the trendy topic 
of graph decompositions, akin to the use of expander decompositions for flow problems (see e.g.,~\cite{Saranurak21}).
While the latter technique is natural for global problems, it has not been meaningfully applied to local problems such as subgraph detection before  (putting aside other computational models such as in distributed computing where it has been used in breakthrough triangle detection algorithms \cite{chang2019distributed}).
Perhaps the most similar to our work is the use of decompositions based on regularity lemmas for triangle detection \cite{bansal2009regularity,AbboudFischerKelleyLovettMeka2024}; so far these techniques have only given subpolynomial improvements.

Finally, a strong motivation for seeking combinatorial algorithms is that they often come with added benefits, even when they are outperformed in runtime by non-combinatorial methods.
For example, one disadvantage of the $\tilde O(n^{\omega})$ algebraic algorithm from \cite{VassilevskaWangWilliamsYu2014} is that it is randomized,
    and appears difficult to derandomize because of its use of polynomial identity testing.
    If we restrict to deterministic algorithms, then no polynomial improvement over the trivial 4-clique runtime for induced 4-cycle detection was known.
    Meanwhile, as stated in Theorem~\ref{thm:deterministic}, our technique can be implemented deterministically, leading to the first truly subcubic \emph{deterministic} algorithm for detecting induced $4$-cycles (even among non-combinatorial algorithms).

\subsection{Related Work}
Detecting induced patterns has also received attention for many larger pattern graphs $H$. To name a few specific examples, the $5$-cycle ($H = C_5$) is well-understood (namely, triangle equivalent)~\cite{BlaserKomarathSreenivasaiah2018}, but for all larger cycles ($H = C_k$ for $k \geq 6$) there are gaps between lower and upper bounds: the fastest known algorithms either take $k$-clique time or $O(n^{k-2})$ time combinatorially~\cite{BlaserKomarathSreenivasaiah2018}, and the current best lower bound only shows that induced $H$ detection requires $\lfloor3k / 4\rfloor - \Theta(1)$-clique time~\cite[Theorem 2.3]{DalirrooyfardVassilevska2022}. 
The state of the art for paths ($H = P_k$) is similar~\cite{BlaserKomarathSreenivasaiah2018,DalirrooyfardVassilevska2022}.
More generally, it is known that \emph{any} $k$-node pattern graph $H$ requires $\Omega(\sqrt{k})$-clique detection time conditioned on Hadwiger's conjecture~\cite[Corollary 1.1]{DalirrooyfardVassilevskaVuong2019}, and $\Omega(k^{1/4})$-clique detection time unconditionally~\cite[Corollary 2.4]{DalirrooyfardVassilevska2022}. For a random $k$-node graph $H$, this lower bound improves to $\Omega(k / \log k)$-clique hardness~\cite[Corollary 1.2]{DalirrooyfardVassilevskaVuong2019}. Finally,~\cite{manurangsi2020strongish} gives $n^{\Omega(k)}$-hardness for all patterns conditioned on a stronger, less standard hypothesis.

The problems of detecting and counting induced subgraphs $H$ are also important in the field of parameterized complexity. Here the size $k = |H|$ is not fixed, but is viewed as a growing parameter. 
For this reason, the natural problem formulation is to fix a family of graphs~$\Phi$ (also called a \emph{graph property}) and to regard $H \in \Phi$ as part of the input. The typical goal is to achieve parameterized classification results that characterize the properties $\Phi$ for which the problem is FPT (i.e., can be solved in $f(k) \cdot \poly(n)$ time for some function $f$) versus those for which it is $\text{W}[1]$-hard (i.e., the problem is as hard as $g(k)$-clique detection for some function $g$). For induced subgraph detection, such a classification was achieved by Khot and Raman~\cite{KhotRaman2002} for the class of \emph{hereditary} properties: properties~$\Phi$ that are closed under taking induced subgraphs (see also~\cite{ChenThurleyWeyer2008,Eppstein2021}). 
Achieving such classifications for the counting problem has attracted even more attention \cite{JerrumMeeks2015, JerrumMeeks2015b, JerrumMeeks2016, Meeks2016, CurticapeanDellMarx2017, RothSchmitt2020, DorflerRothSchmittWellnitz2021, RothSchmittWellnitz2020, FockeRoth2022, DoringMarxWellnitz2024, CurticapeanNeuen2025, DoringMarxWellnitz2025}.

Beyond induced subgraph detection, there is endless literature on detecting subgraphs that are not necessarily induced; for a survey see \cite{MarxPilipczuk2014}. In particular, the non-induced $4$-cycle detection, counting, and listing problems play an important role in various reductions from fine-grained complexity \cite{DudekG19,DudekG20,AbboudBKZ22,AbboudBF23,JinX23,ChanX24}.

\section{Technical Overview} \label{sec:overview}
In this section we give a detailed overview of our algorithm for induced 4-cycle detection. Given an input graph $G = (V,E)$ on $n$ vertices,
our algorithm works in three steps:
\begin{enumerate}
    \item We first partition 
    the vertices
    $V = (\bigsqcup_{X \in \mathcal X} X) \sqcup R$ into a collection of \emph{large cliques} $\mathcal X$ together with a \emph{sparse remainder} $R$.
    \item We detect all induced $4$-cycles among the cliques in $\mathcal X$.
    \item We detect all induced $4$-cycles overlapping with the remainder $R$.
\end{enumerate}
We discuss these three steps individually in the following subsections.

A noteworthy recurring theme is that in several steps we design algorithms with a \emph{win/win} framework: 
either we directly win by finding an induced $4$-cycle, or we win by learning some new \emph{structure} in the input graph. 
We thereby accumulate more and more structural knowledge which we crucially exploit in the subsequent steps.

\subsection{Decomposition into Clusters}
    \label{subsec:overview-decomp}
Our starting point is the following observation: All relevant constructions of graphs without induced $4$-cycles---e.g., the graphs constructed in the fine-grained reduction of~\cite[Section 6]{DalirrooyfardVassilevska2022-arXiv} or in~\cite[Construction 1]{GyarfasHubenkoSolymosi2002}---consist of many large cliques. 
Could this be necessary? 
And could we possibly even obtain a decomposition theorem 
partitioning the input graph into many large cliques? 
A natural first instinct is to be skeptical since after all, the related expander decompositions~\cite{KannanVempalaVetta2004,SpielmanTeng2004,SaranurakWang2019} and regularity decompositions~\cite{Szemeredi1975,FriezeKannan1999} similarly partition a graph into some structured ``clusters'' together with a sparse  remainder, yet they achieve much weaker structural conditions for their respective clusters. 
However, recall that we have the freedom to settle for a win/win decomposition: 
either we achieve a good decomposition \emph{pretending} that the graph does not have induced $4$-cycles, or this pretense fails and we can immediately report an induced $4$-cycle. 
It indeed turns out that graphs $G$ avoiding induced $4$-cycles are structured enough to permit such an argument.

\paragraph{Structural Insight 1: Large Cliques.}
For instance, take any two non-adjacent nodes $x$ and~$y$ in~$G$, and consider their set of common neighbors $N(x) \cap N(y)$. This set must be a clique, as otherwise there exists a non-edge $\grp{z, w}$ in $N(x) \cap N(y)$, 
which then forces $(x, z, y, w)$ 
to be an induced $4$-cycle 
in~$G$. 
Consequently, if $G$ has two non-adjacent nodes $x, y$ whose common neighborhood~\makebox{$N(x) \cap N(y)$} is large, then we have identified a large clique. 
Otherwise, all of these common neighborhoods are small.
In this case, we intuitively expect the graph $G$ to be sparse. This  suggests that any graph without induced $4$-cycles is either \emph{sparse} or contains a \emph{large clique}. 
This suspicion turns out to be correct as proven by Gy\'{a}rf\'{a}s, Hubenko, and Solymosi~\cite{GyarfasHubenkoSolymosi2002}:

\begin{theorem}[{\cite[Theorem 1]{GyarfasHubenkoSolymosi2002}}] \label{thm:4cycle-free-large-clique}
Any $n$-node graph with average degree $d$ that contains no induced $4$-cycle
must have a clique of size $\Omega(d^2 / n)$.
\end{theorem}

The proof of \cref{thm:4cycle-free-large-clique} is simple and elegant.
We quickly sketch the main idea: Let $I$ be a maximum-size independent set in $G$. Then for each pair of distinct nodes $x, y \in I$, their set of common neighbors $N(x) \cap N(y)$ is a clique by the same argument as before.
In addition, for each node $x \in I$, the set $U(x)$ consisting of all nodes in $G$ for which $x$ is the \emph{unique} neighbor in $I$ must be a clique---otherwise, if $U(x)$ contains a non-edge $\set{z, w}$, then $\grp{I \setminus \set{x}} \cup \set{z, w}$ would be an independent set of larger size than $I$. Finally, an averaging argument shows that there is a distinct pair~\makebox{$x, y \in I$} with $|N(x) \cap N(y)| \geq \Omega(d^2 / n)$ or there is a node $x \in I$ with $|U(x)| \geq \Omega(d) = \Omega(d^2 / n)$. In either case we find a clique as claimed.

\paragraph{Decomposition Algorithm.}
Our first contribution is to refine the structural result due to 
Gy\'{a}rf\'{a}s, Hubenko, and Solymosi
\cite{GyarfasHubenkoSolymosi2002}, and turn it into an efficient decomposition algorithm:

\begin{restatable}[Large Cluster Decomposition]{lemma}{lemlargeclusterdecomp} \label{lem:cluster-decomp-large}
Let $G = (V, E)$ be the input graph and let $\Delta \geq 1$. There is a deterministic $O(n^3 / \Delta)$-time algorithm that either detects an induced 4-cycle in $G$, or computes a decomposition
\begin{equation*}
    V = \left(\bigsqcup_{X \in \mathcal X} X\right) \sqcup R,
\end{equation*}
where each $X \in \mathcal X$ is a clique of size $\Theta(\Delta)$ in $G$, and $G[R]$ has at most $O(n^{3/2} \Delta^{1/2})$ edges.
\end{restatable}

In summary, \cref{lem:cluster-decomp-large} gives a subcubic-time algorithm that computes a decomposition into large cliques~$\mathcal X$ plus some remainder $R$ such that the total number of edges in $R$ is small.
In fact,  the construction in~\cite[Construction 1]{GyarfasHubenkoSolymosi2002} shows that this sparsity bound is \emph{optimal}.\footnote{\label{footnote:polar}We follow the construction in~\cite{GyarfasHubenkoSolymosi2002}. It is well-known~\cite{Brown1966,ErdosRenyiSos1966} that there are graphs (so-called \emph{polarity graphs}) with $N$ vertices and~$\Theta(N^{3/2})$ edges with no $4$-cycles (induced or otherwise). Additionally, these graphs are bipartite and thus do not contain triangles. Take any such graph on \smash{$N = 3n/\Delta$} vertices, replace each node $v$ by a \smash{($\Delta/3$)}-size clique~$C_v$ and replace each edge $\set{u, v}$ by a biclique between~$C_u$ and $C_v$. This yields a graph $G$ on $n$ vertices and $\Theta(N^{3/2} \cdot \Delta^2) = \Theta(n^{3/2} \Delta^{1/2})$ edges,
whose largest clique  has size \smash{$2\Delta/3 < \Delta$}. Hence, the decomposition from \Cref{lem:cluster-decomp-large} cannot find any cliques in $G$ of size at least $\Delta$, and thus must return  the trivial partition where $R$ consists of all nodes. Hence, the sparsity bound $\tilde O(n^{3/2} \Delta^{1/2})$ in \Cref{lem:cluster-decomp-large} is the best  possible.}

There are two technicalities to
moving from \Cref{thm:4cycle-free-large-clique} to
\cref{lem:cluster-decomp-large}.
First, the proof of \cref{thm:4cycle-free-large-clique} does not directly yield a polynomial-time
algorithm,
since computing a maximum-size independent set in a graph is \NP{}-hard in general \cite[Theorem 2.15]{AroraBarak2009}. 
To obtain an efficient algorithm,
we instead compute an independent set $I$ that is \emph{maximal}, i.e., cannot be extended to $I \cup \set{z}$, and that additionally cannot be extended by simple node-exchanges of the form $\grp{I \setminus \set{x}} \cup \set{z, w}$, for any vertices $x,z,w$ in the graph.
Second, \cref{thm:4cycle-free-large-clique} only guarantees the existence of \emph{one} large clique in $G$.
To obtain the full decomposition in  \Cref{lem:cluster-decomp-large}, we iteratively apply \Cref{thm:4cycle-free-large-clique}.

\subsection{Detecting Induced 4-Cycles in Clustered Instances}
    \label{subsec:overview-clustered}
By \cref{lem:cluster-decomp-large}, we can assume  the input graph is decomposed into at most $n / \Delta$ cliques of size $\Theta(\Delta)$, plus some sparse remainder $R$ (here we use the fact that the cliques returned by \Cref{lem:cluster-decomp-large} are disjoint). 
We refer to these cliques as \emph{clusters}.
For concreteness let us set $\Delta = \sqrt{n}$ throughout this overview. 
The next big step is to test if there is an induced $4$-cycle among the clusters (ignoring the remainder for now). What could such an induced $4$-cycle look like? 
As illustrated in \Cref{fig:clustered}, there are three options: 
the $4$-cycle spans either two, three, or four clusters, respectively (a $4$-cycle cannot be contained entirely in a single cluster, since each cluster is a clique.) 
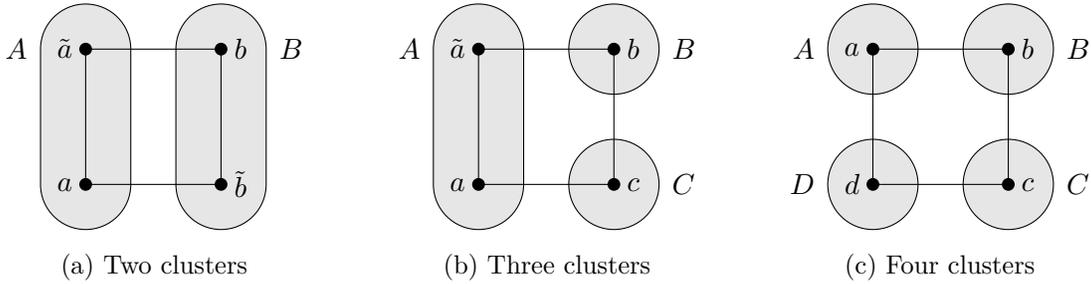
\begin{figure}[t]
\caption{Any induced 4-cycle contained entirely in the clusters spans either exactly (a) two, (b) three, or (c) four clusters.} \label{fig:clustered}
\bigskip%
\tikzset{
    vtx/.style={inner sep=0.6mm, circle, fill},
    edg/.style={},
    clu/.style={fill=gray!20!white, rounded corners=\clupad},
    lab/.style={},
}%
\def\vtxsep{18mm}%
\def\clupad{6mm}%
\def\labsep{0.4mm}%
\centering%
\hfill%
\begin{subfigure}[t]{.31\linewidth}%
\centering%
\begin{tikzpicture}%
    \draw[clu] (-\clupad, \vtxsep+\clupad) rectangle (\clupad, -\clupad);
    \path (-\clupad, \vtxsep) node[lab, left=\labsep] {$A$};
    \draw[clu] (\vtxsep-\clupad, \vtxsep+\clupad) rectangle (\vtxsep+\clupad, -\clupad);
    \path (\vtxsep+\clupad, \vtxsep) node[lab, right=\labsep] {$B$};
    \path (0, \vtxsep) node[vtx] (a1) {} node[lab, left=\labsep] {$\tilde a$};
    \path (0, 0) node[vtx] (a2) {} node[lab, left=\labsep] {$a$};
    \path (\vtxsep, \vtxsep) node[vtx] (b1) {} node[lab, right=\labsep] {$b$};
    \path (\vtxsep, 0) node[vtx] (b2) {} node[lab, right=\labsep] {$\tilde b$};
    \draw[edg] (a1.center) -- (b1.center) -- (b2.center) -- (a2.center) -- cycle;
\end{tikzpicture}%
\smallskip%
\caption{Two clusters}\label{fig:clustered:sub:two}%
\end{subfigure}%
\hfill%
\begin{subfigure}[t]{.31\linewidth}%
\centering%
\begin{tikzpicture}%
    \draw[clu] (-\clupad, \vtxsep+\clupad) rectangle (\clupad, -\clupad);
    \path (-\clupad, \vtxsep) node[lab, left=\labsep] {$A$};
    \draw[clu] (\vtxsep-\clupad, \vtxsep+\clupad) rectangle (\vtxsep+\clupad, \vtxsep-\clupad);
    \path (\vtxsep+\clupad, \vtxsep) node[lab, right=\labsep] {$B$};
    \draw[clu] (\vtxsep-\clupad, \clupad) rectangle (\vtxsep+\clupad, -\clupad);
    \path (\vtxsep+\clupad, 0) node[lab, right=\labsep] {$C$};
    \path (0, \vtxsep) node[vtx] (a1) {} node[lab, left=\labsep] {$\tilde a$};
    \path (0, 0) node[vtx] (a2) {} node[lab, left=\labsep] {$a$};
    \path (\vtxsep, \vtxsep) node[vtx] (b) {} node[lab, right=\labsep] {$b$};
    \path (\vtxsep, 0) node[vtx] (c) {} node[lab, right=\labsep] {$c$};
    \draw[edg] (a1.center) -- (b.center) -- (c.center) -- (a2.center) -- cycle;
\end{tikzpicture}%
\smallskip%
\caption{Three clusters}\label{fig:clustered:sub:three}%
\end{subfigure}%
\hfill%
\begin{subfigure}[t]{.31\linewidth}%
\centering%
\begin{tikzpicture}%
    \draw[clu] (-\clupad, \vtxsep+\clupad) rectangle (\clupad, \vtxsep-\clupad);
    \path (-\clupad, \vtxsep) node[lab, left=\labsep] {$A$};
    \draw[clu] (\vtxsep-\clupad, \vtxsep+\clupad) rectangle (\vtxsep+\clupad, \vtxsep-\clupad);
    \path (\vtxsep+\clupad, \vtxsep) node[lab, right=\labsep] {$B$};
    \draw[clu] (\vtxsep-\clupad, \clupad) rectangle (\vtxsep+\clupad, -\clupad);
    \path (\vtxsep+\clupad, 0) node[lab, right=\labsep] {$C$};
    \draw[clu] (-\clupad, \clupad) rectangle (\clupad, -\clupad);
    \path (-\clupad, 0) node[lab, left=\labsep] {$D$};
    \path (0, \vtxsep) node[vtx] (a) {} node[lab, left=\labsep] {$a$};
    \path (0, 0) node[vtx] (d) {} node[lab, left=\labsep] {$d$};
    \path (\vtxsep, \vtxsep) node[vtx] (b) {} node[lab, right=\labsep] {$b$};
    \path (\vtxsep, 0) node[vtx] (c) {} node[lab, right=\labsep] {$c$};
    \draw[edg] (a.center) -- (b.center) -- (c.center) -- (d.center) -- cycle;
\end{tikzpicture}%
\smallskip%
\caption{Four clusters}\label{fig:clustered:sub:four}%
\end{subfigure}%
\hspace*{\fill}%
\end{figure}

\subsubsection{Two Clusters}
Our first goal is to detect induced $4$-cycles contained in two clusters. We enumerate all $(n / \Delta)^2$ pairs of clusters $A, B \in \mathcal X$, and for each such pair we run an algorithm to determine if there is an induced $4$-cycle in $G[A \sqcup B]$ in time $O(\Delta^2)$.
In particular, the total running time of this procedure is~$O(n^2)$. 
Our $O(\Delta^2)$-time algorithm is based on the following structural insight:

\paragraph{Structural Insight 2: Ordered Clusters.}
We call a pair of clusters $(A, B)$ \emph{ordered} if there are functions \makebox{$f \colon A \to \mathbb{Z}$} and $g \colon B \to \mathbb{Z}$ such that each pair $(a, b) \in A \times B$ is an edge in graph $G$ if and only if $f(a) \leq g(b)$. Then $G[A \sqcup B]$ contains no induced $4$-cycle \emph{if and only if} $(A, B)$ is ordered. 

\medskip
To see why this is true, first assume that $G[A \sqcup B]$ has no induced $4$-cycle. 
We define
\begin{equation}
\label{eq:orderings-def}
    g(b) = \deg_A(b) \qquad\text{and}\qquad f(a) = \min_{b \in N_B(a)} g(b)
\end{equation}
for all $b\in B$ and $a\in A$ respectively. On the one hand, whenever there is an edge~$(a, b)$ then by definition we  have $f(a) \leq g(b)$. On the other hand, assume that $(a, b)$ is a non-edge and suppose that $f(a) \leq g(b)$. Then by definition there is some node $\tilde b \in B$ adjacent to $a$ such that $\deg_A(\tilde b) \leq \deg_A(b)$. 
Since $a$ is adjacent to $\tilde b$ but not to $b$, the degree condition implies there is a node $\tilde a$ adjacent to $b$ but not to $\tilde b$. This then forces $(a, \tilde a, b, \tilde b)$ to be an induced $4$-cycle,
contradicting our initial assumption.

Conversely, suppose that $(A, B)$ is ordered as witnessed by some functions $f\colon A\to\ZZ, g\colon B\to\ZZ$. 
Each induced $4$-cycle in~$G[A \sqcup B]$ must take the form $(a, \tilde a, b, \tilde b)$ for $a, \tilde a \in A$ and $b, \tilde b \in B$ as depicted in \cref{fig:clustered:sub:two}. 
In particular, $f(a) > g(b)$ and $f(\tilde{a}) > g(\tilde{b})$,
yet $f(a) \le g(\tilde{b})$ and $f(\tilde{a}) \le g(b)$,
so we have 
\[f(a) + f(\tilde a) > g(b) + g(\tilde b) \geq f(\tilde a) + f(a)\]
which again yields a contradiction. 
For the full details of this proof, see \cref{lem:cluster-pair}.

\paragraph{Algorithmic Implications.}
The above insight provides an easy way to test if $G[A \sqcup B]$ has an induced $4$-cycle in $O(\Delta^2)$ time: Compute the functions $g(b) = \deg_A(b)$ and $f(a) = \min_{b \in N_B(a)} g(b)$ as in \cref{eq:orderings-def}, and verify the order condition naively in time $O(\Delta^2)$.

However, this structural insight offers much more! 
This algorithm can again be interpreted as implementing a win/win strategy: Either we have detected an induced $4$-cycle, and we are done, or we have learned for all the following steps that the inter-cluster edges are \emph{highly structured}. 
In particular, we only require $\tilde O(\Delta)$ bits to specify the edges between any two clusters $A, B$. 
Moreover, 
testing adjacency between nodes in two clusters reduces to an arithmetic comparison of their associated function values.
This will let us, perhaps surprisingly, leverage geometric data structures to efficiently check for patterns among the clusters. 

\subsubsection{Three Clusters}
Next, we test if there is an induced $4$-cycle in three distinct clusters $A, B, C$. Again, we enumerate all $(n / \Delta)^3$ distinct triples $A, B, C \in \mathcal X$, and design an algorithm to detect if $G[A \sqcup B \sqcup C]$ contains an induced $4$-cycle in  \smash{$\tilde O(\Delta)$} time. For $\Delta = \sqrt{n}$, the total running time is again \smash{$\tilde O(n^2)$}.

\paragraph{Warm-Up: Orthogonal Range Queries.}
As a warm-up, and to give a simple demonstration of how geometric data structures are useful in this context, we begin by describing a simpler algorithm running in  $\tilde O(\Delta^2)$ time. 
Without loss of generality, we may assume each induced $4$-cycle in $G[A \sqcup B \sqcup C]$ takes the form $(a, \tilde a, b, c)$ for~\makebox{$a, \tilde a \in A,\, b \in B,\, c \in C$}, as depicted in \cref{fig:clustered:sub:three}. From the two-cluster case, we may assume that all pairs of clusters $(A, B)$, $(A, C)$, $(B, C)$ are ordered, and that we have access to functions $f_{AB}, g_{AB}, f_{AC}, g_{AC}, f_{BC}, g_{BC}$ describing the inter-cluster edges (i.e., $f_{AB}(a) \leq g_{AB}(b)$ if and only if $(a, b) \in A \times B$ is an edge, and similarly for the other cluster pairs). Consider the following set of points in $\ZZ^4$:
\begin{equation*}
    P = \set{\pair{f_{AB}(a), f_{AB}(\tilde a), f_{AC}(a), f_{AC}(\tilde a)} \mid a, \tilde a \in A, a \neq \tilde a}.
\end{equation*}
We process $P$ into a 4-dimensional orthogonal range query data structure (i.e., a structure with the property that whenever we query it with a 4-dimensional axis-aligned box, we can determine in polylogarithmic time if the given box contains a point in $P$)
in $\tilde{O}(|P|) \le \tilde{O}(\Delta^2)$ time. Then we enumerate all edges $(b, c) \in B \times C$, and for each query if $P$ intersects the box
\begin{equation*}
    (g_{AB}(b), \infty) \times (-\infty, g_{AB}(b)] \times (-\infty, g_{AC}(c)] \times (g_{AC}(c), \infty).
\end{equation*}
It is easy to verify that this box intersects $P$ if and only if there is an induced $4$-cycle $(a, \tilde a, b, c)$. 
Since we make at most $\Delta^2$ queries, 
this procedure takes at most $\tilde O(\Delta^2)$ time as claimed.

While this algorithm is sufficiently fast to detect induced $4$-cycles spanning up to three clusters (in total time $\tilde O( (n/\Delta)^3 \cdot \Delta^2) \le \tilde O(n^3 / \Delta) \le \tilde O(n^{5/2})$), this argument does not efficiently carry over to the upcoming four-cluster case. For this reason, we invest additional effort to solve the three-cluster case in $\tilde O(\Delta)$ time, and along the way discover more structural insights.

\paragraph{Structural Insight 3: Comparable Neighborhoods.}
Let $(b,c)\in B\times C$ be an edge,
and let $a,\tilde{a}\in A$ be distinct vertices.
Our next observation is that the quadruple $(a, \tilde a, b, c)$ is not an induced $4$-cycle if and only if $N_A(b) \subseteq N_A(c)$ or $N_A(c) \subseteq N_A(b)$.
If the sets $N_A(b)$ and $N_A(c)$ satisfy either of these inclusions, we say they are \emph{comparable}.

To see this, note that whenever $(b, c)$ is an edge and $N_A(b)$ and $N_A(c)$ are not comparable then there are nodes $a \in N_A(c) \setminus N_A(b)$ and $\tilde a \in N_A(b) \setminus N_A(c)$, and thus $(a, \tilde a, b, c)$ is an induced $4$-cycle. Conversely, if there is an induced $4$-cycle $(a, \tilde a, b, c)$ then clearly $N_A(b)$ and $N_A(c)$ are not comparable. This is really just a convenient reformulation of what it means for there to be no induced $4$-cycle in $G[A \sqcup B \sqcup C]$ (see \Cref{obs:triple-comparable-neighborhood} for the full details).

\paragraph{An Improved Algorithm.}
We leverage this insight algorithmically as follows. For $b \in B$, define
\begin{equation*}
    h_\low(b) = \min_{a \in A\setminus N_A(b)} f_{AC}(a) 
    \quad
    \text{and}
    \quad 
    h_\high(b) = \max_{a \in N_A(b)} f_{AC}(a).
\end{equation*}
These values are chosen such that for any edge $(b, c)\in B\times C$, 
the sets $N_A(b)$ and $N_A(c)$ are incomparable if and only if $\smash{h_\low(b) \leq g_{AC}(c) < h_\high(b)}$. Indeed, if $N_A(b)$ and $N_A(c)$ are incomparable then there is some $a \in N_A(c) \setminus N_A(b)$ witnessing the inequality $h_\low(b) \leq f_{AC}(a) \leq g_{AC}(c)$, and some $\tilde a \in N_A(b) \setminus N_A(c)$ witnessing the inequality $g_{AC}(c) < f_{AC}(\tilde a) \leq h_\high(b)$.
Chaining these inequalities together,
we get $h_\low(b) \le g_{AC}(c) < h_\high(b)$.
Similar reasoning proves the converse.

Therefore, to test if there is an induced $4$-cycle in $G[A \sqcup B \sqcup C]$ it suffices to compute the values $h_\low(b)$ and $h_\high(b)$ for all $b \in B$, and then test if there is a pair $(b, c) \in B \times C$ of vertices such that $h_\low(b) \leq g_{AC}(c) < h_\high(b)$ 
(enforcing that $N_A(b)$ and $N_A(c)$ are incomparable)
and $f_{BC}(b) \leq g_{BC}(c)$ (enforcing that $(b,c)$ is an edge). 
It turns out that each of these steps can be performed in  $\tilde O(\Delta)$ time using a 2-dimensional orthogonal range query data structure,
which yields the desired algorithm.

\subsubsection{Four Clusters}
The most difficult case remains: Detecting if there is an induced $4$-cycle spanning four different clusters $A, B, C, D \in \mathcal X$. There can be up to $(n / \Delta)^4 = n^2$ quadruples of such clusters, so in order to obtain a truly subcubic  $n^{3-\Omega(1)}$ runtime,
we  have to design an algorithm that runs in truly subquadratic  $\Delta^{2-\Omega(1)}$ time per quadruple. With considerable technical effort, we obtain an algorithm that runs in $\tilde O(\Delta)$ time, thus leading to an $\tilde O(n^{5/2})$ time algorithm overall for $\Delta = \sqrt{n}$. 

Conceptually, our challenge is that, despite our strong knowledge that edges between two clusters are nicely ordered, we do not yet know how two separate orderings $(A, B)$ and $(A, C)$ relate to one another. 
It seems reasonable to expect that the orderings are \emph{correlated} in some way, unless there is an induced $4$-cycle in $G[A \sqcup B \sqcup C]$ (which we would have detected in the three-cluster case). 
This turns out  to be true, and is captured by the following statement:

\paragraph{Structural Insight 4: Orderings Correlate.}
For each node $b \in B$ (and similarly for each node $d \in D$) there are values $k(b) < K(b)$ and $L(b) < \ell(b)$ such that:
\begin{itemize}[itemsep=0pt]
    \item all vertices $a \in A$ with $f_{AC}(a) \leq k(b)$ are neighbors of $b$, and\newline all vertices $c \in C$ with $g_{AC}(c) \geq \ell(b)$ are neighbors of $b$; and 
    \item the only other neighbors $a \in A$ of $b$ satisfy that $f_{AC}(a) = K(b)$, and\newline the only other neighbors $c \in C$ of $b$ satisfy that $g_{AC}(c) = L(b)$.
\end{itemize}
The proof of this statement relies on repeated applications of our previous structural insights. It is on the technical side, and we omit the details here,
deferring the argument to \Cref{lem:splitconn:computation}. 
Instead, we elaborate on how this insight aids in detecting induced $4$-cycles in $G[A \sqcup B \sqcup C \sqcup D]$.

\paragraph{Algorithmic Implications.}
It is clear that each induced $4$-cycle falls into one of two categories:
\begin{itemize}[itemsep=0pt]
    \item \emph{Ordinary:} $f_{AC}(a) \leq k(b)$ and $f_{AC}(a) \leq k(d)$ and $g_{AC}(c) \geq \ell(b)$ and $g_{AC}(c) \geq \ell(b)$.
    \item \emph{Exceptional:} $f_{AC}(a) = K(b)$ or $f_{AC}(a) = K(d)$ or $g_{AC}(c) = L(b)$ or $g_{AC}(c) = L(b)$.
\end{itemize}
We design two different algorithms for these two types of $4$-cycles, starting with the ordinary case. 

When seeking an induced $4$-cycle $(a, b, c, d)$ in  the ordinary case, the main benefit we obtain by the ordinary assumption is that we no longer have to test if the edges $(a, b), (b, c), (a, d), (d, c)$ are present,
because these edges are guaranteed by the structural insight. 
Thus, we merely need to check if there is a quadruple $(a, b, c, d)$ in the ordinary case such that $(a, c)$ and $(b, d)$ are not edges.
With a little care,\footnote{This step requires a certain \emph{conciseness} condition on the values $f_{AC}, g_{AC}$ and $k, K, L, \ell$ which we defer to \cref{sec:cluster}.} this can be tested by detecting a pair $(b, d)$ satisfying the two conditions \makebox{$\max(\ell(b), \ell(d)) < \min(k(b), k(d))$} and $f_{BD}(b) > g_{BD}(d)$. Finally, these last conditions are efficiently testable in  $\tilde O(\Delta)$ time, by  employing orthogonal range query data structures once again.

Next, we focus on seeking an induced $4$-cycle in the exceptional case.
Without loss of generality, suppose 
we seek a solution with
$f_{AC}(a) = K(b)$. Here we have another benefit: Looking only at $b$, we can directly infer information about the edges between clusters $A$ and $C$. Specifically, there is an induced $4$-cycle of this type if and only if there is a pair $(b, d)$ such that 
\begin{enumerate}[itemsep=0pt]
    \item 
        there is a 2-path $(b, a, d)$ with $f_{AC}(a) = K(b)$, 
    \item 
        there is a 2-path $(b, c, d)$ with $g_{AC}(c) < K(b)$, and
    \item 
        \makebox{$(b, d)$ is not an edge}.
\end{enumerate}
Crucially, we do not have to test if $(a, c)$ is a non-edge,
because we enforce this property ``for free''
by employing the assumption that $f_{AC}(a) = K(b)$ in condition 1 (permitted because we are in the exceptional case),
and restricting to $c\in C$ with $g_{AC}(c) < K(b)$ in condition 2.
Condition 3  can  be expressed as $f_{BD}(b) > g_{BD}(d)$. 
Conditions 1 and 2 above can be expressed as $i(b) \leq g_{AD}(d)$ and $j(b) \leq g_{CD}(d)$, for appropriate functions $i, j\colon B\to\ZZ$ (e.g., let $i(b)$ be the minimum value of $f_{AD}(a)$ among all $a \in N_A(b)$ with~\makebox{$f_{AC}(a) = K(b)$}). 
Since these are all simple arithmetic comparisons, we can test these three conditions in  $\tilde O(\Delta)$ time using an orthogonal range query
data structure.

Combining these algorithms for the ordinary and exceptional cases solves the four-cluster case.

\subsection{Dealing with the Sparse Remainder}
In \Cref{subsec:overview-decomp} we outlined an algorithm to decompose the input graph into cliques of size $\Theta(\Delta)$, plus some remainder $R$ with at most $\tilde O(n^{3/2} \Delta^{1/2}) \le \tilde O(n^{7/4})$ edges, and
then in \Cref{subsec:overview-clustered}
we presented an 
$\tilde{O}(n^{5/2})$-time algorithm for detecting an
induced $4$-cycle among the cliques. It remains to test if there are induced $4$-cycles involving the remainder set $R$.

Unfortunately, while the algorithm so far has arguably been quite clean, dealing with the remainder becomes somewhat messy. A conceptually similar phenomenon appears for related problems such as detecting directed $4$-cycles, or counting $4$-cycles in sparse graphs; for both these problems it seems reasonable to expect triangle-detection-time \smash{$O(m^{2\omega/(\omega+1)})$} algorithms~\cite{AlonYusterZwick1997}, but the state of the art is \smash{$O(m^{(4\omega-1)/(2\omega+1)})$}~\cite{yuster2004detecting,VassilevskaWangWilliamsYu2014},
where $m$ denotes the number of edges in the input graph.

Our concrete goal is to test if there is an induced $4$-cycle that involves at least one node from~$R$. It is much simpler to test if there is an induced $4$-cycle with all four nodes in $R$. For instance we could run the $O(m^{11/7})$-time combinatorial (albeit randomized) algorithm due to~\cite{VassilevskaWangWilliamsYu2014}, which in our case takes subcubic time $O((n^{7/4})^{11/7}) = O(n^{11/4})$. But how can we deal with the induced $4$-cycles with some nodes in the cliques $\mathcal X$ and some nodes in the remainder $R$? Note that for an induced $4$-cycle with three nodes among the cliques and just one node in $R$, it seems hard to exploit the sparsity condition in $R$ at all.

Our solution is to \emph{further decompose} $R$ into smaller cliques. More specifically, setting $L = \log n$, we
extend the decomposition from \Cref{lem:cluster-decomp-large} to obtain a partition of the vertices 
\[V = \bigsqcup_{\ell = L/2}^{L} V_\ell\]
into $(1/2)(\log n)$ \emph{levels}, where level $V_\ell$ is a disjoint union of cliques of size $\Theta(n / 2^\ell)$, while ensuring the stronger sparsity condition $|N_{V_\ell}(x) \cap N_{V_\ell}(y)| \leq O(n / 2^\ell)$ for all non-adjacent nodes $x, y$. See \cref{thm:clique-decomp} for the formal statement of this decomposition.

Now any induced $4$-cycle in the graph can be viewed as having nodes belonging to the different cliques comprising the $V_\ell$ levels.
However, these cliques can have very different sizes.
We perform casework on the relative sizes of these cliques.

For
some clique sizes, it is efficient enough to run the structure-based algorithm outlined in \Cref{subsec:overview-clustered}.
Intuitively, this works whenever enough nodes of the induced $4$-cycle belong to very large cliques.
For other sizes, it is more efficient to employ a sparsity-sensitive algorithm. 
For example, to test if there is an induced $4$-cycle among the constant-size cliques in $V_{L}$, we  simply enumerate all non-edges~$(x, y)$ and all choices of common neighbors $z, w \in N_{V_{L}}(x) \cap N_{V_{L}}(y)$. 
By the sparsity condition there are only a constant number of such common neighbors for each choice of~$x,y$, so detecting an induced $4$-cycle takes time~$O(n^2)$ in this case.

A full description of this casework is presented in \cref{sec:full}.
The bottom line is that with an appropriate trade-off between these two approaches (and en route some additional structural insights such as \Cref{obs:neighborhood-size}) we achieve an algorithm that runs in $\tilde O(n^{17/6})$ time overall. 
The bottleneck of our algorithm lies in detecting an induced $4$-cycle with two nodes in ``large'' cliques of size $\sqrt{n}$ and two nodes in ``moderate'' cliques of size~$n^{1/3}$.

\subsection{Relation to the Erd\H{o}s-Hajnal Conjecture}
In a nutshell, our algorithm is based on the fact that we can identify large cliques and exploit that the edges between pairs, triples, 
and quadruples of these cliques must be highly structured. In particular, the larger the cliques we find, the more structure we can infer, and the better our algorithm performs. A related approach would be to extract \emph{independent sets}, which similarly admit some (seemingly weaker) structural properties. 
This suggests that one way to improve our algorithm would be to consider a more general decomposition into cliques \emph{and} independent sets.

This potential approach to solving induced 4-cycle detection is related to the \emph{Erd\H{o}s-Hajnal Conjecture}~\cite{ErdosHajnal1977}. 
This conjecture postulates that for any fixed pattern graph $H$, there exists a constant $\eps = \eps(H) > 0$ such that every $n$-node graph not containing an induced copy of $H$ must have a clique or an independent set of size at least~$\Omega(n^\eps)$. This conjecture has been proven for certain simple classes of pattern graphs, but remains wide open in general (see e.g., \cite{BucicNguyenScottSeymour2024}).

In our setting it would be interesting to obtain tight quantitative bounds on $\eps(H)$ when $H = C_4$, i.e., what is the largest $\eps$ such that each induced $C_4$-free graph contains a clique or independent set of size~$\Omega(n^\eps)$? To our knowledge, this question is open. 
The current best lower bound appears to be $\eps \geq 1/3$ as follows from \cref{thm:4cycle-free-large-clique}.\footnote{Specifically, if $G$ has more than $n^{5/3}$ edges then \cref{thm:4cycle-free-large-clique} implies the existence of a clique of size $\Omega(n^{1/3})$. If~$G$ has less than $n^{5/3}$ edges then it contains an independent set of size $\tilde\Omega(n^{1/3})$ by a greedy construction.} The current best upper bound is $\eps \leq 2/5$ which follows from a graph constructed by the probabilistic method~\cite[Theorem 3.1]{Spencer1977}.\footnote{We sketch the argument here. We first construct an $N$-node graph $\tilde{G}$ that has no $4$-cycles (induced or otherwise) such that the largest independent set in $\tilde{G}$ has size at most $\tilde{O}(N^{2/3})$. To this end, take a random graph~$\tilde{G}$ on $N$ that contains each edge uniformly and independently with probability $p = N^{-2/3} / 2$. 
With high probability, this graph will have $\Theta(N^{4/3})$ edges. Moreover, for each edge $e$, the probability that $e$ is involved in a $4$-cycle is at most $p^3 N^2 \leq 1/2$, so in expectation we can afford to remove all edges involved in $4$-cycles and still keep $\Theta(N^{4/3})$ edges. It can also be verified that the largest independent set in the resulting graph has size $\tilde O(N^{2/3})$ (as would be expected from a truly random graph). Now let $G$ be the graph on $n$ nodes where we replace each node in $\tilde{G}$ by a clique of size $n / N$ and each edge by a biclique (as in \cref{footnote:polar}). Choosing~$N = n^{3/5}$, the largest independent set in $G$ is $\tilde O(N^{2/3}) \le \tilde{O}(n^{2/5})$, and the largest clique has size $n / N = n^{2/5}$ as well.} 
By closing this gap and understanding the structure of the corresponding extremal graphs, we could potentially learn of interesting instances for induced 4-cycle detection that could inspire new algorithmic insights. 
\section{Preliminaries}
    \label{sec:prelim}

For a positive integer $a$, we let $[a] = \set{1, \dots, a}$ denote the set of the first  $a$ positive integers.
For a vector $\vec{v}$ and an index $i$, we let $\vec{v}[i]$ denote the $i^{\text{th}}$ coordinate of $v$. 
We say two sets $S$ and $T$ are \emph{comparable} if either $S\sub T$ or $T\sub S$.
We say that $S$ and $T$ are \emph{incomparable} if neither set contains the other. 
We let $\im(f)$ denote the image of a function $f$. 
By convention, the minimum and maximum over an empty set are $\infty$ and $-\infty$ respectively.

    \begin{proposition}[Bonferroni's Inequality]
        \label{prop:bonferroni}
        Given a family of finite sets $\cal{S}$, we have 
            \[\abs{\bigcup_{S\in\cal{S}} S} \ge \sum_{S\in\cal{S}} |S| - 
            \sum_{\substack{S,T\in\cal{S} \\ S\neq T }} |S\cap T|.\]
    \end{proposition}

\subsubsection*{Graph Notation}

    Throughout, we let $G$ denote the input graph on $n$ vertices.
    We let $V$ and $E$ denote the vertex and edge sets of $G$ respectively.
    Given a node $v\in V$ and subset of vertices $R\sub V$,
    we let $N_R(v)$ denote the set of vertices in $R$ adjacent to $v$ in $G$,
    $\deg_R(v) = |N_R(v)|$ denote the degree of $v$ in $R$, and $\codeg_R(v,w) = |N_R(v)\cap N_R(w)|$ denote the number of common neighbors of vertices $v$ and $w$ in $R$.
    In the case where $R = V$ is the whole vertex set,
    we omit the subscript $R$.
    We let $G[R]$ denote the induced subgraph of $G$ restricted to the vertices in $R$. 

    We let $C_4$ denote the cycle on four vertices.
    We say a tuple $(w,x,y,z)$ forms an induced 4-cycle if $(w,x)$, $(x,y)$, $\grp{y,z}$, and $\grp{z,w}$ are edges,
    but $\grp{w,y}$ and $\grp{x,z}$ are not edges.
    We say a tuple $(u,v,w)$ of distinct vertices forms a 2-path
    if $\grp{u,v}$ and $\grp{v,w}$ are both edges.
    We call the tuple an induced 2-path
    if in addition $\grp{u,w}$ is not an edge.

\subsubsection*{Data Structures}

    Our algorithms make extensive use of data structures for orthogonal range searching.
    We refer the reader to \cite[Chapter 5]{deBergCheongvanKreveldOvermars2008}
     for a primer on this topic. 

    \begin{proposition}[Orthogonal Range Queries]
        \label{prop:range-search}
        Let $d$ be a fixed positive integer. 
        Given a set of $n$ points $S \subset \ZZ^d$,
        we can in $O(n (\log n)^d)$ time construct a data structure that, when given any $d$-dimensional axis-parallel box $B$ as a query,
        returns the value $|S \cap B|$ and a point in $S\cap B$, if any exists,  in $O((\log n)^d)$ time.
    \end{proposition}

\subsubsection*{Search to Decision Reduction}

    For the sake of simplicity, we describe our algorithm in \Cref{thm:deterministic} 
    as \emph{detecting} the presence of an induced 4-cycle in the input graph instead of returning one when it exists. 
    This turns out to be without of loss of generality, because any algorithm for detecting induced 4-cycles in  $n$-node graphs can be converted into an algorithm for \emph{finding} induced 4-cycles in $n$-node graphs with only an $O(\log n)$ overhead.
    In the statement below,
    recall that a function $T$ is \emph{subadditive}
    if for all $a,b$ we have $T(a+b) \le T(a) + T(b)$.

    \begin{proposition}[Search to Detection Reduction]
        If there is an algorithm $\mathcal A$ that can decide if an $n$-node graph contains an induced 4-cycle in time $T(n)$ (for some subadditive function $T$),
        then there is an algorithm that can \emph{find} an induced 
        4-cycle
        (if it exists) in time $O(T(n))$. 
    \end{proposition}
    \begin{proof}
        The algorithm is recursive. Given any graph $G = (V, E)$, partition the vertices arbitrary into eight parts $V_1, \dots, V_8$, each of size at most $\lceil n/8\rceil$. 
        For each choice of $i_1, i_2, i_3, i_4 \in [8]$, we run $\mathcal A$ on the induced subgraph $G[V_{i_1}, V_{i_2}, V_{i_3}, V_{i_4}]$ to test if it contains an induced 4-cycle. If in none of these instances we find an induced 4-cycle, we can safely report that $G$ has  no induced 4-cycle.
        Otherwise, if we succeed for some $i_1, i_2, i_3, i_4$, then we recursively search for an induced 4-cycle in the subgraph $G[V_{i_1}, V_{i_2}, V_{i_3}, V_{i_4}]$. This is a graph on at most $(n/2) + O(1)$ nodes, so the search algorithm takes time $S(n) \le O(T(n)) + S(n/2 + O(1)) \le O(T(n))$ (since $T$ is subadditive).
    \end{proof}
\section{Detection on Clusters}
    \label{sec:cluster}

    Our algorithm for induced 4-cycle detection works by partitioning the vertex set of the input graph into collections of cliques of various sizes.
    We call these cliques \emph{clusters} in the graph.
    In this section, we present algorithms for finding induced 4-cycles with vertices contained in specific sets of clusters. 
    The subroutines we introduce here will later be combined with additional ideas to construct our final induced 4-cycle detection algorithm.
    
    In the rest of this section, a cluster simply refers to the vertex set of a clique in a graph,
    and the \emph{size} of a cluster is the number of vertices it contains.

    \subsection{Cluster Pairs}
    
    The notion of \emph{ordered clusters} is a key idea underlying our algorithm for detecting induced 4-cycles.
    Intuitively, two clusters are ordered if the edges between them determine nested neighborhoods. 

    \begin{definition}[Ordered Clusters]
        \label{def:ordered-clusters}
        We say two clusters $X$ and $Y$ are \emph{ordered} if there exist functions $f_{XY}\colon X\to \ZZ$ and $g_{XY}\colon Y\to\ZZ$ such that the pair $(x,y)\in X\times Y$ is an edge in the underlying graph precisely when $f_{XY}(x)\le g_{XY}(y)$.
        We refer to $f_{XY}$ and $g_{XY}$ as \emph{orderings} for the cluster pair $(X,Y)$.
        We say the orderings are \emph{concise} if they satisfy the additional property
        that for any choice of  $x,\tilde{x}\in X$, if $f_{XY}(x)\neq f_{XY}(\tilde{x})$ then 
        $N_Y(x)\neq N_Y(\tilde{x})$,
        and similarly for any   $y,\tilde{y}\in Y$, if $g_{XY}(y)\neq g_{XY}(\tilde{y})$ then $N_X(y)\neq N_X(\tilde{y})$.
    \end{definition}

    If a pair of clusters is ordered, then the adjacency information between the clusters can be succinctly represented using the orderings described above. 
    The following result shows that if a pair of clusters does not contain an induced 4-cycle, then those clusters must be ordered.

    \begin{lemma}[Detection on Cluster Pairs]
        \label{lem:cluster-pair}
        Given a graph $H$ on clusters $A$ and $B$ of sizes $s$ and $t$ respectively,
        there is an $O(st)$-time algorithm that either
        \begin{itemize}
            \item 
                reports that $H$ has an induced 4-cycle, or 
            \item 
                determines that $H$ has no induced 4-cycle,
                verifies that $A$ and $B$ are ordered, and returns concise orderings for the pair $(A,B)$ with range in $\set{0,\dots, s+1}$.
            \end{itemize}
    \end{lemma}
    \begin{proof}
        Define the function $g\colon B\to\ZZ$ by setting $g(b) = \deg_A(b)$ for all vertices $b\in B$.
        Then define the function $f\colon A\to\ZZ$ by setting 
            \begin{equation}
            \label{eq:fdef-fromg}
            f(a) = \min_{b\in N_B(a)} g(b)
            \end{equation}
        to be the minimum degree of a node $b\in B$ adjacent to $a$, for all vertices $a\in A$.
        We can compute the functions $f$ and $g$ in $O(st)$ time by going over the neighborhoods of each vertex in the graph. 
        If $a$ has no neighbors in $B$, we instead set $f(a) = s+1$ to be greater than $g(b)$ for all $b\in B$.      

        We now go over all pairs $(a,b)\in A\times B$ that are not edges in the graph,
        and check whether they all satisfy $f(a) > g(b)$. 
        This takes $O(st)$ time, because we spend constant time per pair in $A\times B$.

        Suppose first that we find some non-edge pair $(a,b)$ with $f(a) \le g(b)$.
        By \cref{eq:fdef-fromg}, this means there exists a neighbor $\tilde{b}\in B$
        of $a$ such that 
            \[\deg_A(\tilde{b}) = g(\tilde{b}) = f(a) \le g(b) = \deg_A(b).\]
        Since $a$ is adjacent to $\tilde{b}$ but not  to $b$,
        the above equation implies that $b$ has a neighbor $\tilde{a}\in A$ that is not adjacent to $\tilde{b}$.
        Consequently, in this case we can report that $H$ contains an induced 4-cycle, because  $(a,\tilde{b},b,\tilde{a})$ forms an induced 4-cycle in the graph.

        Otherwise, our procedure verifies that for all non-edges $(a,b)\in A\times B$
        we have $f(a)> g(b)$.
        By \cref{eq:fdef-fromg}, for every edge $(a,b)\in A\times B$ we have $f(a)\le g(b)$.  
        So by \Cref{def:ordered-clusters}, the functions $f$ and $g$ are valid orderings for the cluster pair $(A,B)$.

            We claim that in this case, $H$ has no induced 4-cycle.
            Suppose to the contrary that $H$ has an induced 4-cycle.
            Since $A$ and $B$ are cliques, and 
            an induced 4-cycle cannot contain a triangle,
            $H$ must have an induced 4-cycle with exactly two nodes in each of $A$ and $B$.
            The nodes within each cluster must be adjacent,
            so without loss of generality the induced 4-cycle is of the form
            $(a,\tilde{a},\tilde{b},b)$ for some $a,\tilde{a}\in A$ 
            and $b,\tilde{b}\in B$ such that (1) $a$ as adjacent to $b$ but not $\tilde{b}$,
            and (2) $\tilde{a}$ is adjacent to $\tilde{b}$ but not $b$.
            Condition (1) implies that $g(\tilde{b}) < f(a) \le g(b)$.
            This contradicts condition (2),
            which implies that $g(b) < f(\tilde{a}) \le g(\tilde{b})$.
            Thus $H$ cannot have an induced 4-cycle, as claimed.

        Finally, we prove that the orderings $f$ and $g$ are concise.
        Let $b,\tilde{b}\in B$ be vertices such that $g(b)\neq g(\tilde{b})$.
        Then $\deg_A(b) \neq \deg_A(\tilde{b})$,
        so $N_A(b)\neq N_A(\tilde{b})$.
        Similarly, 
        let $a,\tilde{a}\in A$ be vertices such that
        $f(a)\neq f(\tilde{a})$.
        Then \cref{eq:fdef-fromg} implies that $N_B(a)\neq N_B(\tilde{a})$.
        Thus the orderings are concise,
        and we can return $f_{AB} = f$ and $g_{AB} = g$.
    \end{proof}

    Given a collection of clusters in a graph, we can repeatedly apply \Cref{lem:cluster-pair} to check if any pair of these clusters contains an induced 4-cycle. 
    If we find no induced 4-cycle in this fashion, then \Cref{lem:cluster-pair} will have verified that the clusters are pairwise ordered, and provided us orderings for each cluster pair that certify this fact.
    This is a very strong condition that enables the design of fast algorithms on these clusters, because questions about adjacencies between clusters can be reduced to arithmetic comparisons of the outputs of their orderings.
    These comparisons can then be efficiently implemented using the range query data structure provided by \Cref{prop:range-search}.

        \subsection{Cluster Triples}

    The next natural step is to check if any triple of clusters contains an induced 4-cycle.
    The following observation helps with this.

    \begin{observation}[Comparable Neighborhoods in Triples]
        \label{obs:triple-comparable-neighborhood}
        Let $H$ be a graph on clusters $X,Y,Z$.
        Then $H$ contains an induced 4-cycle with exactly two
        nodes in $X$ and one node in each of $Y$ and $Z$ if and only if there exists an edge $(y,z)\in Y\times Z$ such that $N_X(y)$ and $N_X(z)$ are incomparable.
    \end{observation}
    \begin{proof}
        Suppose there exists an edge $(y,z)\in Y\times Z$ such that $N_X(y)$ and $N_X(z)$ are incomparable. 
        Then we can select distinct nodes $x_y\in N_X(y)\setminus N_X(z)$ and $x_z\in N_X(z)\setminus N_X(y)$.
        By definition, we get that $(y,x_y,x_z,z)$ forms an induced 4-cycle in $H$. 

        Conversely, suppose we are given $y\in Y$, $z\in Z$, and $x_1,x_2\in X$ forming an induced 4-cycle in $H$.
        Since $X$ is a cluster, $x_1$ and $x_2$ are adjacent.
        These nodes must have degree two in the 4-cycle,
         so $x_1$ and $x_2$ are each adjacent to unique, distinct nodes in $\set{y,z}$.
        This then implies that $N_X(y)$ and $N_X(z)$ are incomparable, because each of these neighborhoods contains a unique, distinct node from $\set{x_1, x_2}$.
        Finally, vertices $y$ and $z$ must have degree two in the induced 4-cycle, so $(y,z)$ is an edge.
        This proves the desired result.
    \end{proof}

    We now leverage \Cref{obs:triple-comparable-neighborhood} to efficiently detect induced 4-cycles on triples of clusters that are pairwise ordered. 

    \begin{lemma}[Detection on Cluster Triples]
        \label{lem:cluster-triple}
        Given a graph $H$ on  pairwise ordered clusters $A,B,C$ of size $s$ each, together with orderings for each cluster pair,
        we can determine in $\tilde{O}(s)$ time whether $H$ contains an induced 4-cycle.
    \end{lemma}
    
    \begin{proof}
        For each choice of distinct clusters $X,Y\in\set{A,B,C}$, we let $(f_{XY}, g_{XY})$ denote the orderings for the cluster pair $(X,Y)$, as in \Cref{def:ordered-clusters}.
        
        Following \Cref{obs:triple-comparable-neighborhood},
        we try to determine if there are
        adjacent nodes $b\in B$ and $c\in C$ whose neighborhoods $N_A(b)$ and $N_A(c)$ in $A$ are incomparable.
        To check this comparability condition,
        for each node $b\in B$ we will compute some thresholds $h_{\low}(b)$ and $h_{\high}(b)$.
        These thresholds will
        intuitively record information about the largest and smallest neighborhoods in $A$ from nodes in $C$ that ``sandwich'' $N_A(b)$.

        Formally, for each node $b\in B$ we define 
            \begin{equation}
                \label{eq:hlowhigh-def}
                h_\low(b) = \min_{a\in A\setminus N_A(b)} f_{AC}(a)
                    \quad
                    \text{and}
                    \quad
                h_\high(b) = \max_{a\in N_A(b)} f_{AC}(a).
            \end{equation}

        \begin{claim}
            \label{claim:hlowhigh-comp}
            We can compute $h_\low(b)$
            and $h_\high(b)$ for all $b\in B$ in $\tilde{O}(s)$
            time. 
            \end{claim}
        \begin{claimproof}
            Let $S\sub \ZZ^2$ be the set of points
                \[S = \set{\pair{f_{AB}(a), f_{AC}(a)}\mid a\in A}.\]
            By \Cref{prop:range-search},
            we can in $\tilde{O}(s)$ time insert all points of $S$ into a range query data structure.
            For each vertex $b\in B$, we binary search over the range of $f_{AC}(a)$ and make $O(\log s)$ queries to this data structure to find a vertex $a\in A$ that minimizes the value of $f_{AC}(a)$ subject to the condition  
                \[f_{AB}(a) > g_{AB}(b).\]
            By \Cref{def:ordered-clusters}, the above inequality holds precisely when $a$ is not adjacent to $b$. 
            Thus we can compute $h_\low(b) = f_{AC}(a)$
            for the vertex $a$ returned by this procedure
            (if the data structure reports that no $a\in A$ satisfies the above inequality, we instead set $h_\low(a) = \infty$).
            This takes $\tilde{O}(s)$ time because we make $O(\log s)$ queries for each of the $s$ vertices in $B$.

            Similar reasoning lets us compute the $h_\high(b)$ values in the same time bound.
        \end{claimproof}

        We apply \Cref{claim:hlowhigh-comp} to compute $h_\low(b)$ and $h_\high(b)$ values for all $b\in B$ in $\tilde{O}(s)$ time.
        The next claim shows how we can use these values to check for incomparable neighborhoods.

        \begin{claim}
            \label{claim:triple:low-high-sandwich}
            For any nodes $b\in B$ and $c\in C$,
            the neighborhoods $N_A(b)$ and $N_A(c)$
            are incomparable if and only if $h_\low(b) \le g_{AC}(c) < h_{\high}(b)$.
        \end{claim}
        \begin{claimproof}
            By \Cref{def:ordered-clusters}, the inclusion $N_A(b)\sub N_A(c)$ holds
            precisely when every $a\in N_A(b)$ satisfies
            $f_{AC}(a) \le g_{AC}(c)$.
            Then by \cref{eq:hlowhigh-def}, this inclusion is equivalent to $h_\high(b) \le g_{AC}(c)$.

            Similar reasoning shows that $N_A(c)\sub N_A(b)$
            is equivalent to $g_{AC}(c) < h_\low(b)$.

            Since $N_A(b)$
            and $N_A(c)$ are incomparable if and only if 
            neither of the inclusions $N_A(b)\sub N_A(c)$
            or $N_A(c)\sub N_A(b)$ holds,
            the desired result follows.
        \end{claimproof}

        Let $S\sub \ZZ^2$ be the set of points 
            \[S = \set{\pair{g_{AC}(c), g_{BC}(c)}\mid c\in C}.\]
        By \Cref{prop:range-search}, we can in $\tilde{O}(s)$ time
        insert all points of $S$ into a range query data structure.
        For each vertex $b\in B$, we query this data structure to determine if there exists $c\in C$ such that 
            \begin{equation}
                \label{eq:triple:comp-in-A}
                    h_\low(b) \le g_{AC}(c) < h_\high(b)
            \end{equation}
            and
            \begin{equation}
                \label{eq:triple:bc-edge}
                    f_{BC}(b)\le g_{BC}(c).
            \end{equation}

        If $c\in C$ satisfying \cref{eq:triple:comp-in-A,eq:triple:bc-edge} exists, we report the graph has an induced 4-cycle.
        If no such $c$ exists for any $b\in B$,
        we claim there is no induced 4-cycle with two nodes in $A$ and one node in each of $B$ and $C$.
        Indeed, by \Cref{claim:triple:low-high-sandwich}, the inequality from  \cref{eq:triple:comp-in-A}
        holds if and only if $N_A(b)$ and $N_A(c)$ are incomparable.
        By \Cref{def:ordered-clusters}, 
        the inequality from \cref{eq:triple:bc-edge} holds if and only if  $(b,c)$ is an edge.
        Thus by \Cref{obs:triple-comparable-neighborhood},
        our procedure correctly detects
        if the graph has an induced 4-cycle with two nodes in $A$ and one node in each of $B$ and $C$.
        This process takes $\tilde{O}(s)$ time,
        because we make one query for each for each of the $s$ vertices in $B$.

        By symmetric reasoning, we can in $\tilde{O}(s)$ time determine if the graph contains an induced 4-cycle with exactly two nodes in $B$, or exactly two nodes in $C$.
        An induced 4-cycle in the graph cannot have three nodes in a single part from $\set{A,B,C}$, because $A,B,C$ are cliques and a 4-cycle does not contain a triangle.
        Thus if we have not found an induced 4-cycle after performing the above checks, we can report that the graph contains no induced 4-cycle. 
    \end{proof}

            So far, we have seen that given a collection of clusters in a graph, we can repeatedly apply \Cref{lem:cluster-pair} to either find an induced 4-cycle on some pair of the clusters, or obtain orderings for all cluster pairs.
        In the latter case, we can then repeatedly apply \Cref{lem:cluster-triple} to determine if some triple of clusters contains an induced 4-cycle. 
        If we find an induced 4-cycle in this way, then we have successfully solved our problem.
        If we find no such induced 4-cycles, then we would like to use this lack of 4-cycles to infer  additional structural properties about edges between clusters, that then could help us check for induced 4-cycles among \emph{quadruples} of clusters.

        \Cref{lem:cluster-pair} shows that if a pair of clusters has no induced 4-cycle, then that pair is ordered. 
        We build off this structural characterization,
        and show that if a triple of clusters does not contain an induced 4-cycle, then not only is it the case that the clusters are pairwise ordered, but the orderings between them are strongly \emph{correlated}.

        More precisely, suppose we have pairwise ordered clusters $W,X,Z$ with concise orderings 
        $f_{AB}, g_{AB}$
        for each pair $(A,B)$ with  $A,B\in\set{W,X,Z}$.
        By \Cref{def:ordered-clusters},
        the neighborhood in $Z$ of a vertex from $X$ always takes the form
            \[\set{z\in Z\mid g_{XZ}(z) \ge \zeta_\tsuff}\]
        for some integer $\zeta_\tsuff$.
        Our next result shows that, among other properties,
        if we have the additional constraint that the graph on $W,X,Z$ has no induced 4-cycle,
        then the neighborhood in $Z$ of any vertex from $W$ always takes the form
        \[\set{z\in N_Z(w)\mid g_{XZ}(z) = \zeta_\low}
                            \sqcup
                                \set{z\in Z\mid
                                    g_{XZ}(z) \ge \zeta_\tsuff}
        \]
        for some integers $\zeta_\low$ and $\zeta_\tsuff$.
        In other words, 
            the ordering $g_{XZ}$, defined initially only in terms of the edges between clusters $X$ and $Z$,
        also controls adjacencies between clusters $W$ and $Z$.
        Moreover, the structure of neighborhoods from $W$ to $Z$ is 
        is almost the same as the structure of neighborhoods from $X$ to $Z$.
        The only difference is that the former is parameterized by an extra integer $\zeta_\low$, and the associated neighborhood may contain a proper subset of nodes $z\in Z$ with $g_{XZ}(z) = \zeta_\low$ (in comparison, for any neighborhood from $X$ to $Z$ and integer $\zeta$, either the neighborhood contains all nodes $z\in Z$ with $g_{XZ}(z) = \zeta$, or none of them).

    \begin{lemma}[Correlated Neighborhoods]
        \label{lem:splitconn:computation}
        Let $W,X,Z$ be pairwise ordered clusters each of size $s$, such that the graph on these clusters does not contain an induced 4-cycle. 
        Then, given concise orderings 
        $f_{AB}, g_{AB}$
        for each cluster pair $(A,B)$ for  $A,B\in\set{W,X,Z}$, we can in $\tilde{O}(s)$ time
        compute for each vertex $w\in W$ with nonempty neighborhoods $N_X(w)$ and $N_Z(w)$, a vector 
        
            \[\vec{w} = 
                \pair{\xi_\tpre, 
                    \xi_\high,
                    \zeta_\low,
                    \zeta_\tsuff}
                \in\grp{\ZZ\cup\set{-\infty,\infty}}^4\]
        with the property that if
         $\xi_\high > \zeta_\low$,
        then

            \begin{equation}
                \label{eq:neighborhood-interval-structure}
                \begin{cases}
                    N_X(w) = \set{x\in X\mid
                                    f_{XZ}(x) \le \xi_\tpre}
                            \sqcup
                                \set{x\in N_X(w)\mid
                                    f_{XZ}(x) = \xi_\high}
                        \\
                    N_Z(w) = 
                                \set{z\in N_Z(w)\mid g_{XZ}(z) = \zeta_\low}
                            \sqcup
                                \set{z\in Z\mid
                                    g_{XZ}(z) \ge \zeta_\tsuff}
                \end{cases}
            \end{equation}
        where $\xi_\tpre\in \im(f_{XZ})\cup\set{-\infty}$
        and $\zeta_\tsuff\in \im(g_{XZ})\cup\set{\infty}$.
    \end{lemma}

    \begin{proof}
        First, 
        for each vertex $w\in W$ with nonempty neighborhood $N_X(w)$,
        we compute the largest integer $\xi_{\high} = \xi_{\high}(w)$ such that $w$ is adjacent to a node $x\in X$ with $f_{XZ}(x) = \xi_{\high}$.

        To do this, let $S\sub\ZZ^2$ be the set of points 
            \begin{equation}
            \label{eq:block-boundary-set}
            S = \set{\pair{f_{XZ}(x), g_{WX}(x)} \mid x\in X}.
            \end{equation}
        By \Cref{prop:range-search}, we can in $\tilde{O}(s)$ time insert all points of $S$ into a range query data structure.
        For each $w\in W$, we then make $O(\log s)$ queries to this structure by binary searching over
        the range of $f_{XZ}$,
        to find the largest integer $\xi_{\high}$ for which there exists $x\in X$ with
            \begin{equation}
            \label{eq:splitconn:lowhigh}
              f_{XZ}(x) = \xi_{\high}
            \end{equation}
            and
            \begin{equation}
            \label{eq:splitconn:wx-edge}
                f_{WX}(w) \le g_{WX}(x).
            \end{equation}

        By \Cref{def:ordered-clusters}, the inequality from \cref{eq:splitconn:wx-edge} holds precisely when $x\in N_X(w)$.
        So this procedure correctly identifies the maximum value
        $\xi_\high= \xi_\high(w)$ such that $N_X(w)$ has 
        a node $x$ satisfying $f_{XZ}(x) = \xi_\high$.
        Moreover, this takes $\tilde{O}(s)$ time overall because we make $O(\log s)$ queries for each of the $s$ nodes in $W$.
    
        By similar reasoning,
        we compute for each vertex $w\in W$ with nonempty neighborhood $N_Z(w)$ the smallest integer 
        $\zeta_\low = \zeta_\low(w)$ such that
        $w$ is adjacent to a node $z\in Z$
        with $g_{XZ}(z) = \zeta_\low$,
        spending only $\tilde{O}(s)$ time overall.

        Now, take arbitrary $w\in W$ such that $N_X(w)$ and $N_Z(w)$ are both nonempty.
        If 
            \[\xi_\high(w)\le \zeta_\low(w)\]
        then we set $\vec{w} = \pair{-\infty,\xi_\high(w),\zeta_\low(w),\infty}$.

        Otherwise, we have 
            \[\xi_\high(w) > \zeta_\low(w).\]

        In this case, we infer additional structure concerning $N_X(w)$ and $N_Z(w)$ using the assumption that the graph induced on the clusters $W,X,Z$ has no induced 4-cycle.

        \begin{claim}[Comparable Neighborhoods]
            \label{claim:triple-contrapositive}
            For any edge $(w,x)\in W\times X$, the neighborhoods $N_Z(w)$ and $N_Z(x)$ are comparable.
            Similarly, for any edge $(w, z)\in W\times Z$, the neighborhoods $N_X(w)$ and $N_X(z)$ are comparable. 
        \end{claim}
        \begin{claimproof}
            This follows immediately by combining \Cref{obs:triple-comparable-neighborhood} with the assumption that there is no induced 4-cycle on the cluster triple $(W,X,Z)$.
        \end{claimproof}

        In what follows, 
        fix $w\in W$,
        and
        abbreviate $\xi_\high = \xi_\high(w)$ and $\zeta_\low = \zeta_\low(w)$.

        By the definition of the index $\zeta_\low$, there exists a vertex $z\in N_Z(w)$
        with $g_{XZ}(z) = \zeta_\low$.
        Applying \Cref{claim:triple-contrapositive} to the adjacent nodes $w\in W$ and $z\in Z$, we get that $N_X(w)$ and $N_X(z)$ are comparable.
        However, by \Cref{def:ordered-clusters} we have 
            \[N_X(z) = \set{x\in X\mid f_{XZ}(x) \le \zeta_\low}.\]
        By assumption, $N_X(w)$ contains a vertex $x$
        with $f_{XZ}(x) = \xi_\high > \zeta_\low$.
        So in fact we must have 
            \begin{equation}
            \label{eq:neighborhood-has-Xprefix}
            \set{x\in X\mid f_{XZ}(x) \le \zeta_\low} = N_X(z)\subset N_X(w).
            \end{equation}

        Similarly, by definition of $\xi_\high$, there exists $x\in N_X(w)$
        with $f_{XZ}(x) = \xi_\high$.
        Applying \Cref{claim:triple-contrapositive} to the adjacent nodes $w\in W$ and $x\in X$, we get that $N_Z(w)$ and $N_Z(x)$ are comparable.
        However, by \Cref{def:ordered-clusters} we have 
            \[N_Z(x) = \set{z\in Z\mid g_{XZ}(z) \ge \xi_\high}.\]
        By assumption, $N_Z(w)$ contains a vertex $z$ with and $g_{XZ}(z) = \zeta_\low < \xi_\high$.
        So we must have 
            \begin{equation}
                \label{eq:neighborhood-has-Zsuffix}
            \set{z\in Z\mid g_{XZ}(z) \ge \xi_\high} = N_Z(x)\subset N_Z(w).
            \end{equation}

        Let $\xi_{\tmed}$ be the smallest integer in $\im(f_{XZ})$ 
        that is greater than $\zeta_{\low}$.
        This value is well-defined since $\zeta_\low < \xi_{\high}$.
        To further characterize the neighborhoods of $w$ in $X$ and $Z$,
        we perform casework based off whether $w$ has a neighbor $x\in X$ with $f_{XZ}(x) = \xi_{\tmed}$.

        \medskip
        \noindent \stepfont{Case 1: Avoiding Intermediate Values}
        
            Suppose first that $N_X(w)$
            does not contain any $x\in X$ with $f_{XZ}(x) = \xi_{\tmed}$.

            We prove that this assumption constrains the possible values
            $g_{XZ}$ takes on for $z\in N_Z(w)$.

            \begin{claim}[Avoiding $Z$ Values]
                \label{claim:split:empty-neighborhoods:Z}
                In \stepfont{case 1}, the vertex $w$ is not adjacent to any node $z\in Z$ satisfying the inequality $\zeta_{\low} <  g_{XZ}(z) < \xi_{\high}$.
            \end{claim}
            \begin{claimproof}
                Suppose to the contrary 
                that there exists $z\in N_Z(w)$
                such that $\zeta_{\low} < g_{XZ}(z) < \xi_{\high}$.
                
                Set $\zeta_\tmed = g_{XZ}(z)$.
                Applying \Cref{claim:triple-contrapositive} to the adjacent nodes $w\in W$ and $z\in Z$, we get that $N_X(w)$ and $N_X(z)$ are comparable. 
                However, by \Cref{def:ordered-clusters} we have 
                    \begin{equation}
                    \label{style:XZ-zeta-tmed}
                    N_X(z) = \set{x\in X\mid f_{XZ}(x) \le \zeta_\tmed}.
                    \end{equation}
                    
                Since the orderings $f_{XZ}$ and $g_{XZ}$
                are concise,
                if we sort the images of $f_{XZ}$ and $g_{XZ}$
                into a single list, the outputs of $f_{XZ}$ and $g_{XZ}$ must alternate.
                In particular, since $\zeta_\tmed$ and $\xi_\tmed$ are the smallest outputs of $g_{XZ}$ and $f_{XZ}$ greater than $\zeta_\low$ respectively, and $\zeta_\low$ is an output of $g_{XZ}$,
                we must have 
                    \begin{equation}
                    \label{style:xi-zeta-tmed-concise}
                    \zeta_\low < \xi_{\tmed} \le \zeta_{\tmed}.
                    \end{equation}

                By the case assumption, $N_X(w)$ has no vertices in
                $x\in X$ with $f_{XZ}(x) = \xi_{\tmed}$.
                Combining this with 
                \cref{style:XZ-zeta-tmed,style:xi-zeta-tmed-concise}
                and the fact that $\xi_{\tmed}$ is in the image of $f_{XZ}$,
                we deduce that $N_X(z)$ contains a vertex not  in $N_X(w)$.
                On the other hand, by assumption $N_X(w)$ has a node $x$
                satisfying $f_{XZ}(x) = \xi_{\high}$.
                Since $\xi_{\high} > \zeta_{\tmed}$,
                by \cref{style:XZ-zeta-tmed} this node
                cannot appear in $N_X(z)$.
                Thus $N_X(w)$ and $N_X(z)$ are incomparable.
                This contradicts \Cref{claim:triple-contrapositive}.
                Thus our initial assumption was wrong and the desired result holds. 
            \end{claimproof}

            \begin{claim}[Avoiding $X$ Values]
                \label{claim:split:empty-neighborhoods:X}
                In \stepfont{case 1}, 
                the node $w$ is not adjacent to any node $x\in X$
                satisfying the inequality
                $\zeta_\low < f_{XZ}(x) < \xi_\high$.
            \end{claim}
            \begin{claimproof}
                This follows by symmetric reasoning to the proof of \Cref{claim:split:empty-neighborhoods:Z}.
            \end{claimproof}

        In this case, we define $\xi_\tpre = \zeta_\low$ and $\zeta_\tsuff = \xi_\high$, and set 
        \[
        \vec{w} = 
                \pair{\xi_\tpre, 
                    \xi_\high,
                    \zeta_\low,
                    \zeta_\tsuff}.
        \]

        From the definitions of $\xi_\high$ and $\zeta_\low$,
        we know that the neighborhood $N_X(w)$ only contains $x\in X$ with $f_{XZ}(x) \le \xi_\high$,
        and $N_Z(w)$ only contains $z\in Z$ with $g_{XZ}(z) \ge \zeta_\low$.
        By \cref{eq:neighborhood-has-Xprefix,eq:neighborhood-has-Zsuffix},
        we know that $N_X(w)$ contains all  
        $x\in X$
        with $f_{XZ}(x) \le \xi_\tpre$, and $N_Z(w)$ contains all $z\in Z$
        with $g_{XZ}(z) \ge \zeta_\tsuff$.
        Combining these observations together with
        \Cref{claim:split:empty-neighborhoods:Z,claim:split:empty-neighborhoods:X},
        we see that \cref{eq:neighborhood-interval-structure} holds for our choice of $\vec{w}$.

        \medskip
        \noindent \stepfont{Case 2: Connected Neighborhoods}
        
            Suppose instead that $N_X(w)$
            has a vertex $x\in X$ with $f_{XZ}(x) = \xi_\tmed$.

            We prove that this assumption forces 
             $N_Z(w)$ to contain many addition vertices in $Z$.

            \begin{claim}[Capturing $Z$ Values]
                \label{claim:conn:Zneighbors}
                In \stepfont{case 2}, we have 
                $\set{z\in Z\mid g_{XZ}(z) > \zeta_\low}\subset N_Z(w)$.
            \end{claim}
            \begin{claimproof}
                By the case assumption, there exists a vertex $x\in N_X(w)$
                with $f_{XZ}(x) = \xi_\tmed$.

                By \Cref{def:ordered-clusters} we have 
                    \[N_Z(x) = \set{z\in Z\mid g_{XZ}(z) \ge \xi_{\tmed}}.\]
                By definition of $\zeta_\low$,
                the neighborhood $N_Z(w)$
                has a node $z$ with $g_{XZ}(z) = \zeta_\low$.
                applying \Cref{claim:triple-contrapositive} to the adjacent nodes $w\in W$ and $x\in X$, we get that $N_Z(w)$ is comparable to $N_Z(x)$.
                Since $\xi_{\tmed}$ is the smallest integer larger than $\zeta_\low$ in the image of $f_{XZ}$,
                combining these observations together
                with the above equation 
                implies that 
                    \[\set{z\in Z\mid g_{XZ}(z) > \zeta_\low} = N_Z(x)\subset N_Z(w)\]
                as claimed. 
            \end{claimproof}

            \begin{claim}[Capturing $X$ Values]
                \label{claim:conn:Xneighbors}
                In \stepfont{case 2}, we have 
                $\set{x\in X\mid f_{XZ}(x) < \xi_\high} \subset N_X(w)$.
            \end{claim}
            \begin{claimproof}
                This follows by symmetric reasoning to the proof of 
                \Cref{claim:conn:Zneighbors}.
            \end{claimproof}

        In this case, we 
        define $\xi_\tpre$ to be the largest integer less than $\xi_\high$ in the image of $f_{XZ}$
        and $\zeta_\tsuff$
        to be the smallest integer greater than $\zeta_\low$ in the image of $g_{XZ}$, and then set 
        \[
        \vec{w} = 
                \pair{\xi_\tpre, 
                    \xi_\high,
                    \zeta_\low,
                    \zeta_\tsuff}.
        \]
            By combining the definitions of $\xi_\high$ and $\zeta_\low$ with \Cref{claim:conn:Zneighbors,claim:conn:Xneighbors}, we see that \cref{eq:neighborhood-interval-structure} holds for this choice of $\vec{w}$.

            At this point we have defined for every vertex $w\in W$ such that $N_X(w)$ and $N_Z(w)$ are nonempty, a vector $\vec{w}$ satisfying the conditions of the lemma statement. 
            We have also proved that in $\tilde{O}(s)$ time
            we can compute 
            the second and third coordinates $\xi_\high(w)$ and $\zeta_\low(w)$ for all of these vectors.
            It remains to show how we compute the 
            first and final coordinates of each vector.

            If $\xi_\high(w)\le \zeta_\low(w)$, then we already said we set the first coordinate $\xi_\tpre(w) = -\infty$
            and the third coordinate $\zeta_\tsuff(w) = \infty$.
            This takes $O(s)$ time for all vertices $w$ in this case.

            Otherwise, $\xi_\high(w) > \zeta_\low(w)$.
            In this situation, \cref{eq:neighborhood-interval-structure} 
            shows that the first coordinate 
            of $\vec{w}$ should be equal to the largest integer $\xi_\tpre$ less than $\xi_\high$ such that $w$ 
            is adjacent to a node
            $x\in X$
            with $f_{XZ}(x) = \xi_\tpre$.
            To compute this value, let $S\sub \ZZ^2$ be the set of points
            defined in \cref{eq:block-boundary-set}.
            At the beginning of this proof,
            we already inserted the points of $S$ into a range query data structure, following \Cref{prop:range-search}.
            For each $w\in W$,
            we then make $O(\log s)$ queries to this structure by binary searching over the range of $f_{XZ}$, restricted to outputs less than $\xi_\high$,
            to find the largest integer $\xi_{\tpre} < \xi_\high$ for which there exists $x\in X$ such that 
                \[f_{XZ}(x) = \xi_{\tpre}\]
            and
                \[f_{WX}(w) \le g_{WX}(x).\]
            If the structure reports that no vertex $x\in X$ satisfies the above conditions for any $\xi_{\tpre}$ in the range of $f_{XZ}$ with $ \xi_{\tpre} < \xi_{\high}$,
            then we set $\xi_{\tpre} = -\infty$.

            By \Cref{def:ordered-clusters}, the above inequality holds precisely when $x\in N_X(w)$.
            Thus, by the discussion in the previous paragraph, this procedure correctly identifies the value $\xi_\tpre$ in the first entry of $\vec{w}$.
            This takes $\tilde{O}(s)$ time because we make $O(\log s)$ queries  for each of the $s$ nodes in $W$.        

        Similar reasoning lets us compute the final entry $\zeta_\tsuff$
        of each $\vec{w}$ in $\tilde{O}(s)$ time overall.
    \end{proof}

    \subsection{Cluster Quadruples}

    We now employ the neighborhood structure enforced by \Cref{lem:splitconn:computation} to efficiently detect whether a collection of four pairwise ordered clusters contains an induced 4-cycle.  

    \begin{lemma}[Detection on Cluster Quadruples]
        \label{lem:cluster-quadruple}
        Given a graph $H$ on four pairwise ordered clusters $A,B,C,D$ of size $s$ each, together with concise orderings for each cluster pair,
        we can determine in $\tilde{O}(s)$ time whether $H$ contains an induced 4-cycle.
    \end{lemma}
    \begin{proof}
        First, run the algorithm of \Cref{lem:cluster-triple} on the cluster triples $(A,B,C)$, $(B,C,D)$, $(D,C,A)$, and $(C,A,D)$.
        This takes $\tilde{O}(s)$ time.
        If the algorithm ever returns an induced 4-cycle, we report the graph has an induced 4-cycle.
        Otherwise, if the algorithm reports none of the triples have an induced 4-cycle, we have verified that any induced 4-cycle in the graph must have exactly one node from each of the clusters $A,B,C,D$.

        For each choice of distinct clusters $X,Y\in\set{A,B,C,D}$, let $f_{XY}$ and $g_{XY}$ denote the provided 
        concise
        orderings for the pair $(X,Y)$.

        Run the algorithm from \Cref{lem:splitconn:computation} on the triples $(A,B,D)$ and $(C,B,D)$.
        This takes $\tilde{O}(s)$ time.
        Since the triples $(A,B,D)$ and $(B,C,D)$ do not have induced 4-cycles, this algorithm computes vectors $\vec{w}$ for all $w\in A\sqcup C$ with neighbors in both $B$ and $D$, that satisfy 
        the conditions
        from the statement of \Cref{lem:splitconn:computation}.

        Our goal is to determine if there exists an induced 4-cycle using exactly one node from each of $A,B,C,D$.
        If a vertex $w\in A\sqcup C$ does not have neighbors to $B$ or $D$, it cannot participate in such a 4-cycle.
        Thus, we may restrict our attention to nodes $w\in A\sqcup C$ that have neighbors in both $B$ and $D$.
        These are precisely the nodes for which we have computed vectors.

        Without loss of generality, it suffices to check if the graph contains an induced 4-cycle of the form $(a,b,c,d)\in A\times B\times C\times D$ such that 
        \begin{equation}
            \label{eq:quadruple:edge}
            \grp{a,b}, \grp{b,c}, \grp{c,d}, \grp{d,a}\text{ are  edges,} 
        \end{equation} 
            while 
        \begin{equation}
            \label{eq:quadruple:non-edge}
            \grp{a,c}\text{ and }\grp{b,d}\text{ are non-edges.}
        \end{equation}
        This is because if we can perform this check, then we can rearrange the order of the clusters $A,B,C,D$ and employ symmetric reasoning to check for any possible induced 4-cycle.

        Suppose vertex $w\in A\sqcup C$ has vector $\vec{w} = \pair{\xi_\tpre, \xi_\high, \zeta_\low, \zeta_\tsuff}$ with $\xi_\high \le \zeta_\low$.
        Then \cref{eq:neighborhood-interval-structure} together with 
        \Cref{def:ordered-clusters} shows that all nodes in $N_B(w)$ and $N_D(w)$ are adjacent to one another.
        Thus by \cref{eq:quadruple:non-edge}, the vertex $w$ cannot participate in an induced 4-cycle.

        The previous paragraph shows that for the purpose of detecting induced 4-cycles,
        we may restrict our attention to vertices $w\in A\sqcup C$
        with vectors $\vec{w} = \pair{\xi_{\tpre}, \xi_\high, \zeta_\low, \zeta_{\tsuff}}$ such that $\xi_\high > \zeta_\low$.
        By \Cref{lem:splitconn:computation},
        such vertices $w$ must satisfy \cref{eq:neighborhood-interval-structure} for $W = A$, $X = B$, and $Z = D$.
        Intuitively, 
        \cref{eq:neighborhood-interval-structure}
        shows that for any relevant node $w\in A\sqcup C$, the neighborhoods of $w$ in $B$ 
        can be decomposed into the disjoint union of a \emph{full prefix} 
        $\set{b\in B\mid f_{BD}(b) \le \xi_{\tpre}}$ consisting of all nodes in $b\in B$ with small rank with respect to the $(B,D)$ orderings,
        and an \emph{extreme layer}
        $\set{b\in N_B(w)\mid f_{BD}(b) = \xi_\high}$
        consisting of all the nodes in $N_B(w)$ with the largest possible rank.
        In a similar fashion,
        \cref{eq:neighborhood-interval-structure} also demonstrates that the neighborhood of $w$ in $D$ can 
        be decomposed into the disjoint union of a \emph{full suffix}
        $\set{d\in D\mid g_{BD}(d) \ge \zeta_\tsuff}$
        and an \emph{extreme layer}
        $\set{d\in N_D(w)\mid g_{BD}(d) = \zeta_\low}$.

        Using this structure, we seek an induced 4-cycle $(a,b,c,d)$ satisfying the conditions from \cref{eq:quadruple:edge,eq:quadruple:non-edge}.
        For candidate vertices $a\in A$ belonging to this 4-cycle, we perform this search by casework on whether the  nodes
        $b\in B$ and $d\in D$ in the 4-cycle come from the extreme layers or not. 
        For each $w\in A\sqcup C$, we write
                    \[\vec{w} = \pair{\xi_\tpre(w), \xi_\high(w), \zeta_\low(w),\zeta_\tsuff(w)}\]
                for convenience. 

            \medskip

            \noindent\stepfont{Case 1: Extreme Layers}

                In this case, we check if the graph has an induced 4-cycle with vertices $a\in A$, $b\in B$, and $d\in D$ such that at least one of $b$ or $d$ belongs to an extreme layer 
                of the neighborhood of $a$.
                Without loss of generality, suppose that $b$ is in the extreme layer of $N_B(a)$.
                That is, we seek a solution where $f_{BD}(b) = \xi_\high(a)$.
                In this case, the node $d\in D$ in the induced 4-cycle is not adjacent to $b$ if and only if $g_{BD}(d) < \xi_\high(a)$.
                
                To find an induced 4-cycle meeting these conditions, we first compute for each node $a\in A$ some thresholds $\beta_a$ and $\delta_a$ that intuitively identify the node $b$ in the extreme layer of $N_B(a)$
                and the node $d$ in $N_D(a)$ not adjacent to $b$ 
                that each have the largest possible neighborhoods in $C$.

                For convenience, define the extreme layer sets
                    \begin{equation}
                        \label{eq:extreme-Ba}
                        \extreme_B(a) = \set{b\in N_B(a)\mid f_{BD}(b) = \xi_\high(a)}
                    \end{equation}
                and the sets of relevant neighbors in $D$ that are not adjacent to nodes in extreme layers
                    \begin{equation}
                        \label{eq:safe-Da}
                        \safe_D(a) = \set{d\in N_D(a)\mid g_{BD}(d) < \xi_\high(a)}.
                    \end{equation} 
                Then we define 
                    \begin{equation}
                        \label{eq:quadruple:beta-delta-def}
                        \beta_a = \min_{b\in\extreme_B(a)} f_{BC}(b)\quad\text{and}\quad
                        \delta_a = \max_{d\in \safe_D(a)} g_{CD}(d).
                    \end{equation}

                \begin{claim}
                    \label{claim:quadruple:block-expansion}
                    We can compute $\beta_a$ and $\delta_a$ for all $a\in A$ in $\tilde{O}(s)$ time.
                \end{claim}
                \begin{claimproof}
                    Let $S\sub\ZZ^3$ be the set of points 
                        \[S = \set{\pair{g_{AB}(b), f_{BD}(b), f_{BC}(b)}\mid b\in B}.\]
                    By \Cref{prop:range-search},
                    we can in $\tilde{O}(s)$ time insert all points of $S$ into a range query data structure.
                    For each vertex $a\in A$, we 
                    binary search over the range of $f_{BC}$ and
                    make $O(\log s)$ queries to this data structure to find a vertex $b\in B$ that minimizes the value of $f_{BC}(b)$ subject to the constraints that 
                    \begin{equation}
                        \label{eq:quadruple:ab-edge}
                        f_{AB}(a)\le g_{AB}(b)
                    \end{equation}
                    and
                    \begin{equation}
                        \label{eq:quadruple:B-block-identification}
                       f_{BD}(b) = \xi_\high(a).
                    \end{equation}

                    By \Cref{def:ordered-clusters}, the inequality from \cref{eq:quadruple:ab-edge} holds if and only if $b\in N_B(a)$.
                    Hence by \cref{eq:extreme-Ba},
                    we have $b\in \extreme_B(a)$ if and only if 
                    both  \cref{eq:quadruple:ab-edge,eq:quadruple:B-block-identification} hold.
                    Then by \cref{eq:quadruple:beta-delta-def},
                    we can compute $\beta_a$ as $f_{BC}(b)$ for the vertex $b$ obtained by this procedure.
                    This takes $\tilde{O}(s)$ time because we make $O(\log s)$ queries for each of the $s$ vertices in $A$. 

                    Similar reasoning lets us compute all the $\delta_a$ values in the same time bound.
                \end{claimproof}

            Run the algorithm from \Cref{claim:quadruple:block-expansion} to compute $\beta_a$ and $\delta_a$ for all $a\in A$ in $\tilde{O}(s)$ time.

            Now, let $S\sub \ZZ^3$ be the set of points 
                \[S = \set{\pair{g_{AC}(c), g_{BC}(c), f_{CD}(c)}\mid c\in C}.\]
            By \Cref{prop:range-search}, we can in $\tilde{O}(s)$ time insert the points of $S$ into a range query data structure.
            For each vertex $a\in A$, we query this data structure to determine if there exists $c\in C$ such that
                \begin{equation}
                    \label{eq:quadruple:ac-nonedge}
                    f_{AC}(a) > g_{AC}(c)
                \end{equation}
            and
                \begin{equation}
                    \label{eq:quadruple:c-common-neighbors}
                    \beta_a \le g_{BC}(c)\quad\text{and}\quad 
                    f_{CD}(c)\le \delta_a.
                \end{equation}

            If $c\in C$ satisfying \cref{eq:quadruple:ac-nonedge,eq:quadruple:c-common-neighbors} exists, we report that the graph has an induced 4-cycle.
            If no such $c\in C$ exists, we claim there is no induced 4-cycle satisfying the case assumptions.

            This works, because
            \Cref{def:ordered-clusters} shows that the inequality from \cref{eq:quadruple:ac-nonedge} holds precisely when $\grp{a,c}$ is not an edge. 
            \Cref{def:ordered-clusters} together with the definitions of $\beta_a$ and $\delta_a$ from \cref{eq:quadruple:beta-delta-def} shows that \cref{eq:quadruple:c-common-neighbors} holds precisely when $a$ and $c$ have common neighbors in
            the extreme layer $\extreme_B(a)$ and set of non-adjacent nodes $\safe_D(a)$.
            Finally, from \Cref{def:ordered-clusters} and 
            the definitions of these sets in \cref{eq:extreme-Ba,eq:safe-Da},
            we see that no vertex in $\extreme_B(a)$ is adjacent to any vertex in $\safe_D(a)$.

            This proves that if for some $a\in A$ there exists $c\in C$ satisfying \cref{eq:quadruple:ac-nonedge,eq:quadruple:c-common-neighbors}, then we  can pick $b\in \extreme_B(a)$ and $d\in \safe_D(a)$ adjacent to both $a$ and $c$ and obtain an induced 4-cycle $(a,b,c,d)$ as claimed.
            If instead no such $c\in C$ exists, it means that for every vertex $a\in A$, no $c\in C$ not adjacent to $a$ can have common neighbors with $a$ in $\extreme_B(a)$ and $\safe_D(a)$ simultaneously, so $a$ cannot be extended to an induced 4-cycle meeting the case assumptions.

            \medskip
            
            \noindent\stepfont{Case 2: Full Prefix and Suffix}

            It remains to check if the graph contains an induced 4-cycle with nodes $a\in A$, $b\in B$, and $d\in D$ such that $b$ and $d$  come from the full prefix 
            of $N_B(a)$ and full suffix of $N_D(a)$ respectively. 
            We may furthermore assume that the vertex $c\in C$ participating in the prospective induced 4-cycle we seek has the property that $b$ and $d$ come from the full prefix and full suffix of its neighborhoods $N_B(c)$ and $N_C(d)$ respectively.
            This is because if this were not the case, then at least one of $b$ or $d$ would come from the extreme layer of a neighborhood of $c$, and we could apply symmetric reasoning to the argument in \stepfont{case 1} to find the induced 4-cycle in this scenario.

            By definition, for any $w\in A\sqcup C$, 
            the full prefix in $N_B(w)$ is
                \[\set{b\in B\mid f_{BD}(b) \le \xi_\tpre(w)}\]
            and the full suffix in $N_D(w)$ is 
                \[\set{d\in D\mid g_{BD}(d)\ge \zeta_\tsuff(w)}.\]

            We may restrict our attention in this case to $w\in A\sqcup C$ with 
            \[\xi_\tpre(w)\in \im(f_{BD})\quad\text{and}\quad
            \zeta_\tsuff(w)\in \im(g_{BD}).\]
            This is because if this did not hold,
             by \Cref{lem:splitconn:computation}
            we would have $\xi_{\tpre}(w) = -\infty$ or $\zeta_\tsuff = \infty$,
            which would force the prefix or suffix of the relevant neighborhoods of $w$ to be empty,
            so that no solution could exist in this case involving $w$.

            Fix vertices $a\in A$ and $c\in C$.
            In order for $b\in B$ to participate in an induced 4-cycle with $a$ and $c$ in the current case, we need 
                \begin{equation}
                \label{eq:case3-pbound}
                f_{BD}(b)\le \min(\xi_\tpre(a), \xi_\tpre(c))
                \end{equation}
            since this condition is equivalent to saying $b$ is a common neighbor of $a$ and $c$ that belongs to the full prefix portions of the neighborhoods of both nodes. 
            Similarly, in order for $d\in D$ to participate in an induced 4-cycle with $a$ and $c$ in the current case,
            we need  
                \begin{equation}
                \label{eq:case3-qbound}
                g_{BD}(d)\ge \max(\zeta_\tsuff(a),\zeta_\tsuff(c))
                \end{equation}
            since this inequality is equivalent to saying $d$ is a common neighbor of $a$ and $c$ appearing in the full suffixes of the neighborhoods of both nodes.

            For $b$ and $d$ to participate in the same induced 4-cycle, by \cref{eq:quadruple:non-edge} we need $\grp{b,d}$ to not be an edge.
            By \Cref{def:ordered-clusters}, this happens if and only if $f_{BD}(b) > g_{BD}(d)$.
            Combining this with \cref{eq:case3-pbound,eq:case3-qbound}, we see that it is possible to select $b\in B$ and $d\in D$ which are common neighbors of vertices $a$ and $c$, appearing respectively in the prefixes and suffixes of the  neighborhoods of these nodes, precisely when 
                \begin{equation}
                    \label{eq:case3-interval-inequality}
                    \min(\xi_\tpre(a), \xi_\tpre(c)) > \max(\zeta_\tsuff(a),\zeta_\tsuff(c)).
                \end{equation}
            Note that here we used the fact that $\xi_\tpre(w)\in\im(f_{BD})$ and $\zeta_\tsuff(w)\in\im(g_{BD})$ for $w\in \set{a,c}$.

            Provided this inequality holds, the only additional constraint we need for it to be possible to extend $a$ and $c$ to an induced 4-cycle is that $\grp{a,c}$ is not an edge. 

            To that end, define $\tilde{C}\sub C$ to be the subset of  nodes $c\in C$ with the property that 
                \begin{equation}
                    \label{eq:case3:restricted-C}
                    \xi_\tpre(c) > \zeta_\tsuff(c).
                \end{equation}
            Now let $S\sub \ZZ^3$ be the set of points 
                \[S = \set{\pair{g_{AC}(c), \xi_\tpre(c), \zeta_\tsuff(c)}\mid c\in \tilde{C}}.\]
            By \Cref{prop:range-search},
            we can in $\tilde{O}(s)$ time insert the points of $S$ into a range query data structure.

            Let $\tilde{A}\sub A$ be the subset of nodes $a\in A$ with the property that 
                \begin{equation}
                    \label{eq:case3-restricted-A}
                    \xi_\tpre(a) > \zeta_\tsuff(a).
                \end{equation}
            For each vertex $a\in \tilde{A}$, we query the data structure to determine if there exists $c\in\tilde{C}$ with
                \begin{equation}
                    \label{eq:case3:ac-non-edge}
                    f_{AC}(a) > g_{AC}(c)
                \end{equation}
            and 
                \begin{equation}
                    \label{eq:case3:ij-comparison}
                    \xi_\tpre(a) > \zeta_\tsuff(c)\quad\text{and}\quad
                    \xi_\tpre(c) > \zeta_\tsuff(a).
                \end{equation}

            If we find such a $c$, we report that the graph has an induced 4-cycle.
            If for all $a\in\tilde{A}$ we find no  $c\in\tilde{C}$ meeting these conditions, we claim there is no induced 4-cycle in this case.

            Indeed, \Cref{def:ordered-clusters} shows that the inequality from \cref{eq:case3:ac-non-edge} holds precisely when $\grp{a,c}$ is not an edge.
            The inequalities from \cref{eq:case3:restricted-C,eq:case3-restricted-A,eq:case3:ij-comparison} together are equivalent to the inequality from \cref{eq:case3-interval-inequality}, which we already proved holds if and only if $a$ and $c$ have non-adjacent common neighbors $b\in B$ and $d\in D$ belonging to their full prefixes and suffixes respectively.
            Thus if some check succeeds in our queries to the data structure, the graph has an induced 4-cycle.
            If instead no check succeeds, then for every choice of non-adjacent $a\in A$ and $c\in C$, the vertices $a$ and $c$ cannot be extended to an induced 4-cycle using vertices $b\in B$ and $d\in D$ which come from their neighborhoods' full prefixes and suffixes respectively. 
            This approach takes $\tilde{O}(s)$ time overall because we make a query for each of the at most $s$ vertices in $\tilde{A}$.

            This completes the case analysis, and shows that in every situation we can determine whether the given quadruple of clusters has an induced 4-cycle in $\tilde{O}(s)$ time.
        \end{proof}

        In our final algorithm for induced 4-cycle detection,
        we will apply \Cref{lem:cluster-triple,lem:cluster-quadruple} to  identify 4-cycles  whose nodes appear in relatively large clusters. 
        To help detect induced 4-cycles where instead some nodes appear in small clusters, the following observation is helpful. 

        \begin{observation}[Neighborhood Size Characterization]
            \label{obs:neighborhood-size}
            Let $X$ be a cluster,
            and let
             $u,v,w\not\in X$ be distinct vertices such that 
             $(u,v,w)$ forms an induced 2-path.
             If the graph has no induced 4-cycle with two or more nodes in $X$,
             then there is an induced 4-cycle of the form $(x,u,v,w)$ for a vertex $x\in X$
            if and only if $\deg_X(v) < \codeg_X(u,w)$.
        \end{observation}
        \begin{proof}
            By assumption, the graph $H$ does not contain an induced 4-cycle using nodes $u$, $v$, and exactly two nodes in $X$.
            Thus, by applying \Cref{obs:triple-comparable-neighborhood} to the clusters $\set{u}$, $\set{v}$, and $X$,
            we get that the neighborhoods $N_X(u)$ and $N_X(v)$ are comparable because $\grp{u,v}$ is an edge.
            Similar reasoning shows that 
            since $\grp{v,w}$ is an edge,
            the neighborhoods $N_X(w)$ and $N_X(v)$ are comparable.
            This then implies that the sets $N_X(u)\cap N_X(w)$ and $N_X(v)$ are comparable.
        
            If $\deg_X(v) < \codeg_X(u,w)$, 
            the comparability condition implies we have 
                \[N_X(v) \subset \grp{N_X(u)\cap N_X(w)}.\]
            Thus, there exists a node $x\in X$ adjacent to both $u$ and $w$ but not $v$.
            In this case $\grp{x, u,v,w}$ forms an induced 4-cycle.

            Conversely, if $\deg_X(v) \ge \codeg_X(u,w)$, the comparability condition implies we have 
                \[\grp{N_X(u)\cap N_X(w)} \subseteq N_X(v).\]
            In this case, any common neighbor in $X$ of $u$ and $w$ is adjacent to $v$, and so $X\cup\set{u,v,w}$ cannot have an induced 4-cycle that uses exactly one node from $X$. 
            But by assumption, $H$ has no induced 4-cycle using two or more nodes from $X$ either.
            
            Thus $H$ does not have an induced 4-cycle, as claimed.
        \end{proof}

    \Cref{obs:neighborhood-size} shows that 
    computing sizes of common neighborhoods in clusters 
    can help with detecting induced 4-cycles in graphs.
    The following result leverages the adjacency structure of ordered cluster pairs to efficiently compute this information.

                \begin{lemma}[Common Neighborhoods in Clusters]
                \label{lem:common-neighbor-size-in-clusters}
                Let $H$ be a graph on pairwise ordered clusters $W,X,Z$ of sizes $r,s,t$ respectively.
                Given $H$
                together with orderings for each of its cluster pairs,
                we can in $\tilde{O}(r + st)$ time compute $\codeg_W(x,z)$
                for all pairs of vertices $(x,z)\in X\times Z$.
        \end{lemma}
        \begin{proof}
            For $Y\in\set{X,Z}$, let 
            $f_{WY}$ and $g_{WY}$ denote the provided orderings for $W$ and $Y$ respectively.
            Let $S\sub \ZZ^2$ be the set of points 
                \[S = \set{\pair{f_{WX}(w), f_{WZ}(w)}\mid w\in W}.\]
            By \Cref{prop:range-search},
            we can in $\tilde{O}(r)$ time insert all points of $S$ into a range query data structure.
            For each $(x,z)\in X\times Z$, we make a query to the data structure to count the number of vertices $w\in W$ satisfying 
                \[f_{WX}(w)\le g_{WX}(x)\quad\text{and}\quad
                  f_{WZ}(w)\le g_{WZ}(z).\]
            By \Cref{def:ordered-clusters}, a vertex $w\in W$ satisfies the above two inequalities if and only if $w$ is adjacent to both $x$ and $z$.
            Consequently, the count returned by the data structure is precisely $\codeg_W(x,z)$.
            Since we make $st$ queries, this algorithm takes $\tilde{O}(r+st)$ time as claimed. 
        \end{proof}

\section{Cluster Decomposition} 
\label{sec:cluster-decomp}
In this section we
show how to decompose any graph avoiding induced 4-cycles 
into a collection of cliques we call \emph{clusters}.
In \cref{subsec:cluster-decomp-large} we introduce a decomposition that extracts large cliques from the graph until the graph outside the clusters is sparse,
based on the approach of~\cite{GyarfasHubenkoSolymosi2002}. 
In \cref{subsec:cluster-decomp-levels} we extend this decomposition to handle clusters of various different sizes.

\subsection{Decomposition into Large Clusters} \label{subsec:cluster-decomp-large}
Our approach is based on \cite[Proof of Theorem~1]{GyarfasHubenkoSolymosi2002},
and presented in \Cref{alg:clique-extraction}.
Roughly, the algorithm works by passing down to a subgraph $\tilde{G}$ of large minimum degree, finding a maximal independent set $I$ in this subgraph,
and then 
considering the sets of common neighbors in $\tilde{G}$ from vertices $x,y\in I$. The following observation is immediate:

\begin{observation} \label{obs:neighborhood-clique}
If the graph $G = (V, E)$ has
no induced 4-cycle, then for all non-adjacent pairs of vertices $(x,y)$, the common neighborhood $N(x) \cap N(y)$ is a clique.
\end{observation}
\begin{proof}
    We prove the contrapositive.
   If $N(x) \cap N(y)$  is not a clique
   for some non-edge $(x,y)$,
   then this common neighborhood contains non-adjacent nodes $u$ and $v$, and thus $(x, u, y, v)$ forms an induced 4-cycle,
   as desired.
\end{proof}

By \Cref{obs:neighborhood-clique},
if we ever identify a large common neighborhood $N(x)\cap N(y)$ in $G$ for a non-edge $(x,y)$, then either we have found a large clique, or we can report an induced 4-cycle.

Otherwise, no common neighborhood is large.
By maximality of $I$ however,
every vertex in $\tilde{G}$ belongs to $I$ or is adjacent to some node in $I$.
For each $x\in I$, we consider the set $U(x)$ of vertices in $\tilde{G}$ whose unique neighbor in $I$ is $x$.
If all the common neighborhoods from $I$ are small,
then some $U(x)$ must be large, because $I$, its common neighborhoods, and the $U(x)$
sets collectively cover the vertices in $\tilde{G}$.
If this large $U(x)$ set is a clique, we can again return it. 
If $U(x)$ is not a clique, then we show that we can swap the non-edge in it with $x$ to replace $I$ with a larger maximal independent set. 
We then run repeat this whole procedure with the new independent set $I$, and argue that this augmentation step cannot occur too many times, so that the overall algorithm is efficient. Formally, we prove the following. 

\begin{lemma}[Clique Extraction in Dense Graphs] \label{lem:clique-extraction}
Let $G = (V, E)$ be a graph with average degree~$d$. There is a deterministic $O(n^2)$-time algorithm that either detects an induced 4-cycle in $G$, or finds a clique $X \subseteq V$ of size at least $\Omega(d^2 / n)$.
\end{lemma}
\begin{proof}
If $d \le 4\sqrt{n}$, we can simply return a single node as a clique of the desired size.

Otherwise, $d > 4\sqrt{n}$.
In this case, we run the algorithm outlined in \Cref{alg:clique-extraction}.
We first prove that this procedure has the desired behavior, and then afterwards bound its runtime. 

\begin{algorithm}[t]
  \caption{Clique Extraction in Dense Graphs} \label{alg:clique-extraction}
  \begin{algorithmic}[1]
    \Input{A graph $G = (V, E)$ on $n$ vertices.}
    \Output{Either an induced 4-cycle in $G$, or 
    a large clique in $G$.}

    \smallskip
    \setlength\itemsep{\medskipamount}

    \State Let~\smash{$d = \frac{1}{n} \cdot\sum_{v \in V} \deg(v)$} and \smash{$\Delta = d^2/(16 n)$}. \label{alg:clique-extraction:line:params}
    \State Repeatedly remove vertices from $G$ with \smash{$\deg(v) \leq d/2$}. Let $\tilde G$ denote the remaining graph with vertex set $\tilde V$. \label{alg:clique-extraction:line:min-degree}
    \State Greedily construct a maximal independent set $I$ in $\tilde G$. \label{alg:clique-extraction:line:mis}
    \While{\smash{$|I| < 4 |\tilde V|/d$}} \label{alg:clique-extraction:line:loop}
      \State For all $x\in I$, compute $U(x) = \{v\in \tilde{V} \setminus I : N_I(v) = \set{x}\}$. \label{alg:clique-extraction:line:unique}
      \If{there exists $x \in I$ with \smash{$|U(x)| \geq (d/8) - 1$}} \label{alg:clique-extraction:line:test-unique}
        \State\parbox[t]{\linewidth-\algorithmicindent-\algorithmicindent}{Select $S\sub U(x)$ with $|S| = \Delta$. If $S$ is a clique, return it. Otherwise, find a pair $(u,v)$ of non-adjacent nodes in $S$ and update the independent set $I \gets (I \setminus\set{x}) \cup \set{u,v}$. Then extend $I$ greedily to a maximal independent set in $\tilde G$.} \label{alg:clique-extraction:line:check-unique}
      \Else
        \State\parbox[t]{\linewidth-\algorithmicindent-\algorithmicindent}{Find distinct nodes $x, y\in I$ with $\codeg_{\tilde V}(x,y) \geq \Delta$. Compute $Z = N(x)\cap N(y)$. If~$Z$ is a clique, return it. Otherwise, report that $G$ has an induced 4-cycle.} \label{alg:clique-extraction:line:check-pairwise}
      \EndIf
    \EndWhile 
    \State Select $S \sub I$ with~\smash{$|S| = (4|\tilde V|)/d$}, and find distinct nodes $x,y\in S$ with $\codeg_{\tilde V}(x,y) \geq \Delta$. Compute $Z = N(x)\cap N(y)$. If $Z$ is a clique, return it. Otherwise, report that $G$ has an induced 4-cycle. \label{alg:clique-extraction:line:check-final}
  \end{algorithmic}
\end{algorithm}

\begin{claim}[Large Minimum Degree]
    \label{claim:algcliqueextsraction:mindegree}
The graph $\tilde G$ constructed in \cref*{alg:clique-extraction:line:min-degree}
of \Cref{alg:clique-extraction}
is non-empty and satisfies $\deg(v) > d/2$ for all $v \in \tilde V$.
\end{claim}
\begin{claimproof}
As we repeatedly remove nodes with degree $\deg(v) \leq d/2$, all remaining nodes $v$ must have the property that $\deg(v) > d/2$. 
Moreover, with each vertex removal we delete at most $d/2$ edges from the graph and therefore reduce the sum of degrees by at most $d$. 
As we also remove one node per step, the average degree remains at least $d$ at each step. Hence, the remaining graph~\smash{$\tilde G$} is nonempty.
\end{claimproof}

\begin{claim}[Independent Set]
    \label{claim:algcliqueextraction:indset}
The set $I$ is a maximal independent set in $\tilde G$
throughout \Cref{alg:clique-extraction}.
\end{claim}
\begin{claimproof}
By definition, $I$ is a maximal independent set of $\tilde{G}$ when it is first initialized in \cref*{alg:clique-extraction:line:mis} of \Cref{alg:clique-extraction}. 
Afterwards, the set $I$ can only change in \cref*{alg:clique-extraction:line:check-unique}. In this step we only update $I$ if we have identified distinct non-adjacent nodes $u,v \in U(x)$. However,~$U(x)$ consists only of vertices outside $I$ that are adjacent to $x$ and not to any other node in $I$. This implies that $u$ and $v$ are distinct from $x$, and that since $I$ is an independent set, so is $(I\setminus\set{x})\cup\set{u,v}$. Finally, we greedily extend $I$ 
in \cref*{alg:clique-extraction:line:mis}
to ensure the maximality condition still holds for $I$.
\end{claimproof}

\begin{claim}[Correctness of \cref*{alg:clique-extraction:line:check-pairwise}]
\label{claim:algcliqueextsraction:line9}
If 
\Cref{alg:clique-extraction} ever
reaches \cref*{alg:clique-extraction:line:check-pairwise}, then there  indeed exist distinct nodes $x, y \in I$ with $\codeg_{\tilde V}(x, y) \geq \Delta$. Moreover, if the algorithm reports an induced 4-cycle in \cref*{alg:clique-extraction:line:check-pairwise}, then $G$ has an induced 4-cycle.
\end{claim}
\begin{claimproof}
We only reach \cref*{alg:clique-extraction:line:check-pairwise} if \smash{$|I| < 4 |\tilde V|/d$} and \smash{$|U(x)| < (d/8) - 1$} for all $x \in I$. 
By \Cref{claim:algcliqueextraction:indset}, $I$ is a maximal independent set in $\tilde{G}$ and thus all nodes in $\tilde V \setminus I$ have a neighbor in $I$. Recall that $U(x)$ is the set of vertices in $\tilde V \setminus I$ whose \emph{unique} neighbor in $I$ is $x$. Therefore, at least
\begin{equation*}
  |\tilde V \setminus I| - |I| \cdot \left(\frac{d}{8} - 1\right) > |\tilde V| - |I| \cdot \frac{d}{8} > \frac{|\tilde V|}{2}
\end{equation*}
of the vertices in $\tilde V$ each have at least \emph{two} neighbors in $I$. 

By averaging, this implies there exists distinct nodes $x, y \in I$ such that
\begin{equation*}
  \codeg_{\tilde V}(x, y) = |N_{\tilde V}(x) \cap N_{\tilde V}(y)| > \frac{|\tilde V|}{2 \binom{|I|}{2}} \ge \frac{|\tilde V|}{|I|^2} \geq \frac{d^2}{16 |\tilde V|} \geq \frac{d^2}{16 n} = \Delta.
\end{equation*}
Finally, we report an induced 4-cycle in \cref*{alg:clique-extraction:line:check-pairwise}
only if $Z = N_{\tilde R}(x) \cap N_{\tilde R}(y)$ is not a clique. In this case,
\Cref{obs:neighborhood-clique} implies that $G$ contains an induced 4-cycle as desired.
\end{claimproof}

\begin{claim}[Correctness of \cref*{alg:clique-extraction:line:check-final}]
    \label{claim:algcliqueextsraction:line10}
If \Cref{alg:clique-extraction} ever reaches \cref*{alg:clique-extraction:line:check-final}, then there indeed exist distinct nodes $x, y \in S$ with $\codeg_{\tilde V}(x, y) \geq \Delta$. Moreover, if the algorithm reports an induced 4-cycle in \cref*{alg:clique-extraction:line:check-final}, then $G$ indeed contains an induced 4-cycle.
\end{claim}
\begin{claimproof}
The algorithm reaches \cref*{alg:clique-extraction:line:check-final} only if \smash{$|I| \geq (4 |\tilde V|)/d$}. 
As in \Cref{alg:clique-extraction}, let $S \subseteq I$ be an arbitrary set of size \smash{$|S| = (4 |\tilde V|)/d$}.
By Bonferroni's inequality (\cref{prop:bonferroni}), we have
\begin{align*}
  |\tilde V|
  &\geq \abs{\bigcup_{x \in S} N_{\tilde V}(x)} \\
  &\geq \sum_{x \in S} |N_{\tilde V}(x)| - \sum_{\substack{x, y \in S\\x \neq y}} |N_{\tilde V}(x) \cap N_{\tilde V}(y)| \\
  &= \sum_{x \in S} \deg_{\tilde V}(x) - \sum_{\substack{x, y \in S\\x \neq y}} \codeg_{\tilde V}(x, y).
\end{align*}
Rearranging this inequality, and recalling that $\deg_{\tilde V}(x) \geq d/2$ for all $x\in\tilde{V}$, we get that 
\begin{equation*}
  \sum_{\substack{x, y \in S\\x \neq y}} \codeg_{\tilde V}(x, y) \geq \sum_{x \in S} \deg_{\tilde V}(x) - |\tilde V| \geq \frac{4 |\tilde V|}{d} \cdot \frac{d}{2} - |\tilde V| = |\tilde V|.
\end{equation*}
Thus by averaging, there exist distinct nodes $x, y \in S$ with
\begin{equation*}
  \codeg_{\tilde V}(x, y) \geq \frac{|\tilde V|}{\binom{|S|}{2}} \geq \frac{2|\tilde{V}|}{|S|^2}= \frac{d^2}{8 |\tilde V|} \geq \frac{d^2}{8 |V|} > \Delta.
\end{equation*}
Finally,  \cref*{alg:clique-extraction:line:check-final} only reports an induced 4-cycle if the set $Z = N_{\tilde V}(x) \cap N_{\tilde V}(y)$ is not a clique. 
If $Z$ is not a clique, then $G$ contains an induced 4-cycle by \Cref{obs:neighborhood-clique}. 
\end{claimproof}

By 
\Cref{claim:algcliqueextsraction:line9,claim:algcliqueextsraction:line10} and the check in \cref*{alg:clique-extraction:line:check-unique} of \Cref{alg:clique-extraction}, 
whenever
 \Cref{alg:clique-extraction} reports a clique (possibly in \cref*{alg:clique-extraction:line:check-unique,alg:clique-extraction:line:check-pairwise,alg:clique-extraction:line:check-final}) we have explicitly verified that it is indeed a clique on at least~$\Delta$ vertices. 
Similarly,
by \Cref{claim:algcliqueextsraction:line9,claim:algcliqueextsraction:line10}, whenever the algorithm reports an induced 4-cycle (possibly in \cref*{alg:clique-extraction:line:check-pairwise,alg:clique-extraction:line:check-final}) the input graph does indeed have an induced 4-cycle.
This proves that the algorithm is correct as claimed. 

It remains to prove that \Cref{alg:clique-extraction} can be implemented to run in $O(n^2)$ time. 

\cref*{alg:clique-extraction:line:params} only sets parameters,
so takes $O(1)$ time. 

\cref*{alg:clique-extraction:line:min-degree}
can be implemented  by sorting the vertices of the initial graph by degree,
and then repeatedly deleting the vertex of minimum degree while this value is at most $d/2$, and updating degrees of vertices after each deletion. 
Since each deletion and degree update takes time proportional to the number of vertices and edges deleted,
and we never delete the same vertex or edge twice,
this takes at most $O(n^2)$ time overall.

\cref*{alg:clique-extraction:line:mis} can be implemented in $O(n^2)$ time by scanning just once through the vertices and edges of the graph and greedily including vertices to build up the independent set $I$. 

The update rule for $I$ in \cref*{alg:clique-extraction:line:check-unique} ensures that each iteration of the while loop that does not return a clique or report an induced 4-cycle increases the size of $I$ by at least one. The condition in \cref*{alg:clique-extraction:line:loop}
that \smash{$|I| < (4|\tilde V|)/d$} thus implies that there are at most $O(|\tilde V| / d) \le O(n / d)$ loop iterations.

In each iteration of the loop, we first compute $U(x)$
for all $x \in I$ in \cref*{alg:clique-extraction:line:unique}.
We compute these sets by scanning over all vertices $x\in I$, and recording for each node $v\in \tilde{V} \setminus I$ adjacent to $x$ the name of the vertex $x$ in a list $L(v)$ associated with $v$. We then scan over the nodes $v \in \tilde{V} \setminus I$, and for each $v$ where the list $L(v)$ consists of a single vertex $x \in I$, we include $v$ in $U(x)$.
We can also record the sizes of the $U(x)$ sets at this time,
and check if there exists $x\in I$ satisfying the inequality from \cref*{alg:clique-extraction:line:test-unique}.
All of these steps together take at most $O(dn)$ time per iteration because 
we spend time proportional to the number of edges $O(dn)$. 
Since there are at most $O(n/d)$ loop iterations,
this takes at most $O(n^2)$ time overall. 

In \cref*{alg:clique-extraction:line:check-unique},
we may check if a set $S$ of size $\Delta$ is a clique. 
This takes at most $O(\Delta^2) \le O(d^4 / n^2)$ time per iteration and thus at most  
\[O((n/d)\cdot (d^4/n^2)) \le O(d^3 / n) \le O(d^2) \le O(n^2)\] time overall.
Here we used the assumption from the beginning of this proof that $d \ge \Omega(\sqrt{n})$.

If $S$ is not a clique, we then extend $I$ to a maximal independent set.
We do this in $O(|I| \, n)$ time by testing all pairs of nodes inside and outside the independent set.
In every iteration of this step, before $I$ is extended to become maximal, we have $|I|\le (4|\tilde{V}|/d)$ because of the
the condition from 
\cref*{alg:clique-extraction:line:loop} and 
the fact that the update $I\leftarrow (I\setminus\set{x})\cup\set{u,v}$
from \cref*{alg:clique-extraction:line:check-unique} increases the size of $I$ by exactly one. 
While we have this size bound \smash{$|I| \leq (4 |\tilde V|)/d$}, the $O(|I| \, n)$ runtime
is at most $O(n^2 / d)$, so this step takes at most  
\[O(n^3 / d^2) \le O(n^2)\] 
time across all iterations. 
In the final iteration, it may happen that~\smash{$|I| > 4 |\tilde R|/d$} in which case we bound the running time of \cref*{alg:clique-extraction:line:check-unique} by $O(n^2)$.

We execute \cref*{alg:clique-extraction:line:check-pairwise} only once (as the algorithm terminates right after). This step involves testing if a set $Z$ of size at most $n$ is a clique, and thus takes at most $O(n^2)$ time.

Finally, in \cref*{alg:clique-extraction:line:check-final} we compute $\codeg_{\tilde V}(x, y)$ for $O(|S|^2) = O(n^2 / d^2) = O(n)$ pairs~$(x, y)$. Each computation takes $O(n)$ time, so in total we spend at most $O(n^2)$ time. Afterwards we verify if the set $Z$ is a clique, which also takes time $O(n^2)$. This completes the running time analysis and thus the proof of \cref{lem:clique-extraction}.
\end{proof}

We now repeatedly apply \Cref{lem:clique-extraction} to decompose the input graph into a collection of large clusters and a single sparse remainder. 

\lemlargeclusterdecomp*
\begin{proof}
We initialize $\mathcal X \gets \emptyset$ and $R \gets V$. While $G[R]$ has at least $n^{3/2} \Delta^{1/2}$ edges, we apply the algorithm from \cref{lem:clique-extraction} on the graph $G[R]$. 
If the algorithm identifies an induced 4-cycle, 
we immediately stop and report this. Otherwise, \cref{lem:clique-extraction} returns a clique $X \subseteq R$ of size at least $\Omega(d^2 / |R|)$, where $d$ is the average degree in $G[R]$. 
Since
\begin{equation*}
  d \geq \frac{n^{3/2} \Delta^{1/2}}{|R|},
\end{equation*}
it follows that
\begin{equation*}
  |X| \geq \Omega\left(\frac{n^3 \Delta}{|R|^3}\right) = \Omega(\Delta).
\end{equation*}
We take any subset $\tilde{X} \subseteq X$ of size $\Theta(\Delta)$,
then update $\mathcal X \gets \mathcal X \cup \set{X}$ and $R \gets R \setminus \tilde{X}$, and repeat.

It is immediate that when the algorithm terminates, it either correctly reports an induced 4-cycle (by \cref{alg:clique-extraction}), or returns a collection $\mathcal X$ of 
disjoint
cliques of size $\Theta(\Delta)$ and a remainder satisfying that $G[R]$ has less than $n^{3/2} \Delta^{1/2}$ edges.

To bound the running time, note that the total number of iterations is at most $O(n / \Delta)$
because the cliques we extract are of size $\Theta(\Delta)$ and disjoint. 
Each application of \cref{alg:clique-extraction} takes time $O(n^2)$ and the remaining updates take time $O(n)$, so the total runtime is $O(n^3 / \Delta)$ as claimed.
\end{proof}

\subsection{Decomposition into Levels of Clusters} \label{subsec:cluster-decomp-levels}

\begin{lemma}[Low-Level Cluster Decomposition] \label{lem:cluster-decomp-low}
Let $G = (V, E)$ be the input graph. Let $\Delta \geq 1$. In  $\tilde{O}(n^3 / \Delta + n^{5/2} \Delta^{1/2})$ time we can either detect an induced 4-cycle in $G$, or compute a decomposition
\begin{equation*}
  V = \left(\bigsqcup_{X \in \mathcal X} X\right) \sqcup R,
\end{equation*}
where each $X \in \mathcal X$ is a clique of size $\Theta(\Delta)$, 
such that
\begin{equation*}
  |N_R(x) \cap N_R(y)| \leq O(\Delta)
\end{equation*}
for all non-edge pairs of distinct vertices $(x,y)$. Moreover, the algorithm returns $N_R(x) \cap N_R(y)$ for
all such non-edges $(x,y)$.
\end{lemma}
\begin{proof}
We write $N_R(x, y) = N_R(x) \cap N_R(y)$, where we naturally adopt the convention that the order of $x$ and $y$ does not matter (i.e., $N_R(x, y) = N_R(y, x)$). 
The algorithm is described in \cref{alg:cluster-decomp-low}. In the following we analyze its correctness and running time.

\begin{algorithm}[t]
  \caption{Low-Level Cluster Decomposition} \label{alg:cluster-decomp-low}
  \begin{algorithmic}[1]
    \Input{A graph $G = (V, E)$ on $n$ vertices, and a subset of vertices $R \subseteq V$.}
    \Output{Either an induced 4-cycle in $G$, or a collection $\mathcal X$ of cliques along with sets $N_R(x, y)$.}

    \smallskip
    \setlength\itemsep{\medskipamount}

    \State Run \cref{lem:cluster-decomp-large} on $G$. If we detect an induced 4-cycle in $G$, we stop and report the induced 4-cycle. Otherwise, let $V = (\bigsqcup_{X \in \mathcal X} X) \sqcup R$ denote the resulting partition. \label{alg:cluster-decomp-low:line:decomp-large}
    \State Initialize $\mathcal Z \gets \emptyset$ \label{alg:cluster-decomp-low:line:init-Z}
    \State Compute $N_R(x, y) = N_R(x) \cap N_R(y)$ for all distinct $x \in R,\, y \in V$ as follows: Enumerate all adjacent $x, z \in R$ and enumerate all $y \in N(z)$. For each such triple insert $z$ into $N_R(x, y)$. \label{alg:cluster-decomp-low:line:enumerate-incident}
    \State\parbox[t]{\linewidth}{While there is a non-edge $(x,y)\in R\times V$ 
    such that $|N_R(x, y)| \geq \Delta$, test if $Z \gets N_R(x, y)$ is a clique. If it is, update $R \gets R \setminus Z$ and $\mathcal Z \gets \mathcal Z \cup \set{Z}$. If not, report an induced 4-cycle.}\smallskip\label{alg:cluster-decomp-low:line:eliminate-incident}
    \ForEach{$X, Y \in \mathcal X$ with $X \neq Y$} \label{alg:cluster-decomp-low:line:loop}
      \State\parbox[t]{\linewidth-\algorithmicindent}{Apply \cref{lem:cluster-pair} on the pair of cliques $(X, Y)$. If it reports an induced 4-cycle, we stop and report the induced 4-cycle. Otherwise, we get orderings $f_{XY}, g_{XY}$ for $(X,Y)$.}\smallskip\label{alg:cluster-decomp-low:line:ordering}
      \State For each $z \in R$, sort the set $N_Y(z)$ by $g_{XY}$.\label{alg:cluster-decomp-low:line:sort}
      \State\parbox[t]{\linewidth-\algorithmicindent}{Compute $N_R(x, y)$ for all non-adjacent $x \in X,\, y \in Y$ as follows: Enumerate all adjacent $x \in X$ and $z \in R$, and then enumerate all $y \in N_Y(z)$ with $f_{XY}(x) > g_{XY}(y)$. For each such triple insert $z$ into $N_R(x, y)$.}\smallskip\label{alg:cluster-decomp-low:line:enumerate-nonincident}
      \State\parbox[t]{\linewidth-\algorithmicindent}{While there are non-adjacent distinct nodes $x \in X,\, y \in Y$ satisfying that $|N_R(x, y)| \geq \Delta$, test if $Z \gets N_R(x, y)$ is a clique. If yes, update $R \gets R \setminus Z$ and $\mathcal Z \gets \mathcal Z \cup \set{Z}$. If no, report an induced 4-cycle.}\smallskip\label{alg:cluster-decomp-low:line:eliminate-nonincident}
    \EndForEach
    \State Return the cliques in $\mathcal X \cup \mathcal Z$, and the sets $N_R(x, y)$ for all non-edges $(x, y)$.\label{alg:cluster-decomp-low:line:return}
  \end{algorithmic}
\end{algorithm}

\paragraph{Correctness.}
The first observation is that whenever the algorithm reports an induced 4-cycle, then~$G$ indeed contains an induced 4-cycle. In \cref*{alg:cluster-decomp-low:line:decomp-large} this is due to \cref{lem:cluster-decomp-large}, in \cref*{alg:cluster-decomp-low:line:ordering} this is due to \cref*{lem:cluster-pair}, and in \cref*{alg:cluster-decomp-low:line:eliminate-incident,alg:cluster-decomp-low:line:eliminate-nonincident} this is due to \Cref{obs:neighborhood-clique}. For the remaining analysis we assume that the algorithm does not report an induced 4-cycle.

It is straightforward to verify that all sets in $\mathcal X$ and $\mathcal Z$ are cliques of the desired size $\Omega(\Delta)$. For~$\mathcal X$ this is by \cref{lem:cluster-decomp-large}, and for~$\mathcal Z$ 
this is because in
all cases where we insert some set $Z$ into $\mathcal Z$ (namely, \cref*{alg:cluster-decomp-low:line:eliminate-incident,alg:cluster-decomp-low:line:eliminate-nonincident}) we have explicitly tested that $Z$ is a clique with ~\makebox{$|Z| \geq \Delta$}. 
Strictly speaking,
\Cref{alg:cluster-decomp-low} as presented
does not ensure that the cliques have size $\Theta(\Delta)$, but this can be ensured in a postprocessing step where we subdivide cliques that exceed size $2 \Delta$.

Next, we argue that the algorithm correctly returns the sets $N_R(x, y) = N_R(x) \cap N_R(y)$ for all non-edge pairs $(x, y)$, and that each such set has size at most $\Delta$ when the algorithm terminates. We distinguish two cases:
\begin{itemize}
  \item\emph{$x \in R$ or $y \in R$:} Without loss of generality assume that $x \in R$. Then we compute $N_R(x, y)$ correctly in \cref*{alg:cluster-decomp-low:line:enumerate-incident}. Moreover, in \cref*{alg:cluster-decomp-low:line:eliminate-incident} we distinguish two cases for $N_R(x, y)$: If $N_R(x, y) < \Delta$, then the claim is immediate. Otherwise, if $N_R(x, y) \geq \Delta$, then we remove $N_R(x, y)$ from $R$. In particular, after this update we have that $N_R(x, y) = \emptyset$ and the claim holds.
  \item\emph{$x, y \not\in R$:} In this case $x$ and $y$ must appear in some cliques $x \in X,\, y \in Y$ with $X, Y \in \mathcal X$. Moreover, these two cliques must be distinct as otherwise $x, y$ would be adjacent. Focus on the loop iteration that considers pair $(X, Y)$. In \cref*{alg:cluster-decomp-low:line:enumerate-nonincident} we enumerate all $(\tilde x, \tilde z, \tilde y) \in X \times R \times Y$ such that $\grp{\tilde x, \tilde z}$ and $\grp{\tilde z, \tilde y}$ are edges,
   such that $f_{XY}(x) > g_{XY}(y)$. By \cref{def:ordered-clusters}, this last condition is equivalent to $(x, y)$
   being a non-edge. 
  Therefore, in \cref*{alg:cluster-decomp-low:line:enumerate-nonincident} we enumerate all induced 2-paths $(\tilde x, \tilde z, \tilde y)$ and for each insert $\tilde z$ into $N_R(\tilde x, \tilde y)$.
  It follows that the set $N_R(x, y)$ is constructed correctly. Then in \cref*{alg:cluster-decomp-low:line:eliminate-nonincident} we again test if $|N_R(x, y)| < \Delta$ or $|N_R(x, y)| \geq \Delta$. In the former case the claim is immediate. In the latter case, the algorithm removes all nodes in $N_R(x, y)$ from $R$ and so afterwards $N_R(x, y) = \emptyset$,
  so the claim holds in this case too.
\end{itemize}

\paragraph{Implementation Detail: Maintaining \boldmath$N_R(x, y)$.}
During its execution \cref{alg:cluster-decomp-low} keeps deleting nodes from $R$, and this affects the previously computed sets $N_R(x, y) = N_R(x) \cap N_R(y)$. To efficiently deal with these deletions we additionally maintain pointers from each node $z \in R$ to all sets $N_R(x, y)$ containing $z$. That is, whenever we include $z$ into some set $N_R(x, y)$ we additionally spend $O(1)$ time to prepare this pointer. Then, when $z$ is removed from $R$ we traverse all the sets that $z$ is pointing to and remove $z$ from these sets. Additionally, we  maintain the sizes of the sets $N_R(x, y)$ in a priority queue so that we can efficiently decide in time $O(\log n)$ if there is a set of size at least $\Delta$ (in \cref*{alg:cluster-decomp-low:line:eliminate-incident,alg:cluster-decomp-low:line:eliminate-nonincident}).

Note that the total time spent on maintaining these additional data structures is proportional (up to logarithmic factors) to the total size of the sets $N_R(x, y)$ and thus proportional to the time to construct the sets $N_R(x, y)$. For this reason we will neglect the time to update the sets $N_R(x, y)$ in the following runtime analysis.

\paragraph{Running Time.}
We finally analyze the running time of \cref*{alg:cluster-decomp-low}. The initial call to \cref*{lem:cluster-decomp-large} in \cref*{alg:cluster-decomp-low:line:decomp-large} takes time $\tilde O(n^3 / \Delta)$. Initializing $\mathcal Z$ in \cref*{alg:cluster-decomp-low:line:init-Z} is in constant time.

In \cref*{alg:cluster-decomp-low:line:enumerate-incident} we
enumerate all 2-paths with at least one edge in $R$.
By \Cref{lem:cluster-decomp-large} the graph
~$G[R]$ contains at most $O(n^{3/2} \Delta^{1/2})$ edges, so this step takes at most $O(n^{5/2} \Delta^{1/2})$ time. 
Then, in \cref*{alg:cluster-decomp-low:line:eliminate-incident} we repeatedly take one of the previously computed sets $Z = N_R(x, y)$ and test if it is a clique. This is implemented naively in time $O(|Z|^2)$. However, as afterwards we remove $Z$ from $R$, each pair of nodes is involved in at most clique test and so this step takes at most $O(n^2)$ time overall.

Now focus on the loop in \cref*{alg:cluster-decomp-low:line:loop}. Recall that each clique $X \in \mathcal X$ has size at least $\Omega(\Delta)$ and thus $|\mathcal X| \leq O(n / \Delta)$. Hence, there are at most $O((n / \Delta)^2)$ iterations of the loop. Focus on a fixed iteration of this loop, examining the cliques $X, Y$.

In \cref*{alg:cluster-decomp-low:line:ordering} we first run \cref{lem:cluster-pair} to compute the ordering of $(X, Y)$ in time $O(\Delta^2)$ per iteration and thus time $O(n^2)$ in total. Then, in \cref*{alg:cluster-decomp-low:line:sort}, we sort all sets $N_Y(z)$ for $z \in R$ in time $\tilde O(|R| \cdot |Y|) = \tilde O(n \Delta)$ per iteration and time $\tilde O(n^3 / \Delta)$ in total. \cref*{alg:cluster-decomp-low:line:enumerate-nonincident} is more interesting: We compute the sets $N_R(x, y)$ in time proportional to the total size of the computed sets,
\begin{equation*}
  \Sigma := \sum_{\substack{x \in X, y \in Y\\\grp{x, y} \not\in E}} |N_R(x, y)|.
\end{equation*}
We distinguish two cases: If $|\Sigma| \leq |X|  |Y| \cdot \Delta$, then we call this iteration \emph{good}, otherwise we call it \emph{bad}. On the one hand, the total running time of \cref*{alg:cluster-decomp-low:line:enumerate-nonincident} across all good iterations is bounded by $O((n / \Delta)^2 \cdot \Delta^3) = O(n^2 \Delta)$.
On the other hand, by averaging,
in each bad iteration there must be at least one non-edge $(x,y)\in X\times Y$ 
with $|N_R(x, y)| \geq \Delta$. In \cref*{alg:cluster-decomp-low:line:eliminate-nonincident} we will therefore find at least one non-edge $(x, y)$ with $|N_R(x, y)| \geq \Delta$. In this case, the algorithm stops immediately (if it detects an induced 4-cycle), or  we remove at least $\Delta$ nodes from $R$. 
The latter event clearly happens at most $|R| / \Delta$ times, and thus the total number of bad iterations is at most  $O(n / \Delta)$. In each such iteration, \cref*{alg:cluster-decomp-low:line:enumerate-nonincident} takes time $O(|X| \cdot |Y| \cdot |R|) = O(n \Delta^2)$ in the worst case, so the running time of \cref*{alg:cluster-decomp-low:line:enumerate-nonincident} across all bad iterations is at most  $O((n / \Delta) \cdot n \Delta^2) \le O(n^2 \Delta)$. Finally, in \cref*{alg:cluster-decomp-low:line:eliminate-nonincident} we then repeatedly test if some sufficiently large sets $N_R(x, y)$ are cliques. By the same argument as for \cref*{alg:cluster-decomp-low:line:eliminate-incident}, the total time spent on this step across all iterations of the outer loop is $O(n^2)$.
\end{proof}

\begin{theorem}[Layered Cluster Decomposition] \label{thm:clique-decomp}
Let $G = (V, E)$ be a graph, and let $L$ and $H$ be integers with \makebox{$1 \leq L \leq H = \lfloor\log n\rfloor$}. 
There is an $\tilde{O}(n^2\cdot 2^L + n^3/2^L)$-time algorithm 
that \mbox{either reports} an induced 4-cycle in $G$,
or returns 
\begin{itemize}
    \item 
    a vertex partition $V = V_L \sqcup \dots \sqcup V_H$,
    \item 
    collections of disjoint vertex subsets $\mathcal{X}_L, \dots, \mathcal{X}_H$, and 
    \item 
    the sets $N_\ell(x, y) = N_{V_\ell}(x) \cap N_{V_\ell}(y)$ for all $\ell\in\set{L+1, \dots, H}$ and non-edges $\grp{x,y}\not\in E$, 
\end{itemize}
such that we have
\begin{description}
  \item[\textit{\emph{Levels of Clusters:}}] all  parts $V_L, \dots, V_H$ can be further decomposed into \smash{$V_\ell = \bigsqcup_{X \in \mathcal X_\ell} X$}, where each vertex subset $X \in \mathcal X_\ell$ is a clique in $G$ of size $\Theta(n / 2^\ell)$; and 
  \item[\textit{\emph{Bounded Common Neighborhoods:}}]
    $|N_\ell(x, y)| \leq O(n / 2^\ell)$
  for all $L < \ell \leq H$ and  $\grp{x, y} \not\in E$.
\end{description}
\end{theorem}
\begin{proof}
We start by describing the algorithm, outlined in \cref{alg:cluster-decomp}, and explain why it is correct.
We first apply the algorithm from \cref{lem:cluster-decomp-low} 
with parameter $\Delta = n / 2^L$ to decompose the graph into a vertex part $V_L$ that is the disjoint union of large cliques from a family~$\mathcal X_L$, plus some vertices from a remainder set $R$. In addition, \cref{lem:cluster-decomp-low} computes the sets of common neighbors $N_R(x) \cap N_R(y)$ for all non-adjacent nodes $x, y$. We will maintain these sets throughout in a data structure that allows for efficient updates when we remove nodes from $R$ in the future (implemented as in the proof of \cref{lem:cluster-decomp-low}).

\begin{algorithm}[t]
\caption{Layered Cluster Decomposition} \label{alg:cluster-decomp}
\begin{algorithmic}[1]
  \Input{A graph $G = (V, E)$, and integer parameters $L,H$ with $1 \leq L \leq H = \lfloor\log n\rfloor$}
  \Output{Either an induced 4-cycle in $G$, or a decomposition as specified in \cref{thm:clique-decomp}}

  \smallskip
  \setlength\itemsep{\medskipamount}

  \State Run \cref{lem:cluster-decomp-low} on $G$ with parameter $\Delta = n / 2^L$. If the algorithm detects an induced 4-cycle, we stop and report it. Otherwise, the algorithm returns a set of cliques $\mathcal X_L$, a remainder set~$R$ and the $N_R(x) \cap N_R(y)$ for all non-adjacent distinct nodes $x, y \in V$. We keep these sets updated when we remove vertices from $R$ in the steps below. Set \smash{$V_L \gets \bigsqcup_{X \in \mathcal X_L} X$}. \label{alg:cluster-decomp:line:large}
  \ForEach{$\ell \gets L+1, \dots, H-1$} \label{alg:cluster-decomp:line:loop-level}
    \State Let $V_\ell \gets \emptyset$ and $\mathcal X_\ell \gets \emptyset$. \label{alg:cluster-decomp:line:init-part}
    \ForEach{non-edge $\grp{x, y} \not\in E$} \label{alg:cluster-decomp:line:loop-non-edges}
      \State\parbox[t]{\linewidth-\algorithmicindent-\algorithmicindent}{Let $X = N_R(x) \cap N_R(y)$. If $|X| > n / 2^\ell$, then verify that $X$ is a clique in $G$, and in this case update $R \gets R \setminus X$ and $V_\ell \gets V_\ell \cup X$ and $\mathcal X_\ell \gets \mathcal X_\ell \cup \set{X}$. Otherwise, report that~$G$ contains an induced 4-cycle.} \label{alg:cluster-decomp:line:extract}
    \EndForEach
  \EndForEach
  \State Set $V_H \gets R$ and take $\mathcal X_H = \set{\set{x} : x \in V_H}$.
\end{algorithmic}
\end{algorithm}

We then iterate over $\ell \gets L+1, \dots, H-1$. 
For each choice of $\ell$,
we try out all non-edges $\grp{x, y}$ and check if the set $X = N_R(x) \cap N_R(y)$ has size more than $n / 2^\ell$. 
If so, we check if $X$ is a clique.
If it is a clique, we remove $X$ from $R$ and include it as a set in $\cal{X}_\ell$ and its vertices in $V_\ell$. 
If $X$ is not a clique, we report that there is an induced 4-cycle in the graph.
This final reporting step is correct by \Cref{obs:neighborhood-clique}.

Once these iterations are complete, we take $V_H$ to be the set of remaining nodes in $R$, and let~$\mathcal X_H$ be the trivial partition of $V_H$ into cliques of size one.

At this point, we have computed all sets $V_\ell$.
By scanning through the vertices in these sets,
we can record for each vertex $v$ in the graph the unique index $\ell$ such that $v\in V_\ell$. Now, let~$\tilde{R}$ denote the initial set $R$ returned in \cref*{alg:cluster-decomp:line:large} of \Cref{alg:cluster-decomp}.
In that step, we will have computed $N_{\tilde{R}}(x)\cap N_{\tilde{R}}(y)$ for all non-edges $\grp{x,y}\not\in E$. By definition, for each $\ell$ the sets $N_\ell(x, y) =  N_{V_\ell}(x)\cap N_{V_\ell}(y)$ can be equivalently written as 
\[N_\ell(x,y) = V_\ell\cap \grp{N_{\tilde{R}}(x)\cap N_{\tilde{R}}(y)}.\] 
So having computed and saved the $N_{\tilde{R}}(x,y)$ sets,
we can compute all of the $N_\ell(x,y)$ sets by scanning through the the vertices $v$ in the $N_{\tilde{R}}(x,y)$ sets, for each $v$ checking which part $V_\ell$ it belongs to, and including $v$ in $N_\ell(x,y)$ (with efficient data structures as in the proof of \Cref{lem:cluster-decomp-low}).
This completes the description of the algorithm.

\paragraph{Correctness.}
We now explain why the algorithm is correct.
First, we have already proved above that \Cref{alg:cluster-decomp} reports an induced 4-cycle only when $G$ contains an induced 4-cycle.

Second, we claim each set $\mathcal{X}_\ell$ 
only contains cliques of size at least $n / 2^\ell$.
This holds for $\ell=L$ by 
setting $\Delta = n/2^L$ for our application
of \Cref{lem:cluster-decomp-low} in
\cref*{alg:cluster-decomp:line:large},
and holds for $\ell > L$ by our rule for adding cliques in 
\cref*{alg:cluster-decomp:line:extract}.
Also, although not explicitly written in \Cref{alg:cluster-decomp}, if a clique in $\cal{X}_\ell$ has size greater than $n/2^{\ell-1}$,
then we split it into several cliques of size at most $n/2^{\ell-1}$ and at least $n/2^\ell$.
This extra post-processing step ensures that all cliques in $\cal{X}_\ell$ have size $\Theta(n/2^\ell)$.
  
Third, we claim that  $|N_\ell(x,y)|\le n/2^{\ell-1}$ for all $L < \ell \leq H$ and all non-edges $\grp{x, y}$.
Indeed, if this were not the case for some index $\ell$ and non-edge $\grp{x,y}$,
then in the $(\ell-1)^{\text{st}}$ iteration of \cref*{alg:cluster-decomp:line:extract} of \Cref{alg:cluster-decomp}
we would have extracted 
$N_{R}(x)\cap N_R(y) \supseteq N_\ell(x,y)$ and included it in $\mathcal{X}_{\ell-1}$.
In particular, all the vertices in $N_\ell(x,y)$ would be deleted from $R$ before iteration $\ell$,
contradicting the assumption that $|N_\ell(x,y)| > n/2^{\ell-1}$.

\paragraph{Running Time.}
We finally analyze the algorithm's running time. The initial call to \cref{lem:cluster-decomp-low} in \cref*{alg:cluster-decomp:line:large}
of \Cref{alg:cluster-decomp}
takes  $\tilde O(n^3 / \Delta + n^{5/2} \Delta^{1/2}) = \tilde O(n^2 \cdot 2^L + n^3 / 2^{L/2})$ time. In each of the $O(\log n)$ iterations of the loop in \cref*{alg:cluster-decomp:line:loop-level}, we enumerate $O(n^2)$ non-edges in \cref*{alg:cluster-decomp:line:loop-non-edges}. 
For each non-edge we query the size of~$N_R(x) \cap N_R(y)$ in constant time (by maintaining an appropriate counter), and then possibly test if $N_R(x) \cap N_R(y)$ is a clique. 
If this set is a clique, we remove it from $R$ and never again need to check if its edges belong to a clique.
If this set is not a clique, we halt.
Thus, this last clique checking and extraction procedure takes at most $O(n^2)$ time overall, because we only ever extract disjoint cliques from an $n$-node graph. 
\end{proof}

\section{Induced 4-Cycle Detection} \label{sec:full}

In this section, we present our algorithm for induced 4-cycle detection.
Our algorithm begins by applying the algorithm from  \cref{thm:clique-decomp} with parameters $L = \lfloor\frac12 \log n\rfloor$ and $H = 
\lfloor \log n\rfloor$.
The algorithm reports an induced 4-cycle, in which case we are done,
or obtains a partition $V = V_L \sqcup \dots \sqcup V_H$ where each part~$V_\ell$ is the disjoint union of cliques of size $\Theta(n / 2^\ell)$ from a collection $\mathcal X_\ell$,
along with some additional data.
In \Cref{subsec:cluster-levels}, we assume we are given the data of such a decomposition, and apply the results from \Cref{sec:cluster} to design various algorithms for finding induced 4-cycles depending on which levels $V_\ell$ its nodes come from.
In \Cref{subsec:final-alg} we combine all of these algorithms together to prove \Cref{thm:deterministic}.

\subsection{Casework on Cluster Levels}
    \label{subsec:cluster-levels}

    Let $G = (V,E)$ be the input graph.
    Set parameters $L = \lfloor (1/2)\log n\rfloor$ and $H = \lfloor \log n\rfloor$.
    Throughout this subsection, we assume we have a vertex partition $V = V_L\sqcup\dots \sqcup V_H$, families of clusters $X_\ell$ for $\ell\in\set{L, \dots, H}$, and access to common neighborhoods
    $N_\ell(x,y) = N_{V_\ell}(x)\cap N_{V_\ell}(y)$ for all $\ell\in\set{L+1, \dots, H}$ and non-edges $\grp{x,y}\not\in E$ 
    that satisfy the \emph{Levels of Clusters} and \emph{Bounded Common Neighborhoods} conditions from \Cref{thm:clique-decomp}.
    We refer to the cliques appearing in the $\mathcal{X}_\ell$ collections as \emph{clusters}.

    Recall that we 
    represent induced 4-cycles as tuples $(a,b,c,d)$ of their vertices,
    ordered such that $\grp{a,b}$, $\grp{b,c}$, $\grp{c,d}$, $\grp{d,a}$ are edges, and $\grp{a,c}$, $\grp{b,d}$ are not edges in $G$. 
    We say an induced 4-cycle is \emph{$k$-clustered} if its 
    vertices come from $k$ distinct clusters.

    No induced 4-cycle can be 0-clustered, because the clusters partition the vertices of $G$,
    and 
    no induced 4-cycle can be 1-clustered,
    since if four vertices lie in a single cluster they form a four-clique. 
    Thus, each induced 4-cycle is $k$-clustered for some $k\in\set{2,3,4}$.

    The following result lets us detect induced 4-cycles which are 2-clustered.
    If we do not find any such 4-cycles,
    we are able to impose orderings on the inter-cluster edges as described in \Cref{def:ordered-clusters}.

        \begin{lemma}[2-Clustered Detection] \label{lem:2-clustered}
        There is an $O(n^2)$-time algorithm that detects a 2-clustered induced 4-cycle in $G$ if any exist, and otherwise returns concise orderings for all pairs of clusters.
        \end{lemma}
        \begin{proof}
        Let
            \[\mathcal{X} = \bigsqcup_{\ell=L}^{H} \mathcal{X}_\ell\]
        be the collection of all clusters in our decomposition.
        We go over all pairs $(X,Y)\in\mathcal{X}$ of distinct clusters,
        and for each run the $O(|X||Y|)$ time algorithm from \Cref{lem:cluster-pair}.
        
        If any call to this algorithm detects an induced 4-cycle, we 
        can report this.
        Otherwise, if no call to \Cref{lem:cluster-pair} detects an induced 4-cycle,
        we have certified that the clusters in $\cal{X}$ are pairwise ordered,
        and obtained orderings for each cluster pair.

        The total runtime of this algorithm is asymptotically at most 
            \[\sum_{X,Y\in\mathcal{X}} |X||Y| = \grp{\sum_{X\in\mathcal{X}} |X|}^2 \le n^2\]
        where we have used the fact that the clusters in $\cal{X}$ are disjoint. 
        \end{proof}

    If \Cref{lem:2-clustered} fails to find an induced 4-cycle,
    we may now assume that the clusters in our graph are pairwise ordered.
    We will use this additional structure to seek induced 4-cycles that are $k$-clustered for $k\in\set{3,4}$
    using more sophisticated algorithms.

    \subsubsection{Cycles Among Three Clusters}
    \label{subsubsec:3-clustered}
We represent 3-clustered induced 4-cycles
as tuples $(v_1, \tilde{v}_1, v_2, v_3)$
such that $v_1, \tilde{v}_1$ belong to the same cluster, and $v_2$ and $v_3$ belong to two other distinct clusters.
We say such an induced 4-cycle has \emph{type} $\vec{t} = \pair{t_1, t_2, t_3}$
if $v_1, \tilde{v}_1\in V_{t_1}$, and $v_i\in V_{t_{i}}$ for $i\in\set{2,3}$.

To organize our casework, we informally associate each type $\vec{t}$ with labels from $\set{\t{L}, \t{H}, \starup}^3$ encoding the relative sizes of the $t_i$ coordinates of $\vec{t}$.
Intuitively, if $\vec{s}\in \set{\t{L}, \t{H}, \starup}^3$ is associated with the type $\vec{t}$, then for each $i\in [3]$,
 $\vec{s}[i] = \t{L}$ means that $t_i$ is ``low'' (close in value to $L$),
 $\vec{s}[i] = \t{H}$ means that $t_i$ is ``high'' (close in value to $H$),
and $\vec{s}[i] = (\starup)$ does provide any information about $t_i$.
We now present several different algorithms for detecting 3-clustered induced 4-cycles, 
parameterized by the types of these cycles.

\begin{lemma}[LLL Types] \label{lem:type-lll}
Fix $\vec{t} = \pair{t_1, t_2, t_3}$.
Given orderings between all cluster pairs in $G$,
we can determine in $\tilde O(n \cdot 2^{t_1 + t_2 + t_3 - \min(t_1, t_2, t_3)})$ time 
whether $G$ contains an induced 4-cycle of type $\vec{t}$.
\end{lemma}
\begin{proof}
Try out all clusters $X_1 \in \mathcal X_{t_1},\, X_2 \in \mathcal X_{t_2},\, X_3 \in \mathcal X_{t_3}$. 
Because each $\mathcal{X}_\ell$ consists of disjoint clusters of size $\Theta(n/2^\ell)$,
there are at most $\Theta(2^{t_1+t_2+t_3})$ such triples.
We can check if $G$ has an induced 4-cycle with two nodes in $X_1$ and one node in each of $X_2$ and $X_3$ in 
\[\tilde O(|X_1| + |X_2| + |X_3|) \le \tilde{O}(n/2^{\min(t_1,t_2,t_3)})\] time by \Cref{lem:cluster-triple}. 
Thus the total runtime is at most 
\[\tilde O(2^{t_1+t_2+t_3} \cdot n / 2^{\min\grp{t_1, t_2, t_3}}) \le \tilde O(n \cdot 2^{t_1 + t_2 + t_3 - \min(t_1, t_2, t_3)})\]
as claimed.
\end{proof}

\begin{lemma}[{$\starup$}HH Types] \label{lem:type-xhh}
Fix $\vec{t}= \pair{t_1, t_2, t_3}$ with $t_2, t_3 \ge L+1$. 
Given orderings between all cluster pairs in $G$,
we can determine if $G$ has an induced 4-cycle of type $\vec{t}$ in $ O(n^3 / 2^{\min(t_2,t_3)})$ time.
\end{lemma}
\begin{proof}
By the assumption from the first paragraph of 
\Cref{subsec:cluster-levels},
we have access to the common neighborhoods in $V_\ell$ from all non-edges $\grp{x,y}$,
for all $\ell\in\set{L+1, \dots, H}$.
By scanning through these common neighborhoods,
we can compute for all nodes $v_2\in V_{t_2}$ and $v_3\in V_{t_3}$,
the collection $\mathcal{Y}(v_2,v_3)\sub \mathcal{X}_{t_1}$ of clusters $X$ in $\mathcal{X}_{t_1}$ such that there exists a node $v_1\in X$ 
so that $(v_1, v_2, v_3)$ is an induced 2-path in $G$.
By similar reasoning, we can compute for all vertices $v_2\in V_{t_2}$ and $v_3\in V_{t_3}$ the collection $\mathcal Z(v_2, v_3) \subseteq \mathcal X_{t_1}$ of clusters $X \in \mathcal X_{t_1}$ such that $X$ has a node $\tilde{v}_1$ such that $(v_2, v_3, \tilde{v}_1)$ is an induced 2-path.

Since any two nodes in a cluster are adjacent, 
$G$ has an induced 4-cycle of type $\vec{t}$ if and only if there exist vertices $v_2 \in V_{t_2}$ and $v_3 \in V_{t_3}$ with $\mathcal Y(v_2, v_3) \cap \mathcal Z(v_2, v_3) \neq \emptyset$. 
Having constructed these sets, we can check if they have empty intersection or not in time linear in the sizes of these sets.
The sum of the sizes of the $\mathcal{Y}(v_2, v_3)$ sets is at most the number of induced 2-paths with middle node in $V_{t_2}$.
By the \emph{Bounded Common Neighborhoods} condition of \Cref{thm:clique-decomp}, each of the at most $n^2$ non-edges in $G$ can be extended to at most $O(n/2^{t_2})$ induced 2-paths with middle node in $V_{t_2}$.
Thus the sum of the sizes of the $\mathcal{Y}(v_2, v_3)$ sets is at most   $O(n^2\cdot n/2^{t_2}) \le O(n^3/2^{t_2})$.
Symmetric reasoning shows that the  sum of the sizes of the $\mathcal{Z}(v_2, v_3)$ sets is at most $O(n^3/2^{t_3})$.
These bounds on the number of induced 2-paths with middle nodes in $V_{t_2}$ and $V_{t_3}$ also upper bound the time needed to construct the $\mathcal{Y}(v_2, v_3)$ and $\mathcal{Z}(v_2, v_3)$ sets in the first place.

Thus the overall runtime of the algorithm is at most
    \[O(n^3/2^{t_2} + n^3/2^{t_3}) \le O(n^3/2^{\min(t_2,t_3)})\]
as claimed. 
\end{proof}

\begin{lemma}[HH{$\starup$} Types] \label{lem:type-hhx}
Fix $\vec{t} = \pair{t_1, t_2, t_3}$ with $t_1, t_2 \ge L+1$. 
Given orderings between all cluster pairs in $G$,
we can determine if $G$ has an induced 4-cycle of type $\vec{t}$ in time $O(n^4 / 2^{t_1 + t_2})$.
\end{lemma}
\begin{proof}
We try out all non-edges $\grp{v_1, v_3}\in V_{t_1}\times V_{t_3}$.
For each such choice of $v_1, v_3$,
we enumerate all common neighbors $v_2 \in N_{t_2}(v_1, v_3)$ and $\tilde{v}_1 \in N_{t_1}(v_1, v_3)$, and test if $(\tilde{v}_1, v_1, v_2, v_3)$ is an induced 4-cycle. 
Any induced 4-cycle in $G$ of type $\vec{t}$ must be of this form,
so this algorithm will find such a cycle if it exists.

Since $t_1, t_2\ge L+1$,
the \emph{Bounded Common Neighborhoods} condition of \Cref{thm:clique-decomp}
ensures that
for each choice of $v_1$ and $v_3$, we try out at most $O(n/2^{t_1})$ choices of $\tilde{v}_1$ and $O(n/2^{t_2})$ choices of $v_2$.
Thus this algorithm takes at most 
    \[O(n^2\cdot (n/2^{t_1}) \cdot (n/2^{t_2})) \le O(n^4/2^{t_1+t_2})\]
time as claimed.
\end{proof}

\begin{lemma}[H{$\starup$}H Types] \label{lem:type-hhx-symmetric}
Fix $\vec{t} = \pair{t_1, t_2, t_3}$ with $t_1, t_3 \ge L+1$. 
Given orderings between all cluster pairs in $G$,
we can determine if $G$ has an induced 4-cycle of type $\vec{t}$ in time $O(n^4 / 2^{t_1 + t_3})$.
\end{lemma}
\begin{proof}
    Follows by symmetric reasoning to the proof of \Cref{lem:type-hhx}.
\end{proof}

We now combine \Cref{lem:type-lll,lem:type-xhh,lem:type-hhx,lem:type-hhx-symmetric} to detect induced 4-cycles that are 3-clustered. 

\begin{lemma}[3-Clustered Detection]
    \label{lem:3-clustered}
Given orderings between all cluster pairs in $G$,
there is an $\tilde{O}(n^{5/2})$ time algorithm that determines if $G$ contains a 
3-clustered induced 4-cycle.
\end{lemma}
\begin{proof}

Try out all $O((\log n)^3)$ possible types $\vec{t} = \pair{t_1,t_2,t_3}\in\set{L,\dots, H}^3$.
For each choice $\vec{t}$, we seek a 3-clustered, induced 4-cycle in $G$ with 
type $\vec{t}$.

Fix $\vec{t} = \pair{t_1,t_2,t_3}$.
Without loss of generality, suppose that $t_2\le t_3$.

Consider the following three cases:

\begin{enumerate}
  \item\emph{If $t_2\ge L+1$:} 
    In this case, we also have $t_3\ge t_2\ge L+1$.
    Thus
    we can apply the algorithm from 
    \cref{lem:type-xhh}
    to detect an induced 4-cycle of type $\vec{t}$ in 
        \[O(n^3/2^{\min(t_2,t_3)}) \le O(n^3/2^L) \le O(n^{5/2})\]
    time.
    
  \item\emph{If $t_1 + t_3 > (3/2)\log n$:} 
    Since $L \leq (1/2)\log n$ and $H \leq \log n$, 
    in this case we have $t_1, t_3 \ge L+1$. 
    Thus, we can apply \cref{lem:type-hhx-symmetric} 
    to detect an induced 4-cycle of type $\vec{t}$ in 
    \[O(n^4/2^{t_1+t_3}) \le O(n^4/2^{(3/2)\log n}) \le O(n^{5/2})\]
    time.
    
  \item\emph{If $t_1 + t_2 + t_3 - \min(t_1, t_2, t_3) \leq (3/2) \log n$:} In this case, the algorithm from \cref{lem:type-lll}
  detects an induced 4-cycle of type $\vec{t}$ in
    \[\tilde{O}(n \cdot 2^{t_1 + t_2 + t_3 - \min(t_1, t_2, t_3)}) \le
    \tilde{O}(n\cdot 2^{(3/2)\log n}) \le \tilde{O}(n^{5/2})\]
    time.
\end{enumerate}

We claim that every type $\vec{t}$ falls into one of the three cases above.
Indeed, if a type does not satisfy case 1 above,
then we have $t_2 = L$.
This then forces $\min(t_1, t_2, t_3)= t_2$,
so 
\[t_1 + t_2 + t_3 - \min(t_1,t_2,t_3) = t_1+t_3.\]
If the above sum is at most $(3/2)(\log n)$, we satisfy case 3.
If instead the above sum is greater than $(3/2)(\log n)$, we satisfy case 2. 
Thus for each of the $\poly(\log n)$ choices of $\vec{t}$ we can check if $G$ has an induced 4-cycle of type $\vec{t}$ in $\tilde{O}(n^{5/2})$ time,
which proves the desired result. 
\end{proof}

We now move on to detecting 4-clustered induced 4-cycles. 

    \subsubsection{Cycles Among Four Clusters}
    
We say 4-clustered, induced 4-cycle $(v_1, v_2, v_3, v_4)$ has type 
$\vec{t} = \pair{t_1, t_2, t_3, t_4}$ if 
$v_i\in V_{t_i}$ for each index $i\in [4]$.
Note that since the cycle is 4-clustered,
 the $t_i$ are all distinct. 
We present various algorithms for detecting 4-clustered, induced 4-cycles with prescribed types.
To organize our casework,
we informally associate each type $\vec{t}$ with a label in  $\set{\t{L}, \t{H}, \starup}^4$, analogous to the labeling in the previous \Cref{subsubsec:3-clustered} subsection. 

The following is the 4-clustered analogue of \Cref{lem:type-lll}.

\begin{lemma}[Type LLLL] \label{lem:type-llll}
Fix a type $\vec{t} = \pair{t_1,t_2,t_3,t_4}$.
Suppose $G$ contains no 3-clustered, induced 4-cycle.
Then given concise orderings between all cluster pairs,
we can determine whether $G$ has an induced 4-cycle of type $\vec{t}$
in $\tilde O(n \cdot 2^{t_1 + t_2 + t_3 + t_4 - \min\grp{t_1, t_2, t_3, t_4}})$  time.
\end{lemma}
\begin{proof}
Try out all clusters $X_1 \in \mathcal X_{t_1},\, X_2 \in \mathcal X_{t_2},\, X_3 \in \mathcal X_{t_3},\, X_4 \in \mathcal X_{t_4}$. 
Because each $\mathcal{X}_\ell$ consists of disjoint clusters of size $\Theta(n/2^\ell)$,
there are at most $O(2^{t_1+t_2+t_3+t_4})$ choices for these clusters.
We can check if $G$ has an induced 4-cycle with one node in each of the $X_{t_i}$ in 
    \[\tilde{O}(|X_1| + |X_2| + |X_3| + |X_4|) \le \tilde{O}(n/2^{\min(t_1,t_2,t_3,t_4)})\]
time by \Cref{lem:cluster-quadruple}.
Thus the total runtime is at most 
    \[\tilde{O}(2^{t_1+t_2+t_3+t_4}\cdot (n/2^{\min(t_1,t_2,t_3,t_4)}) ) \le \tilde O(n \cdot 2^{t_1 + t_2 + t_3 + t_4 - \min\grp{t_1, t_2, t_3, t_4}})\]
as claimed.
\end{proof}

Next, we prove a 4-clustered analogue of \Cref{lem:type-hhx}.

\begin{lemma}[Type H{$\starup$}H{$\starup$}] \label{lem:type-hxhx}
Let $t = \pair{t_1, t_2, t_3, t_4}$ be a type with $t_1, t_3\ge L+1$. 
Given orderings between all cluster pairs,
we can determine if $G$ has an induced 4-cycle
of type $\vec{t}$ in time $O(n^4 / 2^{t_1 + t_3})$. 
\end{lemma}
\begin{proof}
We try out all non-edges $\grp{v_2, v_4}\in V_{t_2}\times V_{t_4}$.
For each such choice of $v_2,v_4$, we enumerate all common neighbors $v_1 \in N_{t_1}(v_2, v_4)$ and $v_3 \in N_{t_3}(v_2, v_4)$
and test if $(v_1,v_2,v_3,v_4)$ forms an induced 4-cycle.
Any induced 4-cycle of type $\vec{t}$ must be of this form,
so the algorithm will find such a cycle if it exists.

Since $t_1, t_3\ge L+1$,
the \emph{Bounded Common Neighborhoods} condition of \Cref{thm:clique-decomp} ensures that for each choice of $v_2,v_4$, we try out at most $O(n/2^{t_1})$ choices of $v_1$ and $O(n/2^{t_3})$ choices of $v_3$.
Thus this algorithm takes at most 
    \[O(n^2\cdot (n/2^{t_1})\cdot (n/2^{t_3})) \le O(n^4/2^{t_1+t_3})\]
time as desired.
\end{proof}

\begin{lemma}[Type {$\starup$}H{$\starup$}H] \label{lem:type-hxhx-symmetric}
Let $t = \pair{t_1, t_2, t_3, t_4}$ be a type with $t_2, t_4\ge L+1$. 
Given orderings between all cluster pairs,
we can determine if $G$ has an induced 4-cycle
of type $\vec{t}$ in time $O(n^4 / 2^{t_2 + t_4})$. 
\end{lemma}
\begin{proof}
    Follows by symmetric reasoning to the proof of \Cref{lem:type-hxhx}.
\end{proof}

Our final helper algorithm does not have an analogue in the 3-clustered case. 

\begin{lemma}[Type L{$\starup$}H{$\starup$}] \label{lem:type-hxlx}
Let $\vec{t} = \pair{t_1, t_2, t_3, t_4}$ be a type with $t_3\ge L+1$. 
Suppose $G$ contains no 3-clustered, induced 4-cycle.
Then given orderings between all cluster pairs, we can determine whether $G$ has an induced 4-cycle of type $\vec{t}$ in \smash{$\tilde O(n^3 / 2^{t_3 - t_1} + n^2 \cdot 2^{t_1} + n \cdot 2^{t_2 + t_4})$} time.
\end{lemma}
\begin{proof}
Our algorithm works in three steps.

\begin{enumerate}
  \item\emph{\emph{Step 1:} Computing $\codeg_{X_1}(v_2, v_4)$ for all clusters $X_1 \in \mathcal X_{t_1}$ and $(v_2, v_4) \in V_{t_2} \times V_{t_4}$:}\newline
  We try out all triples of clusters $X_1 \in \mathcal X_{t_1}, X_2 \in \mathcal X_{t_2}, X_4 \in \mathcal X_{t_4}$. 
  For each triple we apply the
    \[\tilde{O}(|X_1| + |X_2||X_4|) \le \tilde{O}(n/2^{t_1} + n^2/2^{t_2+t_4})\]
time algorithm from \Cref{lem:common-neighbor-size-in-clusters} to 
compute $\codeg_{X_1}(v_2,v_4)$ for all $(v_2,v_4)\in V_{t_2}\times V_{t_4}$.

    Since each $\mathcal{X}_{\ell}$
    consists of disjoint clusters of size $\Theta(n/2^\ell)$,
    we run the above procedure for at most $2^{t_1+t_2+t_4}$ triples of clusters.
    Hence
    this step takes at most 
        \[\tilde{O}\grp{2^{t_1+t_2+t_4}\cdot (n/2^{t_1} + n^2/2^{t_2+t_4})} \le 
        \tilde{O}(n^2\cdot 2^{t_1} + n\cdot 2^{t_2+t_4})\]
    time. 
  
  \item\emph{\emph{Step 2:} Computing $\deg_{X_1}(v_3)$ for all clusters $X_1 \in \mathcal X_{t_1}$ and all nodes $v_3 \in V_{t_3}$:}\newline
  We  compute these degrees by scanning through the neighborhoods of each vertex $v_3\in V_{t_3}$.
  Anytime we find a neighbor of $v_3$ in a cluster $X_1\in\mathcal{X}_{t_1}$,
  we increment a counter corresponding to the pair $(v_3, X_1)$.
  This takes at most $O(n^2)$ time overall,
  because we encounter each edge in the graph at most two times. 
  
  \item\emph{\emph{Step 3:} Detecting induced 4-cycles:}\newline
  We try out all clusters $X_1\in\mathcal{X}_{t_1}$.
  Since each $\mathcal{X}_\ell$ consists of disjoint clusters of size $\Theta(n/2^\ell)$,
  there are at most $2^{t_1}$ clusters $X_1$ we try out.
  For each $X_1$, we
  enumerate all of the non-edges 
  $(v_2,v_4)\in V_{t_2}\times V_{t_4}$.
  There are at most $n^2$ such choices for $v_2$ and $v_4$.
  For each choice of $X_1$, $v_2$, and $v_4$,
  we go over 
   the common neighbors $v_3 \in N_{t_3}(v_2, v_4)$.
  The \emph{Bounded Common Neighborhoods} condition of 
  \Cref{thm:clique-decomp} ensures that we try out at most $O(n/2^{t_3})$ nodes $v_3$ in this step.  
  We check if 
    \[\deg_{X_1}(v_3) < \codeg_{X_1}(v_2, v_4).\]
If this inequality holds, we report an induced 4-cycle in $G$.
If this inequality never holds for any choice $X_1, v_2, v_4, v_3$,
then we report that $G$ has no 4-cycle of the given type. 
Since $G$ has no 3-clustered, induced 4-cycle,
this algorithm has the desired behavior by
  \Cref{obs:neighborhood-size}.

This final step takes 
\[O(2^{t_1} \cdot n^2 \cdot  n/ 2^{t_3}) \le O(n^3 / 2^{t_3 - t_1})\]
time.
\end{enumerate}
Combining the runtimes from \emph{steps 1} to \emph{3}
proves the desired result.
\end{proof}

\begin{lemma}[Type {$\starup$}L{$\starup$}H] \label{lem:type-hxlx-symmetric}
Let $\vec{t} = \pair{t_1, t_2, t_3, t_4}$ be a type with $t_4\ge L+1$. 
Suppose $G$ contains no 3-clustered, induced 4-cycle.
Then given orderings between all cluster pairs, we can determine whether $G$ has an induced 4-cycle of type $\vec{t}$ in \smash{$\tilde O(n^3 / 2^{t_4 - t_2} + n^2 \cdot 2^{t_2} + n \cdot 2^{t_1 + t_3})$} time.
\end{lemma}
\begin{proof}
    Follows by symmetric reasoning to the proof of \Cref{lem:type-hxlx}.
\end{proof}

We now combine \Cref{lem:type-llll,lem:type-hxhx,lem:type-hxlx} to detect induced 4-cycles that are 4-clustered.

\begin{lemma}[4-Clustered] \label{lem:4-clustered}
Suppose $G$ does not contain any 2-clustered or 3-clustered induced 4-cycles.
Then given concise orderings between all cluster pairs,
there is an $\tilde{O}(n^{17/6})$ time algorithm that determines if $G$ has a 4-clustered induced 4-cycle.
\end{lemma}
\begin{proof}
Try out all $O((\log n)^4)$ possible types $\vec{t} = \pair{t_1,t_2,t_3,t_4}\in \set{L, \dots, H}^4$.
For each such $\vec{t}$, we seek a 4-clustered,
induced 4-cycle in $G$ with type $\vec{t}$.

Fix $\vec{t} = \pair{t_1,t_2,t_3,t_4}$.
Without loss of generality, suppose that $t_1 = \min(t_1,t_2,t_3,t_4)$ (since we can cyclically shift vertices in an order $(v_1,v_2,v_3,v_4)$ without changing the underlying 4-cycle) 
and $t_2\le t_4$ (since we can reverse the order of the vertices without changing the underlying 4-cycle). 
Consider the following cases (where we apply cases successively, so that if we ever reach a case, we assume that the conditions in all previous cases are not met):

  \begin{enumerate}
    \item\emph{If $t_1 + t_2 + t_3 + t_4 - \min(t_1, t_2, t_3, t_4) \leq (11/6) \log n$:} 
    
    In this case, we apply \cref{lem:type-llll} to
    detect an induced 4-cycle of type $\vec{t}$ in $\tilde O(n^{17/6})$ time.
    
    \item[2(a).]\emph{If $t_1 + t_3 > (11/6)\log n$:} 

    Since $t_3\le H \le \log n$,
    in this case we have \[t_1 > (11/6)\log n - t_3 \ge (5/6)\log n \ge L+1.\]
    Since $t_1$ is the minimum entry of $\vec{t}$,
    we have $t_3 \ge t_1 \ge L+1$ as well.
    Thus we can apply \cref{lem:type-hxhx} 
    to detect an induced 4-cycle of type $\vec{t}$
     in  $O(n^4/2^{t_1+t_3}) \le O(n^{13/6})$ time.
    
    \item[2(b).]\emph{If $t_2 + t_4 > (11/6) \log n$:} 
    
    In this case, we detect an induced 4-cycle of type $\vec{t}$ in $O(n^{13/6})$
    time by applying \Cref{lem:type-hxhx-symmetric} together with similar reasoning to the proof of case 2(a) above.

    \item[3(a).]\emph{If $t_1 \leq t_3 - (1/6) \log n$:}

    We have $t_3 - t_1 \ge (1/6)\log n$,
    and $t_1\le H - (1/6)\log n \le (5/6)\log n$.
    Since we only reach this case if the condition in case 2(b) is not met,
    we must also have $t_2 + t_4 \leq (11/6) \log n$.

    So in this case we apply 
    \cref{lem:type-hxlx} to detect an induced 4-cycle of type $\vec{t}$ in 
        \[\tilde{O}(n^3 / 2^{t_3 - t_1} + n^2 \cdot 2^{t_1} + n \cdot 2^{t_2 + t_4})
        \le \tilde{O}(n^{17/6})\]
    time.

    \item[3(b).]\emph{If $t_2 \leq t_4 - (1/6)\log n$:} 
    
    In this case, we detect an induced 4-cycle of type $\vec{t}$ in $\tilde{O}(n^{17/16})$ time by applying \Cref{lem:type-hxlx-symmetric} together with similar reasoning to the proof of case 3(a) above.
    
    \item[4(a).]\emph{If $t_1 + t_3 > (7/6) \log n$:} 
    
    Since we only reach this case if case 3(a)'s condition is not met,
    we have 
     $t_3 \leq t_1 + (1/6)\log n$.
     This implies that 
        \[t_1  = (1/2)\cdot [(t_1 + t_3) + (t_1-t_3)] > (1/2)\cdot [(7/6 - 1/6)\log n] = (1/2)(\log n) \ge L\]
    so $t_1\ge L+1$. 
    Since $t_1$ is the minimum entry of $\vec{t}$,
    we also get $t_3\ge t_1 \ge L+1$.
    
    So in this case we apply \cref{lem:type-hxhx} to detect an induced 4-cycle of type $\vec{t}$ in 
        \[O(n^4/2^{t_1+t_3}) \le O(n^{17/6})\]
    time. 
    
    \item[4(b).]\emph{If $t_2 + t_4 > (7/6) \log n$:} 

    In this case we detect an induced 4-cycle of type $\vec{t}$ in $O(n^{17/6})$ time
    by applying \Cref{lem:type-hxhx-symmetric} together with similar reasoning to the proof of case 4(a) above.
  \end{enumerate}

    We claim that every type $\vec{t}$ falls into one of the cases above.
    Indeed, if a type does not 
    satisfy case
    4(b) above,
    then we have    
        \[t_2 + t_4 \leq (7/6) \log n.\]
    If the same type does not satisfy case 1 either, then
    since $t_1 = \min(t_1,t_2,t_3,t_4)$ we have 
    \[t_2 + t_3 + t_4 > (11/6) \log n.\]
    Subtracting the first inequality from the second inequality above yields 
    \[t_3 > (2/3)\log n.\]
    Now if the type does not satisfy case 3(a), we have 
        \[t_1 > t_3 - (1/6)\log n > (1/2)\log n.\]
    But now adding the last two inequalities yields
        \[t_1 + t_3 > (7/6)\log n\]
    which implies the type satisfies case 4(a).

    This implies the cases are exhaustive, and for each of the $\poly(\log n)$
    choices of $\vec{t}$
    we can check if $G$ has an induced 4-cycle of type $\vec{t}$ in
    $\tilde{O}(n^{17/6})$ time, which proves the desired result. 
\end{proof}

\subsection{Final Algorithm}
    \label{subsec:final-alg}

    We can now prove our main result.

    \deterministic*
    \begin{proof}
        Apply \Cref{thm:clique-decomp} with parameters $L = \lfloor (1/2)\log n\rfloor$ and $H = \lfloor \log n\rfloor$.
        This takes 
            \[\tilde{O}(n^2\cdot 2^L + n^3/2^L) \le \tilde{O}(n^{5/2})\]
        time.
        If the algorithm detects an induced 4-cycle, we report it.
        Otherwise,
        \Cref{thm:clique-decomp} returns a decomposition of the graph into clusters
        (cliques satisfying certain technical conditions). 

        We then apply \Cref{lem:2-clustered} to the graph with this decomposition.
        This takes $O(n^2)$ time.
        If the algorithm detects an induced 4-cycle, we report it.
        Otherwise, \Cref{lem:2-clustered} reports that the graph contains no induced 4-cycle with nodes in at most two clusters.
        Moreover, \Cref{lem:2-clustered} verifies that every pair of clusters is ordered,
        and returns consise orderings for each cluster pair witnessing this.

        We then apply \Cref{lem:3-clustered} to the graph with its decomposition and orderings.
        This takes $\tilde{O}(n^{5/2})$  time.
        If the algorithm detects an induced 4-cycle, we report it.
        Otherwise, \Cref{lem:3-clustered} verifies that the graph contains no induced 4-cycle with nodes in at most three clusters.

        Finally, we apply \Cref{lem:4-clustered} to the graph with its decomposition and orderings,
        and the guarantee that there is no induced 4-cycle using nodes from at most three clusters.
        This takes $\tilde{O}(n^{17/6}) \le \tilde{O}(n^{3-1/6})$ time.
        If the algorithm detects an induced 4-cycle,
        we report it.
        Otherwise, \Cref{lem:4-clustered} verifies that the graph contains no induced 4-cycle with nodes in at most four clusters.
        This then implies that the graph has no induced 4-cycles whatsoever, and we can report that no such cycles exist.
    \end{proof}
\section{Conclusion}
In this paper, we presented a combinatorial, deterministic, truly subcubic  algorithm for detecting induced 4-cycles.
Prior to our work,
no truly subcubic-time algorithm for induced 4-cycle detection was known that even met \emph{either} of the conditions of being combinatorial or deterministic individually.
The most natural question in light of our result is whether the complexity of detecting induced 4-cycles can be brought all the way down to an optimal $O(n^2)$ runtime bound.

\begin{problem}
    \label{prob1:gotta-go-fast}
Can induced 4-cycle detection be solved in quadratic time? 
\end{problem}

Although in this paper we focused on runtimes for subgraph detection problems in terms of the number of vertices $n$,
parameterizing by the number of \emph{edges} $m$ in the input graph is also an interesting research direction.
Obtaining  faster algorithms for detecting induced 4-cycles in sparse graphs  could potentially help accelerate some of the subroutines used in our framework (namely the ``high level'' procedures from \Cref{subsec:cluster-levels}),
which may in turn help resolve \Cref{prob1:gotta-go-fast}.

\begin{problem}
    Can induced 4-cycle detection be solved in $m$-edge graphs in $O(m^{4/3})$ time?     
\end{problem}

We note that induced 4-cycle detection requires \smash{$m^{4/3-o(1)}$} time to solve in general, assuming a hypothesis from the field of fine-grained complexity
\cite[Theorem 2.4]{DalirrooyfardVassilevska2022}.
The current fastest algorithm for detecting induced 4-cycles in $m$-edge graphs is randomized and algebraic,
and runs in $\tilde{O}(m^{(4\omega-1)/(2\omega+1)}) \le \tilde{O}(m^{1.48})$ time \cite[Corollary 4.1]{VassilevskaWangWilliamsYu2014}.
Obtaining faster combinatorial and deterministic algorithms for this task on sparse graphs is an interesting problem.

Finally, 
the overall structure of our algorithm 
differs from most other fast  subgraph detection algorithms we are aware of in the literature.
Namely, rather than reducing  the detection problem to a randomized counting procedure as in \cite{VassilevskaWangWilliamsYu2014,BlaserKomarathSreenivasaiah2018} for example,
we  decompose the graph into large cliques we call clusters,
and then employ win/win strategies to either report induced 4-cycles,
or iteratively gain more knowledge of the structure of inter-cluster edges.

Although there is some  sense in which the 
induced 4-cycle is an exceptional pattern $H$ when it comes to identifying large cliques in induced $H$-free graphs (as discussed in \cite[Proposition 1]{GyarfasHubenkoSolymosi2002}),
the Erd\H{o}s-Hajnal conjecture proposes that for \emph{every} pattern graph $H$, 
there exists a corresponding constant $\eps = \eps(H) > 0$
such that every $n$-node graph with no induced copy of $H$ 
has a clique or independent set of size $\Omega(n^\eps)$.
This structure seems qualitatively similar to the guarantees of \Cref{thm:4cycle-free-large-clique}, the starting point of our clique decomposition,
and suggests that similar win/win strategies (based off decompositions into large cliques \emph{and independent sets})
may be possible for induced $H$-detection for patterns $H$ beyond the 4-cycle.
Even if this specific strategy turns out not to be applicable for other subgraph detection problems (because the 4-cycle is such a special pattern),
investigating algorithmic and effective versions of the Erd\H{o}s-Hajnal conjecture (even for small pattern graphs) and its potential connection to other graph algorithms questions seems like a potentially fruitful research direction, in light of our work.

\begin{problem}
    Can clique decompositions
    or proven instances of the Erd\H{o}s-Hajnal conjecture help
    obtain faster combinatorial algorithms for induced $H$ detection for other pattern graphs $H$? 
\end{problem}

\bibliographystyle{alpha}
\bibliography{main}

@article{FloderusKLL15,
  author       = {Peter Floderus and
                  Miroslaw Kowaluk and
                  Andrzej Lingas and
                  Eva{-}Marta Lundell},
  title        = {Detecting and Counting Small Pattern Graphs},
  journal      = {{SIAM} J. Discret. Math.},
  volume       = {29},
  number       = {3},
  pages        = {1322--1339},
  year         = {2015},
  url          = {https://doi.org/10.1137/140978211},
  doi          = {10.1137/140978211},
  bibsource    = {dblp computer science bibliography, https://dblp.org}
}

@article{williams2009finding,
  title={Finding paths of length {$k$} in {$O^*(2^k)$} time},
  author={Williams, Ryan},
  journal={Information Processing Letters},
  volume={109},
  number={6},
  pages={315--318},
  year={2009},
  publisher={Elsevier}
}

@inproceedings{Saranurak21,
  author       = {Thatchaphol Saranurak},
  editor       = {Hung Viet Le and
                  Valerie King},
  title        = {A Simple Deterministic Algorithm for Edge Connectivity},
  booktitle    = {4th Symposium on Simplicity in Algorithms, {SOSA} 2021, Virtual Conference,
                  January 11-12, 2021},
  pages        = {80--85},
  publisher    = {{SIAM}},
  year         = {2021},
  url          = {https://doi.org/10.1137/1.9781611976496.9},
  doi          = {10.1137/1.9781611976496.9},
  bibsource    = {dblp computer science bibliography, https://dblp.org}
}

@inproceedings{chang2019distributed,
  title={Distributed triangle detection via expander decomposition},
  author={Chang, Yi-Jun and Pettie, Seth and Zhang, Hengjie},
  booktitle={Proceedings of the Thirtieth Annual ACM-SIAM Symposium on Discrete Algorithms},
  pages={821--840},
  year={2019},
  organization={SIAM}
}

@inproceedings{bansal2009regularity,
  title={Regularity lemmas and combinatorial algorithms},
  author={Bansal, Nikhil and Williams, Ryan},
  booktitle={2009 50th Annual IEEE Symposium on Foundations of Computer Science},
  pages={745--754},
  year={2009},
  organization={IEEE}
}

@article{corneil1985linear,
  title={A linear recognition algorithm for cographs},
  author={Corneil, Derek G. and Perl, Yehoshua and Stewart, Lorna K},
  journal={SIAM Journal on Computing},
  volume={14},
  number={4},
  pages={926--934},
  year={1985},
  publisher={SIAM}
}

@article{GyarfasHubenkoSolymosi2002,
  title = {Large Cliques in {$C_4$}-Free Graphs},
  volume = {22},
  ISSN = {1439-6912},
  url = {http://dx.doi.org/10.1007/s004930200012},
  DOI = {10.1007/s004930200012},
  number = {2},
  journal = {Combinatorica},
  publisher = {Springer Science and Business Media LLC},
  author = {Gy\'{a}rf\'{a}s,  Andr\'{a}s and Hubenko,  Alice and Solymosi,  J\'{o}zsef},
  year = {2002},
  month = apr,
  pages = {269--274}
}

@misc{DalirrooyfardVassilevska2022-arXiv,
Author = {Mina Dalirrooyfard and Virginia Vassilevska Williams},
Title = {Induced Cycles and Paths Are Harder Than You Think},
Year = {2022},
Eprint = {arXiv:2209.01873},
}

@book{AroraBarak2009, place={Cambridge}, title={Computational Complexity: A Modern Approach}, publisher={Cambridge University Press}, author={Arora, Sanjeev and Barak, Boaz}, year={2009}}

@inproceedings{DalirrooyfardVassilevska2022,
  title = {Induced Cycles and Paths Are Harder Than You Think},
  url = {http://dx.doi.org/10.1109/FOCS54457.2022.00057},
  DOI = {10.1109/focs54457.2022.00057},
  booktitle = {2022 IEEE 63rd Annual Symposium on Foundations of Computer Science (FOCS)},
  publisher = {IEEE},
  author = {Dalirrooyfard,  Mina and Williams,  Virginia Vassilevska},
  year = {2022},
  month = oct,
  pages = {531–542}
}

@article{BucicNguyenScottSeymour2024,
  title = {Induced Subgraph Density. I. A loglog Step Towards Erd{\H{o}}s–Hajnal},
  volume = {2024},
  ISSN = {1687-0247},
  url = {http://dx.doi.org/10.1093/imrn/rnae065},
  DOI = {10.1093/imrn/rnae065},
  number = {12},
  journal = {International Mathematics Research Notices},
  publisher = {Oxford University Press (OUP)},
  author = {Buci{\'{c}},  Matija and Nguyen,  Tung and Scott,  Alex and Seymour,  Paul},
  year = {2024},
  month = may,
  pages = {9991–10004}
}

@inproceedings{VassilevskaWangWilliamsYu2014,
  title = {Finding Four-Node Subgraphs in Triangle Time},
  url = {http://dx.doi.org/10.1137/1.9781611973730.111},
  DOI = {10.1137/1.9781611973730.111},
  booktitle = {Proceedings of the Twenty-Sixth Annual ACM-SIAM Symposium on Discrete Algorithms},
  publisher = {Society for Industrial and Applied Mathematics},
  author = {Williams,  Virginia Vassilevska and Wang,  Joshua R. and Williams,  Ryan and Yu,  Huacheng},
  year = {2014},
  month = dec 
}

@article{VassilevskaWilliams2018,
  title = {Subcubic Equivalences Between Path,  Matrix,  and Triangle Problems},
  volume = {65},
  ISSN = {1557-735X},
  url = {http://dx.doi.org/10.1145/3186893},
  DOI = {10.1145/3186893},
  number = {5},
  journal = {Journal of the ACM},
  publisher = {Association for Computing Machinery (ACM)},
  author = {Williams,  Virginia Vassilevska and Williams,  R. Ryan},
  year = {2018},
  month = aug,
  pages = {1–38}
}

@inproceedings{DahlgaardKnudsenStockel2017,
  series = {STOC ’17},
  title = {Finding even cycles faster via capped k-walks},
  url = {http://dx.doi.org/10.1145/3055399.3055459},
  DOI = {10.1145/3055399.3055459},
  booktitle = {Proceedings of the 49th Annual ACM SIGACT Symposium on Theory of Computing},
  publisher = {ACM},
  author = {Dahlgaard,  S{\o}ren and Knudsen,  Mathias B{\ae}k Tejs and St\"{o}ckel,  Morten},
  year = {2017},
  month = jun,
  pages = {112–120},
  collection = {STOC ’17}
}

@inproceedings{AbboudFischerKelleyLovettMeka2024,
  series = {STOC ’24},
  title = {New Graph Decompositions and Combinatorial Boolean Matrix Multiplication Algorithms},
  url = {http://dx.doi.org/10.1145/3618260.3649696},
  DOI = {10.1145/3618260.3649696},
  booktitle = {Proceedings of the 56th Annual ACM Symposium on Theory of Computing},
  publisher = {ACM},
  author = {Abboud,  Amir and Fischer,  Nick and Kelley,  Zander and Lovett,  Shachar and Meka,  Raghu},
  year = {2024},
  month = jun,
  pages = {935–943},
  collection = {STOC ’24}
}

@inproceedings{DalirrooyfardVassilevskaVuong2019,
  series = {STOC ’19},
  title = {Graph pattern detection: hardness for all induced patterns and faster non-induced cycles},
  url = {http://dx.doi.org/10.1145/3313276.3316329},
  DOI = {10.1145/3313276.3316329},
  booktitle = {Proceedings of the 51st Annual ACM SIGACT Symposium on Theory of Computing},
  publisher = {ACM},
  author = {Dalirrooyfard,  Mina and Vuong,  Thuy Duong and Williams,  Virginia Vassilevska},
  year = {2019},
  month = jun,
  pages = {1167–1178},
  collection = {STOC ’19}
}

@InProceedings{BlaserKomarathSreenivasaiah2018,
  author =	{Bl\"{a}ser, Markus and Komarath, Balagopal and Sreenivasaiah, Karteek},
  title =	{{Graph Pattern Polynomials}},
  booktitle =	{38th IARCS Annual Conference on Foundations of Software Technology and Theoretical Computer Science (FSTTCS 2018)},
  pages =	{18:1--18:13},
  series =	{Leibniz International Proceedings in Informatics (LIPIcs)},
  ISBN =	{978-3-95977-093-4},
  ISSN =	{1868-8969},
  year =	{2018},
  volume =	{122},
  publisher =	{Schloss Dagstuhl -- Leibniz-Zentrum f{\"u}r Informatik},
  URL =		{https://drops.dagstuhl.de/entities/document/10.4230/LIPIcs.FSTTCS.2018.18},
  URN =		{urn:nbn:de:0030-drops-99172},
  doi =		{10.4230/LIPIcs.FSTTCS.2018.18},
}

@book{LiLin2022,
  title = {Elementary Methods of Graph Ramsey Theory},
  ISBN = {9783031127625},
  ISSN = {2196-968X},
  url = {http://dx.doi.org/10.1007/978-3-031-12762-5},
  DOI = {10.1007/978-3-031-12762-5},
  journal = {Applied Mathematical Sciences},
  publisher = {Springer International Publishing},
  author = {Li,  Yusheng and Lin,  Qizhong},
  year = {2022}
}

@book{Szemeredi1975,
  title = {Regular partitions of graphs.},
  author = {Szemer{\'e}di, Endre},
  year = {1975},
  publisher = {Stanford University}
}

@article{KannanVempalaVetta2004,
  title = {On clusterings: Good,  bad and spectral},
  volume = {51},
  ISSN = {1557-735X},
  url = {http://dx.doi.org/10.1145/990308.990313},
  DOI = {10.1145/990308.990313},
  number = {3},
  journal = {Journal of the ACM},
  publisher = {Association for Computing Machinery (ACM)},
  author = {Kannan,  Ravi and Vempala,  Santosh and Vetta,  Adrian},
  year = {2004},
  month = may,
  pages = {497–515}
}

@inbook{SaranurakWang2019,
  title = {Expander Decomposition and Pruning: Faster,  Stronger,  and Simpler},
  ISBN = {9781611975482},
  url = {http://dx.doi.org/10.1137/1.9781611975482.162},
  DOI = {10.1137/1.9781611975482.162},
  booktitle = {Proceedings of the Thirtieth Annual ACM-SIAM Symposium on Discrete Algorithms},
  publisher = {Society for Industrial and Applied Mathematics},
  author = {Saranurak,  Thatchaphol and Wang,  Di},
  year = {2019},
  month = jan,
  pages = {2616–2635}
}

@inproceedings{SpielmanTeng2004,
  series = {STOC04},
  title = {Nearly-linear time algorithms for graph partitioning,  graph sparsification,  and solving linear systems},
  url = {http://dx.doi.org/10.1145/1007352.1007372},
  DOI = {10.1145/1007352.1007372},
  booktitle = {Proceedings of the thirty-sixth annual ACM symposium on Theory of computing},
  publisher = {ACM},
  author = {Spielman,  Daniel A. and Teng,  Shang-Hua},
  year = {2004},
  month = jun,
  pages = {81–90},
  collection = {STOC04}
}

@article{FriezeKannan1999,
  title = {Quick Approximation to Matrices and Applications},
  volume = {19},
  ISSN = {1439-6912},
  url = {http://dx.doi.org/10.1007/s004930050052},
  DOI = {10.1007/s004930050052},
  number = {2},
  journal = {Combinatorica},
  publisher = {Springer Science and Business Media LLC},
  author = {Frieze,  Alan and Kannan,  Ravi},
  year = {1999},
  month = feb,
  pages = {175–220}
}

@article{Brown1966,
  title = {On Graphs that do not Contain a Thomsen Graph},
  volume = {9},
  ISSN = {1496-4287},
  url = {http://dx.doi.org/10.4153/CMB-1966-036-2},
  DOI = {10.4153/cmb-1966-036-2},
  number = {3},
  journal = {Canadian Mathematical Bulletin},
  publisher = {Canadian Mathematical Society},
  author = {Brown,  W. G.},
  year = {1966},
  month = aug,
  pages = {281–285}
}

@article{ErdosRenyiSos1966,
  title = {On a Problem of Graph Theory},
  volume = {1},
  journal = {Studia Scientiarum Mathematicarum Hungarica},
  author = {Erd\"{o}s,  P. and R\'enyi, A. and S\'os, V. T.},
  year = {1966},
  pages = {215–235}
}

@inproceedings{AbboudBackursVassilevska2015,
  title = {If the Current Clique Algorithms are Optimal,  So is Valiant’s Parser},
  url = {http://dx.doi.org/10.1109/FOCS.2015.16},
  DOI = {10.1109/focs.2015.16},
  booktitle = {2015 IEEE 56th Annual Symposium on Foundations of Computer Science},
  publisher = {IEEE},
  author = {Abboud,  Amir and Backurs,  Arturs and Williams,  Virginia Vassilevska},
  year = {2015},
  month = oct,
  pages = {98–117}
}

@article{AlonYusterZwick1997,
  title = {Finding and counting given length cycles},
  volume = {17},
  ISSN = {1432-0541},
  url = {http://dx.doi.org/10.1007/BF02523189},
  DOI = {10.1007/bf02523189},
  number = {3},
  journal = {Algorithmica},
  publisher = {Springer Science and Business Media LLC},
  author = {Alon,  N. and Yuster,  R. and Zwick,  U.},
  year = {1997},
  month = mar,
  pages = {209–223}
}

@book{deBergCheongvanKreveldOvermars2008,
  title = {Computational Geometry: Algorithms and Applications},
  ISBN = {9783540779742},
  url = {http://dx.doi.org/10.1007/978-3-540-77974-2},
  DOI = {10.1007/978-3-540-77974-2},
  publisher = {Springer Berlin Heidelberg},
  author = {de Berg,  Mark and Cheong,  Otfried and van Kreveld,  Marc and Overmars,  Mark},
  year = {2008}
}

@article{Spencer1977,
  title = {Asymptotic lower bounds for Ramsey functions},
  volume = {20},
  ISSN = {0012-365X},
  url = {http://dx.doi.org/10.1016/0012-365X(77)90044-9},
  DOI = {10.1016/0012-365x(77)90044-9},
  journal = {Discrete Mathematics},
  publisher = {Elsevier BV},
  author = {Spencer,  Joel},
  year = {1977},
  pages = {69–76}
}

@inproceedings{ErdosHajnal1977,
  title = {On spanned subgraphs of graphs},
  author = {Erd\H{o}s, P. and Hajnal, A.},
  booktitle = {Contributions to graph theory and its applications (Internat. Colloq., Oberhof, 1977) (German)},
  publisher = {Tech. Hochschule Ilmenau},
  year = {1977},
  pages = {80–96}
}

@article{KloksKratschMuller2000,
  title = {Finding and counting small induced subgraphs efficiently},
  volume = {74},
  ISSN = {0020-0190},
  url = {http://dx.doi.org/10.1016/S0020-0190(00)00047-8},
  DOI = {10.1016/s0020-0190(00)00047-8},
  number = {3–4},
  journal = {Information Processing Letters},
  publisher = {Elsevier BV},
  author = {Kloks,  Ton and Kratsch,  Dieter and M\"{u}ller,  Haiko},
  year = {2000},
  month = may,
  pages = {115–121}
}

@article{KowalukLingasLundell2013,
  title = {Counting and Detecting Small Subgraphs via Equations},
  volume = {27},
  ISSN = {1095-7146},
  url = {http://dx.doi.org/10.1137/110859798},
  DOI = {10.1137/110859798},
  number = {2},
  journal = {SIAM Journal on Discrete Mathematics},
  publisher = {Society for Industrial & Applied Mathematics (SIAM)},
  author = {Kowaluk,  Mirosław and Lingas,  Andrzej and Lundell,  Eva-Marta},
  year = {2013},
  month = jan,
  pages = {892–909}
}

@article{JerrumMeeks2015,
  title = {The parameterised complexity of counting connected subgraphs and graph motifs},
  volume = {81},
  ISSN = {0022-0000},
  url = {http://dx.doi.org/10.1016/j.jcss.2014.11.015},
  DOI = {10.1016/j.jcss.2014.11.015},
  number = {4},
  journal = {Journal of Computer and System Sciences},
  publisher = {Elsevier BV},
  author = {Jerrum,  Mark and Meeks,  Kitty},
  year = {2015},
  month = jun,
  pages = {702–716}
}

@article{JerrumMeeks2015b,
  title = {Some Hard Families of Parameterized Counting Problems},
  volume = {7},
  ISSN = {1942-3462},
  url = {http://dx.doi.org/10.1145/2786017},
  DOI = {10.1145/2786017},
  number = {3},
  journal = {ACM Transactions on Computation Theory},
  publisher = {Association for Computing Machinery (ACM)},
  author = {Jerrum,  Mark and Meeks,  Kitty},
  year = {2015},
  month = jul,
  pages = {1–18}
}

@article{JerrumMeeks2016,
  title = {The parameterised complexity of counting even and odd induced subgraphs},
  volume = {37},
  ISSN = {1439-6912},
  url = {http://dx.doi.org/10.1007/s00493-016-3338-5},
  DOI = {10.1007/s00493-016-3338-5},
  number = {5},
  journal = {Combinatorica},
  publisher = {Springer Science and Business Media LLC},
  author = {Jerrum,  Mark and Meeks,  Kitty},
  year = {2016},
  month = oct,
  pages = {965–990}
}

@article{Meeks2016,
  title = {The challenges of unbounded treewidth in parameterised subgraph counting problems},
  volume = {198},
  ISSN = {0166-218X},
  url = {http://dx.doi.org/10.1016/j.dam.2015.06.019},
  DOI = {10.1016/j.dam.2015.06.019},
  journal = {Discrete Applied Mathematics},
  publisher = {Elsevier BV},
  author = {Meeks,  Kitty},
  year = {2016},
  month = jan,
  pages = {170–194}
}

@inproceedings{CurticapeanDellMarx2017,
  series = {STOC ’17},
  title = {Homomorphisms are a good basis for counting small subgraphs},
  url = {http://dx.doi.org/10.1145/3055399.3055502},
  DOI = {10.1145/3055399.3055502},
  booktitle = {Proceedings of the 49th Annual ACM SIGACT Symposium on Theory of Computing},
  publisher = {ACM},
  author = {Curticapean,  Radu and Dell,  Holger and Marx,  Dániel},
  year = {2017},
  month = jun,
  pages = {210–223},
  collection = {STOC ’17}
}

@article{RothSchmitt2020,
  title = {Counting Induced Subgraphs: A Topological Approach to {$\#W[1]$}-hardness},
  volume = {82},
  ISSN = {1432-0541},
  url = {http://dx.doi.org/10.1007/s00453-020-00676-9},
  DOI = {10.1007/s00453-020-00676-9},
  number = {8},
  journal = {Algorithmica},
  publisher = {Springer Science and Business Media LLC},
  author = {Roth,  Marc and Schmitt,  Johannes},
  year = {2020},
  month = jan,
  pages = {2267–2291}
}

@article{DorflerRothSchmittWellnitz2021,
  title = {Counting Induced Subgraphs: An Algebraic Approach to {$\#W[1]$}-Hardness},
  volume = {84},
  ISSN = {1432-0541},
  url = {http://dx.doi.org/10.1007/s00453-021-00894-9},
  DOI = {10.1007/s00453-021-00894-9},
  number = {2},
  journal = {Algorithmica},
  publisher = {Springer Science and Business Media LLC},
  author = {D\"{o}rfler,  Julian and Roth,  Marc and Schmitt,  Johannes and Wellnitz,  Philip},
  year = {2021},
  month = dec,
  pages = {379–404}
}

@inproceedings{RothSchmittWellnitz2020,
  title = {Counting Small Induced Subgraphs Satisfying Monotone Properties},
  url = {http://dx.doi.org/10.1109/FOCS46700.2020.00128},
  DOI = {10.1109/focs46700.2020.00128},
  booktitle = {2020 IEEE 61st Annual Symposium on Foundations of Computer Science (FOCS)},
  publisher = {IEEE},
  author = {Roth,  Marc and Schmitt,  Johannes and Wellnitz,  Philip},
  year = {2020},
  month = nov,
  pages = {1356–1367}
}

@inproceedings{FockeRoth2022,
  series = {STOC ’22},
  title = {Counting small induced subgraphs with hereditary properties},
  url = {http://dx.doi.org/10.1145/3519935.3520008},
  DOI = {10.1145/3519935.3520008},
  booktitle = {Proceedings of the 54th Annual ACM SIGACT Symposium on Theory of Computing},
  publisher = {ACM},
  author = {Focke,  Jacob and Roth,  Marc},
  year = {2022},
  month = jun,
  pages = {1543–1551},
  collection = {STOC ’22}
}

@inproceedings{DoringMarxWellnitz2024,
  series = {STOC ’24},
  title = {Counting Small Induced Subgraphs with Edge-Monotone Properties},
  url = {http://dx.doi.org/10.1145/3618260.3649644},
  DOI = {10.1145/3618260.3649644},
  booktitle = {Proceedings of the 56th Annual ACM Symposium on Theory of Computing},
  publisher = {ACM},
  author = {D\"{o}ring,  Simon and Marx,  Dániel and Wellnitz,  Philip},
  year = {2024},
  month = jun,
  pages = {1517–1525},
  collection = {STOC ’24}
}

@inbook{CurticapeanNeuen2025,
  title = {Counting Small Induced Subgraphs: Hardness via Fourier Analysis},
  ISBN = {9781611978322},
  url = {http://dx.doi.org/10.1137/1.9781611978322.122},
  DOI = {10.1137/1.9781611978322.122},
  booktitle = {Proceedings of the 2025 Annual ACM-SIAM Symposium on Discrete Algorithms (SODA)},
  publisher = {Society for Industrial and Applied Mathematics},
  author = {Curticapean,  Radu and Neuen,  Daniel},
  year = {2025},
  month = jan,
  pages = {3677–3695}
}

@inbook{DoringMarxWellnitz2025,
  title = {From Graph Properties to Graph Parameters: Tight Bounds for Counting on Small Subgraphs},
  ISBN = {9781611978322},
  url = {http://dx.doi.org/10.1137/1.9781611978322.121},
  DOI = {10.1137/1.9781611978322.121},
  booktitle = {Proceedings of the 2025 Annual ACM-SIAM Symposium on Discrete Algorithms (SODA)},
  publisher = {Society for Industrial and Applied Mathematics},
  author = {D\"{o}ring,  Simon and Marx,  Dániel and Wellnitz,  Philip},
  year = {2025},
  month = jan,
  pages = {3637–3676}
}

@article{KhotRaman2002,
  title = {Parameterized complexity of finding subgraphs with hereditary properties},
  volume = {289},
  ISSN = {0304-3975},
  url = {http://dx.doi.org/10.1016/S0304-3975(01)00414-5},
  DOI = {10.1016/s0304-3975(01)00414-5},
  number = {2},
  journal = {Theoretical Computer Science},
  publisher = {Elsevier BV},
  author = {Khot,  Subhash and Raman,  Venkatesh},
  year = {2002},
  month = oct,
  pages = {997–1008}
}

@inbook{Eppstein2021,
  title = {Parameterized Complexity of Finding Subgraphs with Hereditary Properties on Hereditary Graph Classes},
  ISBN = {9783030865931},
  ISSN = {1611-3349},
  url = {http://dx.doi.org/10.1007/978-3-030-86593-1_15},
  DOI = {10.1007/978-3-030-86593-1_15},
  booktitle = {Fundamentals of Computation Theory},
  publisher = {Springer International Publishing},
  author = {Eppstein,  David and Gupta,  Siddharth and Havvaei,  Elham},
  year = {2021},
  pages = {217–229}
}

@inbook{ChenThurleyWeyer2008,
  title = {Understanding the Complexity of Induced Subgraph Isomorphisms},
  ISBN = {9783540705758},
  ISSN = {1611-3349},
  url = {http://dx.doi.org/10.1007/978-3-540-70575-8_48},
  DOI = {10.1007/978-3-540-70575-8_48},
  booktitle = {Automata,  Languages and Programming},
  publisher = {Springer Berlin Heidelberg},
  author = {Chen,  Yijia and Thurley,  Marc and Weyer,  Mark},
  year = {2008},
  pages = {587–596}
}

@inproceedings{MarxPilipczuk2014,
  doi = {10.4230/LIPICS.STACS.2014.542},
  url = {https://drops.dagstuhl.de/entities/document/10.4230/LIPIcs.STACS.2014.542},
  author = {Marx,  Dániel and Pilipczuk,  Michal},
  language = {en},
  title = {Everything you always wanted to know about the parameterized complexity of Subgraph Isomorphism (but were afraid to ask)},
  booktitle =	{31st International Symposium on Theoretical Aspects of Computer Science (STACS 2014)},
  publisher = {Schloss Dagstuhl – Leibniz-Zentrum f\"{u}r Informatik},
  year = {2014},
  copyright = {Creative Commons Attribution 3.0 Unported license}
}

@article{nevsetvril1985complexity,
  title={On the complexity of the subgraph problem},
  author={Ne{\v{s}}et{\v{r}}il, Jaroslav and Poljak, Svatopluk},
  journal={Commentationes Mathematicae Universitatis Carolinae},
  volume={26},
  number={2},
  pages={415--419},
  year={1985},
  publisher={Charles University in Prague, Faculty of Mathematics and Physics}
}

@article{eisenbrand2004complexity,
  title={On the complexity of fixed parameter clique and dominating set},
  author={Eisenbrand, Friedrich and Grandoni, Fabrizio},
  journal={Theoretical Computer Science},
  volume={326},
  number={1-3},
  pages={57--67},
  year={2004},
  publisher={Elsevier}
}

@article{Olariu1988,
  title = {Paw-free graphs},
  volume = {28},
  ISSN = {0020-0190},
  url = {http://dx.doi.org/10.1016/0020-0190(88)90143-3},
  DOI = {10.1016/0020-0190(88)90143-3},
  number = {1},
  journal = {Information Processing Letters},
  publisher = {Elsevier BV},
  author = {Olariu,  Stephan},
  year = {1988},
  month = may,
  pages = {53–54}
}

@article{floderus2015induced,
  title={Induced subgraph isomorphism: Are some patterns substantially easier than others?},
  author={Floderus, Peter and Kowaluk, Miros{\l}aw and Lingas, Andrzej and Lundell, Eva-Marta},
  journal={Theoretical Computer Science},
  volume={605},
  pages={119--128},
  year={2015},
  publisher={Elsevier}
}

@inproceedings{yuster2004detecting,
  title={Detecting short directed cycles using rectangular matrix multiplication and dynamic programming.},
  author={Yuster, Raphael and Zwick, Uri},
  booktitle={SODA},
  volume={4},
  pages={254--260},
  year={2004}
}

@article{Eschen2011,
  title = {{On graphs without a {$C_4$} or a diamond}},
  volume = {159},
  ISSN = {0166-218X},
  url = {http://dx.doi.org/10.1016/j.dam.2010.04.015},
  DOI = {10.1016/j.dam.2010.04.015},
  number = {7},
  journal = {Discrete Applied Mathematics},
  publisher = {Elsevier BV},
  author = {Eschen,  Elaine M. and Hoàng,  Chính T. and Spinrad,  Jeremy P. and Sritharan,  R.},
  year = {2011},
  month = apr,
  pages = {581–587}
}

@article{manurangsi2020strongish,
  title={The strongish planted clique hypothesis and its consequences},
  author={Manurangsi, Pasin and Rubinstein, Aviad and Schramm, Tselil},
  journal={arXiv preprint arXiv:2011.05555},
  year={2020}
}

@inproceedings{DudekG19,
  author       = {Bartlomiej Dudek and
                  Pawel Gawrychowski},
  editor       = {Moses Charikar and
                  Edith Cohen},
  title        = {Computing quartet distance is equivalent to counting 4-cycles},
  booktitle    = {Proceedings of the 51st Annual {ACM} {SIGACT} Symposium on Theory
                  of Computing, {STOC} 2019, Phoenix, AZ, USA, June 23-26, 2019},
  pages        = {733--743},
  publisher    = {{ACM}},
  year         = {2019},
  url          = {https://doi.org/10.1145/3313276.3316390},
  doi          = {10.1145/3313276.3316390},
  bibsource    = {dblp computer science bibliography, https://dblp.org}
}

@inproceedings{DudekG20,
  author       = {Bartlomiej Dudek and
                  Pawel Gawrychowski},
  editor       = {Yixin Cao and
                  Siu{-}Wing Cheng and
                  Minming Li},
  title        = {Counting 4-Patterns in Permutations Is Equivalent to Counting 4-Cycles
                  in Graphs},
  booktitle    = {31st International Symposium on Algorithms and Computation, {ISAAC}
                  2020, December 14-18, 2020, Hong Kong, China (Virtual Conference)},
  series       = {LIPIcs},
  volume       = {181},
  pages        = {23:1--23:18},
  publisher    = {Schloss Dagstuhl - Leibniz-Zentrum f{\"{u}}r Informatik},
  year         = {2020},
  url          = {https://doi.org/10.4230/LIPIcs.ISAAC.2020.23},
  doi          = {10.4230/LIPICS.ISAAC.2020.23},
  bibsource    = {dblp computer science bibliography, https://dblp.org}
}

@inproceedings{AbboudBKZ22,
  author       = {Amir Abboud and
                  Karl Bringmann and
                  Seri Khoury and
                  Or Zamir},
  editor       = {Stefano Leonardi and
                  Anupam Gupta},
  title        = {Hardness of approximation in p via short cycle removal: cycle detection,
                  distance oracles, and beyond},
  booktitle    = {{STOC} '22: 54th Annual {ACM} {SIGACT} Symposium on Theory of Computing,
                  Rome, Italy, June 20 - 24, 2022},
  pages        = {1487--1500},
  publisher    = {{ACM}},
  year         = {2022},
  url          = {https://doi.org/10.1145/3519935.3520066},
  doi          = {10.1145/3519935.3520066},
  bibsource    = {dblp computer science bibliography, https://dblp.org}
}

@inproceedings{AbboudBF23,
  author       = {Amir Abboud and
                  Karl Bringmann and
                  Nick Fischer},
  editor       = {Barna Saha and
                  Rocco A. Servedio},
  title        = {Stronger 3-SUM Lower Bounds for Approximate Distance Oracles via Additive
                  Combinatorics},
  booktitle    = {Proceedings of the 55th Annual {ACM} Symposium on Theory of Computing,
                  {STOC} 2023, Orlando, FL, USA, June 20-23, 2023},
  pages        = {391--404},
  publisher    = {{ACM}},
  year         = {2023},
  url          = {https://doi.org/10.1145/3564246.3585240},
  doi          = {10.1145/3564246.3585240},
  bibsource    = {dblp computer science bibliography, https://dblp.org}
}

@inproceedings{JinX23,
  author       = {Ce Jin and
                  Yinzhan Xu},
  editor       = {Barna Saha and
                  Rocco A. Servedio},
  title        = {Removing Additive Structure in 3SUM-Based Reductions},
  booktitle    = {Proceedings of the 55th Annual {ACM} Symposium on Theory of Computing,
                  {STOC} 2023, Orlando, FL, USA, June 20-23, 2023},
  pages        = {405--418},
  publisher    = {{ACM}},
  year         = {2023},
  url          = {https://doi.org/10.1145/3564246.3585157},
  doi          = {10.1145/3564246.3585157},
  bibsource    = {dblp computer science bibliography, https://dblp.org}
}

@inproceedings{ChanX24,
  author       = {Timothy M. Chan and
                  Yinzhan Xu},
  editor       = {Merav Parter and
                  Seth Pettie},
  title        = {Simpler Reductions from Exact Triangle},
  booktitle    = {2024 Symposium on Simplicity in Algorithms, {SOSA} 2024, Alexandria,
                  VA, USA, January 8-10, 2024},
  pages        = {28--38},
  publisher    = {{SIAM}},
  year         = {2024},
  url          = {https://doi.org/10.1137/1.9781611977936.4},
  doi          = {10.1137/1.9781611977936.4},
  bibsource    = {dblp computer science bibliography, https://dblp.org}
}

\end{document}